\documentclass[a4j,10pt,showkeys]{article}

\usepackage{amsmath,amsthm,amssymb}
\usepackage{geometry}
\usepackage{hyperref}
\usepackage[capitalize]{cleveref}
\usepackage{cancel}
\usepackage{mathtools}
\usepackage{bm}
\usepackage{bbm}
\usepackage{color}
\usepackage{mathrsfs}
\usepackage[title]{appendix}
\usepackage{comment}
\usepackage{caption}
\usepackage{subcaption}
\usepackage{enumerate}

\newtheorem{theo}{Theorem}[section]
\newtheorem{lemm}[theo]{Lemma}
\newtheorem{prop}[theo]{Proposition}
\newtheorem{cor}[theo]{Corollary}
\theoremstyle{definition}
\newtheorem{defi}[theo]{Definition}
\newtheorem{rem}[theo]{Remark}
\newtheorem{assum}{Assumption}

\newcommand{\bE}{\mathbb{E}}
\newcommand{\bF}{\mathbb{F}}

\newcommand{\bN}{\mathbb{N}}

\newcommand{\bP}{\mathbb{P}}

\newcommand{\bR}{\mathbb{R}}

\newcommand{\cA}{\mathcal{A}}

\newcommand{\cF}{\mathcal{F}}
\newcommand{\cG}{\mathcal{G}}
\newcommand{\cH}{\mathcal{H}}

\newcommand{\cX}{\mathcal{X}}
\newcommand{\cY}{\mathcal{Y}}

\newcommand{\I}{\mathbbm{1}}

\DeclareMathOperator*{\argmax}{arg\,max}

\newcommand{\ep}{\varepsilon}
\newcommand{\rd}{\,\mathrm{d}}
\newcommand{\relmiddle}[1]{\mathrel{}\middle#1\mathrel{}}
\newcommand{\1}{\mbox{\rm{1}}\hspace{-0.25em}\mbox{\rm{l}}}
\newcommand{\loc}{\mathrm{loc}}
\newcommand{\pre}{\mathrm{pre}}
\newcommand{\expo}{\mathrm{exp}}
\newcommand{\nav}{\mathrm{naive}}
\newcommand{\my}{\mathrm{myopic}}
\newcommand{\so}{\mathrm{soph}}

\providecommand{\keywords}[1]{\textbf{Keywords:} #1}
\makeatletter

\@addtoreset{equation}{section}
\makeatother
\def\widebar{\accentset{{\cc@style\underline{\mskip10mu}}}}
\numberwithin{equation}{section}
\allowdisplaybreaks

\title{Periodic portfolio selection with quasi-hyperbolic discounting}

\author{
Yushi Hamaguchi\footnote{Department of Mathematics, Kyoto University, Kyoto 606-8502, Japan (E-mail: hamaguchi@math.kyoto-u.ac.jp.). This author was supported by JSPS KAKENHI Grant Number 22K13958.}
~~ and ~~
Alex S.L. Tse\footnote{Corresponding author. Department of Mathematics, University College London, UK (E-mail: alex.tse@ucl.ac.uk). This author acknowledges the grants offered by The London Mathematical Society and The Heilbronn Institute for Mathematical Research (the UKRI/EPSRC Additional Funding Programme for Mathematical Sciences).}
}

\begin{document}
\maketitle

\begin{abstract}
We introduce an infinite-horizon, continuous-time portfolio selection problem faced by an agent with periodic S-shaped preference and present bias. The inclusion of a quasi-hyperbolic discount function leads to time-inconsistency and we characterize the optimal portfolio for a pre-committing, naive and sophisticated agent respectively. In the more theoretically challenging problem with a sophisticated agent, the time-consistent planning strategy can be formulated as an equilibrium to a static mean field game. Interestingly, present bias and naivety do not necessarily result in less desirable risk taking behaviors, while agent's sophistication may lead to excessive leverage (underinvestement) in the bad (good) states of the world.
\end{abstract}


\keywords
portfolio selection, quasi-hyperbolic discounting, present bias, time-inconsistency, S-shaped utility, equilibrium, fixed point.




\section{Introduction}

Present bias is a well-documented intertemporal behavioral bias of individuals. It refers to the tendency that a typical person enjoys receiving a smaller-and-sooner reward relative to a larger-and-later one, but such preference is reversed when both options are delayed equally. For example, a certain individual will prefer getting \$100 in one month rather than \$105 in three months, and this same individual will also prefer getting \$105 in fifteen months rather than \$100 in thirteen months. An individual with present bias is disproportionately impatient over short term outcomes while they are more patient over long term outcomes. Such phenomenon has been observed in many experimental and field studies. See, for example, \cite{thaler81}, \cite{loewenstein-prelec93}, \cite{kirby97}, \cite{ariely-wertenbroch02}, \cite{mcclure-etal04}, \cite{oster-scott05},  \cite{meier-sprenger10}, among others. 

In order to capture decreasing impatience rate, hyperbolic discounting has been widely adopted as a popular alternative to the classical exponential discounting criterion. In particular, one of its special variants known as quasi-hyperbolic discounting has found success in enabling macro-finance models to produce consumption and saving patterns that are more consistent with empirical data (\cite{phelps-pollak68}, \cite{laibson97}, \cite{laibson-repetto-tobacman98}, \cite{diamond-koszegi03}, \cite{harris-laibson01}, etc). Meanwhile, the impact of present bias on individuals' risk taking behaviors appears to be less commonly explored in the literature. While there are several theoretical studies of portfolio optimization under (quasi-)hyperbolic discounting, a typical finding is that the present bias component has no impact at all on the optimal portfolio choice. For example, \cite{palacios11}, \cite{zou-chen-wedge14}, \cite{chen-xiang-he19}, \cite{love-phelan15}, \cite{shin-roh19} and \cite{shigeta22} consider similar portfolio optimization problems formulated as some optimal consumption-and-investment problem based on Merton (\cite{Merton:69}, \cite{Merton:71}), and these papers all report that the optimal investment policy is just given by the classical Merton ratio which is independent of the agent's hyperbolic discount function. The irrelevance of present bias to risk taking is perhaps not too surprising in a consumption-based model. It is because the reward functional solely depends on the agent's lifetime consumption strategy. The investment decision and the portfolio value do not enter the agent's objective function directly but rather they only implicitly appear within the optimization problem at the level of an intertemporal budget constraint. Under a more general formulation of a consumption-and-investment problem under non-exponential discounting with objective function depending on both intertemporal consumption and terminal wealth, \cite{hamaguchi21} finds that the open-loop equilibrium strategy does not depend on how consumption utility is discounted (although it depends on the discount function applied to the utility of terminal wealth. See Remark 4.2 in \cite{hamaguchi21}). Further meaningful investigations of how present bias affects investment behaviors can also be found in applications beyond portfolio optimization. Examples include real option (\cite{grenadier-wang07}), corporate capital structure (\cite{tian16}) and production economy (\cite{liu16}).

While the ``Merton-style'' optimal consumption-and-investment problem is arguably the most important canonical approach to study portfolio selection, there are many real life applications where this framework is perhaps not too appropriate. An example is delegated portfolio management, in which the agent is a fund manager overseeing a portfolio for a client or a principal. The agent who invests on behalf of someone else cannot directly ``consume'' the underlying portfolio, but rather their incentives are tied to some remunerations derived from their trading performance assessed on a regular basis. Moreover, the agent's preference may also deviate drastically from a standard concave utility function due to managerial incentive distortion like limited liability protection as well as other psychological biases such as those described by Prospect Theory of \cite{tversky-kahneman1992}. To this end, we adopt the periodic portfolio selection model of \cite{TsZh21} under which the agent's reward functional is defined as the total discounted S-shaped utilities over the portfolio performance across some exogenously fixed periods. A novel consideration of our study is that we incorporate a quasi-hyperbolic discount function to describe the agent's intertemporal preference. Since the portfolio performance now directly enters the agent's running reward function, present bias has a first-order impact on the optimal investment strategy. Periodic portfolio selection under quasi-hyperbolic discounting therefore results in a mathematically and economically rich problem, which has not been covered in the literature yet to the best of our knowledge.

The contributions of our work are twofold. On the theoretical side, we give a complete characterization of the optimal portfolios, covering several criteria of optimality. Deviation from exponential discounting introduces time-inconsistency in a dynamic decision problem. Following \cite{strotz1956}, we consider three types of the agent: i) a pre-committing one who only solves the optimization problem once at the initial time and follows the derived strategy throughout the rest of the investment horizon; ii) a naive one who keeps re-optimizing at the beginning of each investment period, overriding any once-optimal plan made in the past; and iii) a sophisticated one who is aware of their time-inconsistency but is unable to commit to a given strategy, and hence they will act optimally at the current time point in response to the actions to be adopted by their future incarnations.

Our problem falls into the broad category of stochastic control problem with non-exponential discounting. Theoretical works in this direction include those previously cited papers on optimal consumption-and-investment problems with non-exponential discounting, as well as \cite{yong12}, \cite{hu-jin-zhou12}, \cite{hu-jin-zhou17}, \cite{bjork-khapko-murgoci17} \cite{jaskiewicz-nowak21}, \cite{ekeland-pirvu08}, \cite{ekeland-mbodji-pirvu12}, \cite{liu16}, among others. Many of these cited works employ the (extended) Hamilton-Jacobi-Bellman (HJB) equation or a flow of forward-backward stochastic differential equation (FBSDE) as a tool to analyze the time-inconsistent problem, especially for the derivation of an intrapersonal equilibrium strategy adopted by a sophisticated agent. We wish to stress that the structure of our periodic problem features a number of modeling elements which make it hard to proceed with the classical approaches. First, the periodic rewards depend on the historical values of the portfolio and therefore our problem is path-dependent. Second, the S-shaped utility function prevents us from conveniently characterizing the optimal portfolio strategy as a feedback control since it is difficult to establish the concavity/convexity behaviors of the value function upfront. Lastly, our problem has an infinite horizon without any ``terminal condition''. All these difficulties, when arising in conjunction with time-inconsistency and multiple notions of optimality, require us to seek alternative avenue to attack the problem. In this regard, we contribute to the literature of continuous-time, time-inconsistent stochastic control problem by showcasing new mathematical techniques which do not rely on the commonly adopted primal HJB nor the FBSDE approaches.

The philosophy of our approach is inspired by \cite{TsZh21} which utilizes a combination of discrete-time dynamic programming principle and martingale duality. The key insight is that for each type of agent (pre-committing, naive and sophisticated), one can attempt to characterize the optimal one-period gross return of the portfolio via a family of some finite horizon problems with maturity given by the length of the evaluation period. Each individual problem in such family can be solved by martingale duality. Then the correct element of this family can be identified by solving a suitable fixed point problem depending on the nature of the agent. Once we have found the optimal portfolio gross returns, the entire optimal wealth process can be constructed by a replication argument. Among the three types of agent we consider, the case of the sophisticated agent represents the most mathematically interesting and challenging problem. It turns out the problem resembles a static mean filed game with countably infinite number of players where the control is parametrized as a random variable. In general, the solution approaches considered in this paper are quite different from those in the existing literature. We believe this will shed lights on how non-standard time-inconsistent control problem can be analyzed by tools beyond the standard methods of HJB or FBSDE.

In parallel, our work also makes economic contributions to the portfolio optimization literature by connecting present bias and risk taking behaviors, where such connection has not received much attention to date. An agent with S-shaped utility is generally risk seeking in negative skewness, meaning that they tend to gamble aggressively in bad states of the world but reduce the risk taken in good states of the world. But the repeated nature of the periodic rewards introduces a ``continuation value component'' that distorts the agent's utility function. At a high level, the agent's degree of present bias affects their subjective weights across the short-term reward and the long-term continuation value. This in turn governs the agent's overall incentive and eventually leads to different investment behaviors. Our modeling framework provides a useful theoretical foundation to deduce some important policy and empirical implications concerning present bias and managerial risk taking. As a preview of our results, we find that: i) A present-biased agent takes more (less) negatively skewed risk relative to an exponential discounter when the investment prospect is sufficiently good (bad); and ii) A sophisticated agent takes more negatively skewed risk compared to their naive counterpart. 

The rest of the paper is organized as follows. Section \ref{sect:model} presents our modeling framework and we introduce the three concepts of optimality. Then in Section \ref{sect:dpe} we derive the dynamic programming equations for both the pre-committing and the sophisticated agent. An auxiliary family of optimization problems is studied in Section \ref{sect:auxprob}, which will then be used in Section \ref{sect:mainresults} to characterize the optimal portfolio strategies for the pre-committing, naive and sophisticated agent. Some comparative statics and further discussion of the economic intuitions can be found in Section \ref{sect:discuss}. Section \ref{sect:conclude} concludes. Miscellaneous technical materials are collected in the appendix. 

\section{The model}
\label{sect:model}

We present a periodic portfolio selection model which follows closely to that of \cite{TsZh21}.

\subsection{The economy}

Let $B$ be a one-dimensional Brownian motion on a complete probability space $(\Omega,\cF,\bP)$. For $s\geq0$, $\bF^s=(\cF^s_t)_{t\geq s}$ denotes the augmented filtration generated by $(B_t-B_s)_{t\geq s}$. Let $\bF:=\bF^0$ and $\cF_t:=\cF^0_t$ for $t\geq0$. We work with a Black-Scholes market consisting of one risky asset and one risk free money market instrument. The price process $S$ of the risky asset is a geometric Brownian motion with dynamics
\begin{equation*}
	\rd S_t=\mu S_t\rd t+\sigma S_t\rd B_t,\ t\geq0,\ S_0>0.
\end{equation*}
The price process $D$ of the risk free money market instrument is given by $D_t=e^{rt}$, $t\geq0$. We assume that $\mu,r\in\bR$ and $\sigma>0$ are constants. Write $\phi:=\frac{\mu-r}{\sigma}$ as the market price of risk or the Sharpe ratio of the risky asset. Throughout the paper, we assume that $\phi\neq0$. Then the unique pricing kernel $(Z_t)_{t\geq 0}$ of this Black-Scholes market is non-degenerate and has an expression of
\begin{equation}
	Z_t=\exp\left(-\phi B_t-\left(r+\frac{\phi^2}{2}\right)t\right),\ t\geq0.
\label{eq:bs_kernel}
\end{equation}

\subsection{Periodic evaluation and admissible strategies}

An agent forms a portfolio via investing dynamically in the risky asset and the risk free money market instrument. The portfolio value will be inspected on a sequence of evenly spaced dates of $T_i:=i\tau$ for $i\in\bN_0:=\bN\cup\{0\}$. Here, $\tau>0$ is a given constant representing the length of each evaluation period. For example, a choice of $\tau=1$ refers to an annual evaluation scheme.

Now we introduce the notion of a portfolio strategy. For $\pi=(\pi_t)_{t\geq T_n}\in L^2_{\bF^{T_n},\loc}(T_n,\infty;\bR)$, let $\pi_t$ represent the {\it relative rate} of the dollar amount invested in the risky asset at time $t$ to the wealth at the initial time of each period. Specifically, $u_t:=\pi_tX_{T_i}$ represents the dollar amount invested in the risky asset where $X_{T_i}$ denotes the agent's wealth at time $T_i$. The dynamics of the wealth process $X=X^{n,x,\pi}$ with a fixed initial wealth $X_{T_n}=x$ is then given by the following recursion
\begin{equation*}
	\rd X_t=(rX_t+(\mu-r)\pi_tX_{T_i})\rd t+\sigma\pi_tX_{T_i}\rd B_t,\ t\in[T_i,T_{i+1}),\ i\geq n,\ X_{T_n}=x,
\end{equation*}
which can be equivalently written as
\begin{equation*}
	\rd X_t=(rX_t+(\mu-r)\pi_tX_{\tau\left[\frac{t}{\tau}\right]})\rd t+\sigma\pi_tX_{\tau\left[\frac{t}{\tau}\right]}\rd B_t,\ t\geq T_n,\ X_{T_n}=x,
\end{equation*}
where $[\cdot]$ denotes the floor operator. It is easy to see that the above (path dependent) SDE has a unique continuous solution given by
\begin{equation*}
	X^{n,x,\pi}_t=X^{n,x,\pi}_{T_i}Y^{i,\pi}_t,\ t\in[T_i,T_{i+1}],\ i\geq n,\ X^{n,x,\pi}_{T_n}=x,
\end{equation*}
where $Y^{i,\pi}=(Y^{i,\pi}_t)_{t\in[T_i,T_{i+1}]}$ is the unique continuous solution to the SDE
\begin{equation*}
	\rd Y^{i,\pi}_t=(rY^{i,\pi}_t+(\mu-r)\pi_t)\rd t+\sigma\pi_t\rd B_t,\ t\in[T_i,T_{i+1}),\ Y^i_{T_i}=1.
\end{equation*}
Economically, each $Y^{i,\pi}$ is the {\it gross return rate process} of the portfolio over the $i$-th period $[T_i,T_{i+1})$.
Observe that $Y^{i,\pi}$ does not depend on $x$, but the wealth process $X^{n,x,\pi}$ and the dollar amount process $u_t=\pi_tX^{n,x,\pi}_{\tau\left[\frac{t}{\tau}\right]}$ depend on $x$.


\begin{rem}
In the literature, a usual practice to parametrize a trading strategy is via either the proportion of wealth invested in the risky asset or the dollar amount invested in the risky asset. In the former approach, one typically considers $p=(p_t)_{t\geq T_n}\in L^2_{\bF^{T_n},\loc}(T_n,\infty;\bR)$ such that $p_t$ denotes the {\it proportion of wealth} invested in the risky asset at time $t\geq T_n$. The dynamics of the wealth process $X=X^{n,x,p}$ with a given initial wealth $X_{T_n}=x\geq0$ is given by
\begin{equation*}
	\rd X_t=(r+(\mu-r) p_t)X_t\rd t+\sigma p_tX_t\rd B_t,\ t\geq T_n,\ X_{T_n}=x.
\end{equation*}
The above SDE has a $\bP$-a.s. positive unique solution given by
\begin{equation*}
	X_t=x\exp\Big(\int^t_{T_n}\Big(r+(\mu-r)p_s-\frac{1}{2}\sigma^2p^2_s\Big)\rd s+\int^t_{T_n}\sigma p_s\rd B_s\Big),\ t\geq T_n.
\end{equation*}
This therefore a priori excludes the possibility of bankruptcy (i.e. $X_t=0$ for some $t$), which could indeed arise within our problem. In the latter approach, one considers $u=(u_t)_{t\geq T_n}\in L^2_{\bF^{T_n},\loc}(T_n,\infty;\bR)$ where $u_t$ represents the {\it dollar amount} invested in the risky asset at time $t\geq T_n$. The corresponding wealth process $X=X^{n,x,u}$ with a given initial wealth $X_{T_n}=x\geq0$ has dynamics given by
\begin{equation*}
	\rd X_t=(rX_t+(\mu-r)u_t)\rd t+\sigma u_t\rd B_t,\ t\geq T_n,\ X_{T_n}=x.
\end{equation*}
In this framework, the future dollar amounts invested in the risky asset will be chosen depending on the initial wealth $x$ at time $T_n$. However, this is not suitable for our purpose as we want to allow the dollar amounts invested on each period, $[T_i,T_{i+1})$ with $i\geq n$, to be depending on the starting wealth $X_{T_i}$. In our problem, we want the control $u$ to be Markovian in $(t,X_t,L_t)$ where $L_t:=X_{\tau\left[\frac{t}{\tau}\right]}$. The notion of $\pi$ defined as the relative rate of dollar investment is a convenient way to incorporate this idea. 
\end{rem}

We introduce the following sets of admissible portfolios:
\begin{equation*}
	\Pi^i:=\left\{\pi=(\pi_t)_{t\in[T_i,T_{i+1})}\relmiddle|\begin{aligned}\bF^{T_i}\text{-progressively measurable}, \int^{T_{i+1}}_{T_i}|\pi_t|^2\rd t<\infty\ \text{a.s.},\\Y^{i,\pi}_t\geq0\ \text{for any $t\in[T_i,T_{i+1}]$ a.s.}\end{aligned}\right\}
\end{equation*}
for $i\in\bN_0$, and
\begin{equation*}
	\Pi^{(n)}:=\{\pi=(\pi_t)_{t\in[T_n,\infty)}\,|\,(\pi_t)_{t\in[T_i,T_{i+1})}\in\Pi^i\ \text{for any $i\geq n$}\}
\end{equation*}
for $n\in\bN_0$. We denote $\Pi:=\Pi^{(0)}$. Specifically, we require an admissible wealth process (or equivalently its gross return rate process) to be non-negative. For $\pi\in\Pi^{(n)}$ and $x\geq 0$, the wealth process $X^{n,x,\pi}=(X^{n,x,\pi}_t)_{t\in[T_n,\infty)}$ is not Markovian, but semi-Markovian in the sense that $(X^{n,x,\pi}_{T_i})_{i\geq n}$ is a Markov process. In our framework, wealth is allowed to hit zero with positive probability.

By standard arguments concerning a non-negative self-financing portfolio, for each $i\in\bN_0$ and $\pi\in\Pi^i$, the process $(\frac{Z_t}{Z_{T_i}}Y^{i,\pi}_t)_{t\in[T_i,T_{i+1}]}$ is a nonnegative local martingale and in turn a supermartingale. This implies that the ``static budget constraint''
\begin{equation*}
	\bE\Big[\frac{Z_{T_{i+1}}}{Z_{T_i}}Y^{i,\pi}_{T_{i+1}}\Big]\leq1
\end{equation*}
holds for any $i\in\bN_0$ and $\pi\in\Pi^i$. 

\begin{rem}
By Remark 3.4 in \cite{KaSh16}, bankruptcy is an absorbing state for $Y^{i,\pi}$: If the gross return rate becomes zero before time $T_{i+1}$, it stays there, and $\pi=0$ a.e.\ on $\{(\omega,t)\in\Omega\times[T_i,T_{i+1}]\,|\,Y^{i,\pi}_t(\omega)=0\}$. Actually, there exists an admissible $\pi\in\Pi^i$ such that $Y^{i,\pi}_{T_{i+1}}=0$. Let $\kappa^i:=\inf\{t\in[T_i,T_{i+1})\,|\,\int^t_{T_i}\frac{1}{\sqrt{T_{i+1}-s}}(\rd B_s+\theta\rd s)=-1\}\wedge T_{i+1}$, and define $\pi$ by $\pi_t:=\frac{e^{rt}}{\sigma\sqrt{T_{i+1}-t}}\1_{\{t\leq \kappa^i\}}$, $t\in[T_i,T_{i+1})$. Then $\pi\in\Pi^i$ and $Y^{i,\pi}_{T_{i+1}}=0$ a.s. This is the classical ``doubling down'' strategy.
\end{rem}


\subsection{Performance functional with quasi-hyperbolic discounting}

At time $t=T_{i+1}$ for each $i\in\bN_0$, the agent's trading performance within the period $[T_i,T_{i+1}]$ is measured as $X_{T_{i+1}}-\gamma X_{T_i}$, where $\gamma> 0$ is a performance benchmark parameter. The performance measure is positive if and only if the gross return of the portfolio $Y^{i}_{T_{i+1}}$ is larger than $\gamma$. The agent derives a burst of utility linked to their periodic performance given by $U(X_{T_{i+1}}-\gamma X_{T_i})$, where $U$ is a piecewise power utility function of \cite{tversky-kahneman1992} given by
\begin{equation*}
	U(x):=
	\begin{cases}
		x^\alpha,\ &x\geq0;\\
		-k|x|^\alpha,\ &x<0.
	\end{cases}
\end{equation*}
Here, $\alpha\in(0,1)$ is the parameter of risk aversion/seeking over the domain of gains/losses respectively, and $k\geq 0$ is the loss aversion parameter. 

We assume the agent exhibits present bias. Consider a quasi-hyperbolic discount function given by
\begin{eqnarray}
	D(s):=
	\begin{cases}
		e^{-\delta s}, & s\leq \tau;\\
		\beta e^{-\delta s}, & s>\tau.
	\end{cases}
\label{eq:qd}
\end{eqnarray}
In the above, $\delta>0$ plays the role of the discount rate while $\beta\in[0,1]$ represents the agent's myopia level. In the special case of $\beta=1$, $D(\cdot)$ degenerates to the standard exponential discount function.
\begin{rem}
Quasi-hyperbolic discount function is typically deployed in a discrete-time model. For our discount function in \eqref{eq:qd} under the current continuous-time setup, we assume the agent discounts outcomes in the distant future more heavily by an additional factor of $\beta\in[0,1]$. The duration cut-off between ``near-term'' and ``long-term'' is implicitly assumed to be $\tau$ in the definition of \eqref{eq:qd}, where the form of the discount factors changes beyond $s=\tau$. It is not an unreasonable assumption as the forthcoming evaluation date can serve as a natural ``mental anchor'' such that the agent will place more psychological focus on what is happening in the current evaluation period, and thus all outcomes within $s\in[0,\tau]$ are treated as ``near-term''.
\end{rem}

For each $n\in\bN_0$, $x\geq0$ and $\pi\in\Pi^{(n)}$, define the reward functional as the total net present value of the utilities over an infinite horizon beyond the initial time $T_n$, i.e.
\begin{align*}
	J_n(\pi;x)&:=\bE\Big[\sum^\infty_{i=n+1}D(T_i-T_n)U(X^{n,x,\pi}_{T_i}-\gamma X^{n,x,\pi}_{T_{i-1}})\Big]\\
	&=\bE\Big[e^{-\delta(T_{n+1}-T_n)}U(X^{n,x,\pi}_{T_{n+1}}-\gamma x)+\sum^\infty_{i=n+2}\beta e^{-\delta(T_i-T_n)}U(X^{n,x,\pi}_{T_i}-\gamma X^{n,x,\pi}_{T_{i-1}})\Big].
\end{align*}

If $\beta=0$, then the agent is completely myopic in the sense that they do not care about the performance in the future periods. In this case, the reward functional is reduced to a standard finite horizon maximization problem
	\begin{equation*}
		J_n(\pi;x)=\bE\Big[e^{-\delta(T_{n+1}-T_n)}U(X^{n,x,\pi}_{T_{n+1}}-\gamma x)\Big].
	\end{equation*}
Solution to the optimization problem with the above functional is well known (see \cite{berkelaar-kouwenberg-post04} for example). The optimal wealth process $\hat{X}$ is given by
\begin{equation*}
	\hat{X}_{T_{n+1}}=\hat{X}_{T_n}y\Big(\frac{Z_{T_{n+1}}}{Z_{T_n}}\Big),\ n\in\bN_0,\ \hat{X}_0=x.
\end{equation*}
Here,
\begin{equation*}
	y(z):=
	\begin{dcases}
		\gamma+\Big(\frac{\alpha}{\lambda^*z}\Big)^{\frac{1}{1-\alpha}},\ &z<z^*,\\
		0,\ &z\geq z^*,
	\end{dcases}
\end{equation*}
with $z^*\in(0,\infty)$ satisfying $\alpha^{\frac{\alpha}{1-\alpha}}(1-\alpha)(\lambda^*z^*)^{-\frac{\alpha}{1-\alpha}}-\gamma\lambda^*z^*+k\gamma^\alpha=0$ and $\lambda^*>0$ satisfying $\bE[Z_\tau y(Z_\tau)]=1$. 

If $\beta=1$, then the time-preference is reduced to the usual exponential discounting, and the problem becomes the one studied in \cite{TsZh21}:
	\begin{equation*}
		J_n(\pi;x)=\bE\Big[\sum^\infty_{i=n+1}e^{-\delta(T_i-T_n)}U(X^{n,x,\pi}_{T_i}-\gamma X^{n,x,\pi}_{T_{i-1}})\Big]. 
	\end{equation*}

If $\beta\in(0,1)$, then the agent's time-preference exhibits non-degenerate quasi-hyperbolic discounting. Portfolio optimization featuring both S-shaped utility and quasi-hyperbolic discounting has not been considered in the literature to date. The main goal of this paper is to study the corresponding optimal portfolio under different concepts of optimality.

Following \cite{TsZh21}, we impose a standing assumption which is a sufficient condition to ensure well-posedness of the problem when $\beta=1$. 
\begin{assum}
	The model parameters are such that
	\begin{equation}
		\delta>h:=r\alpha+\frac{\alpha\phi^2}{2(1-\alpha)}.
		\label{eq:assump}
	\end{equation}
\end{assum}

\begin{rem}
The constant $h$ defined in \eqref{eq:assump} is related to the solution to a finite-horizon Merton problem via $e^{h\tau}=\sup_{\pi\in\Pi^0}\bE[(Y^{0,\pi}_\tau)^\alpha]$.
\end{rem}

\subsection{Optimality criteria}
It is well known that non-exponential discounting induces {\it time-inconsistency}. If $\beta\in (0,1)$, then the (discrete) family of optimization problems $\{J_n(\cdot;x)\}_{n\in\bN_0,x\geq0}$ is time-inconsistent in the following sense: even if $\pi^{(n)}=(\pi^{(n)}_t)_{t\in[T_n,\infty)}$ is an optimal portfolio for $J_n(\cdot;x)$, the restriction $\pi^{(n+1)}=(\pi^{(n)}_t)_{t\in[T_{n+1},\infty)}$ of $\pi^{(n)}$ on the future periods $[T_{n+1},\infty)$ might not be optimal for $J_{n+1}(\cdot;X^{n,x,\pi^{(n)}}_{T_{n+1}})$. It is therefore not even clear upfront what the meaning of optimality is in presence of quasi-hyperbolic discounting. Following \cite{strotz1956}, we will consider three different notions of optimality: pre-committing, naive, and sophisticated.

\begin{defi}[Pre-committing agent]
Fix a reference time $T_n$ for some $n\in\bN_0$. $\pi^{\pre,(n)}=(\pi^{\pre,(n)}_t)_{t\geq T_n}\in\Pi^{(n)}$ is said to be an optimal pre-committing strategy with respect to time $T_n$ if 
\begin{equation*}
	J_n(\pi^{\pre,(n)};x)=\sup_{\pi\in\Pi^{(n)}} J_n(\pi;x)
\end{equation*}
for all $x>0$.
\label{def:precomm}
\end{defi}

The pre-committing agent solves the optimization problem only once at some initial reference time point, say time zero. Then the agent is able to adhere to this derived optimal strategy $(\pi^{\pre,(0)}_{t})_{t\geq 0}$ throughout the rest of the investment horizon. If this agent ever attempts to reevaluate this strategy in the future, say at $t=T_1$, they will find that $(\pi^{\pre,(0)}_{t})_{t\geq T_1}$ is not the strategy that gives them the maximized value of $J_1$. But for as long as they are able to commit, they will stick to $\pi^{\pre,(0)}$ perpetually even though it becomes sub-optimal when reviewed again in the future.

\begin{defi}[Naive agent]
For $\pi^{\pre,(i)}=(\pi^{\pre,(i)}_t)_{t\geq T_i}$ being the optimal pre-committing strategy with respect to time $T_i$ as defined in Definition \ref{def:precomm}, $(\pi^{\nav}_t)_{t\geq 0}\in\Pi$ is said to be an optimal naive strategy if
\begin{equation*}
	\pi^{\nav}_t=\pi^{\pre,(i)}_t,\ t\in[T_i,T_{i+1})
\end{equation*}
for all $i\in\bN_0$.
\label{def:naive}
\end{defi}
Unlike a pre-committing agent, a naive agent always reoptimizes at the beginning of each period. Their trading strategy to be taken in the period $[T_{i},T_{i+1})$ is guided by the first-period segment of the solution to the problem $\sup_{\pi\in\Pi^{(i)}} J_i(\pi;x)$, overriding any planned strategy derived in the past. An implicit assumption we are making in Definition \ref{def:naive} is that the naive agent is still able to commit to a derived strategy $\pi^{\pre,(i)}_t$ for the duration of one period $[T_i,T_{i+1})$, and reoptimization only takes place at the beginning of each period. It is a reasonable assumption as we expect individuals are indeed capable of self-control over a relatively short time horizon.

The last notion of optimality is an intrapersonal equilibrium strategy adopted by a  sophisticated agent. In this case, the agent is aware that they will suffer from time-inconsistency and is not able to commit to a strategy for more than one period. They view the future incarnations of themself (whom they cannot control) as opponents in a sequential game. They then act optimally in the current period in response to the strategies adopted by their future selves. Equilibrium is achieved when each incarnation of the agent at each time point has no incentive to deviate from their chosen action.

Before stating the formal definition, we need to introduce some further notations. For each $\pi=(\pi_t)_{t\in[0,\infty)}\in\Pi=\Pi^{(0)}$ and $n\in\bN$, $\pi^{(n)}:=(\pi_t)_{t\in[T_n,\infty)}\in\Pi^{(n)}$ denotes the restriction of $\pi$ on $[T_n,\infty)$ and $\pi^n:=(\pi_t)_{t\in[T_{n},T_{n+1})}\in\Pi^n$. Let $n\in\bN_0$. For each $\pi^{(n+1)}=(\pi^{(n+1)}_t)_{t\in[T_{n+1},\infty)}\in\Pi^{(n+1)}$ and $\pi^n=(\pi^n_t)_{t\in[T_n,T_{n+1})}\in\Pi^n$, define the concatenated portfolio strategy $\pi^n\oplus\pi^{(n+1)}\in\Pi^{(n)}$ by
\begin{equation*}
		(\pi^n\oplus\pi^{(n+1)})_t:=
		\begin{cases}
			\pi^n_t,\ &t\in[T_n,T_{n+1}),\\
			\pi^{(n+1)}_t,\ &t\in[T_{n+1},\infty).
		\end{cases}
\end{equation*}

\begin{defi}[Sophisticated agent]
$\hat{\pi}\in\Pi$ is said to be a subgame perfect equilibrium portfolio strategy of a sophisticated agent if 
\begin{equation*}
	J_n(\hat{\pi}^{(n)};x)=\sup_{\pi^n\in\Pi^n}J_n(\pi^n\oplus\hat{\pi}^{(n+1)};x)
\end{equation*}
for any $n\in\bN$ and any $x\geq0$. We call the corresponding wealth process $\hat{X}=X^{0,x,\hat{\pi}}$ (with an initial wealth $x\geq0$) the equilibrium wealth process, and the function $V^{\hat{\pi}}_n(x):=J_n(\hat{\pi}^{(n)};x)$ the equilibrium value function corresponding to $\hat{\pi}$.
\label{def:soph}
\end{defi}

\begin{rem}
	The problem in Definition \ref{def:soph} is a game of countably infinite number of players. Compared to the finite horizon time-inconsistent problem, the existence of subgame perfect equilibria is a highly nontrivial issue since we have no boundary condition given by the terminal reward. Concerning the uniqueness, we conjecture that we may generally have \emph{multiple equilibria}. See the discussion in Section 6 of \cite{BjMu14}. Moreover, unlike the time-consistent problem, the equilibrium value function may not be unique, and it depends on each subgame perfect equilibrium portfolio strategy.
\end{rem}

\section{Dynamic programming equation for pre-committing and sophisticated agent}
\label{sect:dpe}

In this section, we derive the dynamic programming principle associated with the optimization problem faced by both the pre-committing agent and the sophisticated agent.

\subsection{Pre-committing agent}

For each $n\in\bN_0$, $x\geq0$ and $\pi\in\Pi^{(n)}$, define
\begin{equation*}
	\tilde{J}_n(\pi;x):=\bE\Big[\sum^\infty_{i=1}e^{-\delta\tau i}U(X^{n,x,\pi}_{T_{n+i}}-\gamma X^{n,x,\pi}_{T_{n+i-1}})\Big],
\end{equation*}
which represents the reward functional faced by an exponential discounter (i.e. an agent with $\beta=1$). One can then express the reward functional of a possibly myopic agent with $\beta\in[0,1]$ in terms of that of an exponential discounter. In particular, observe that
\begin{align*}
	J_n(\pi;x)&=J_n(\pi^n\oplus\pi^{(n+1)};x)\\
	&=\bE\Big[e^{-\delta\tau}U(X^{n,x,\pi^n}_{T_{n+1}}-\gamma x)+\beta e^{-\delta\tau}\tilde{J}_{n+1}(\pi^{(n+1)};X^{n,x,\pi^n}_{T_{n+1}})\Big].
\end{align*}
Define the pre-committed value function by
\begin{align*}
	V_\pre(x)&:=\sup_{\pi\in\Pi^{(n)}}J_n(\pi^{(n)};x)=\sup_{\pi\in\Pi}J_0(\pi;x),
\end{align*}
which does not depend on $n$ due to the time-homogeneous structure of the problem. Similarly, let
\begin{align*}
	V_\expo(x):=\sup_{\pi\in\Pi^{(n)}}\tilde{J}_n(\pi;x)=\sup_{\pi\in\Pi}\tilde{J}_0(\pi;x),
\end{align*}	
which is the value function of an exponential discounter. By dynamic programming principle, we have 
\begin{equation}
	\begin{dcases}
		V_\pre(x)=\sup_{\pi\in\Pi^0}\bE\Big[e^{-\delta\tau}U(X^{0,x,\pi}_\tau-\gamma x)+\beta e^{-\delta\tau}V_\expo(X^{0,x,\pi}_\tau)\Big],\\
		V_\expo(x)=\sup_{\pi\in\Pi^0}\bE\Big[e^{-\delta\tau}U(X^{0,x,\pi}_\tau-\gamma x)+e^{-\delta\tau}V_\expo(X^{0,x,\pi}_\tau)\Big].
	\end{dcases}
	\label{eq:precom_dpp}
\end{equation}

By the scaling property of the utility function $U$ where $U(cz)=c^\alpha U(z)$ for any $c\geq 0$ and $z\in\mathbb{R}$, we have
\begin{equation*}
	V_\pre(x)=A_\pre x^\alpha,\ V_\expo(x)=A_\expo x^\alpha,
\end{equation*}
where the constants $A_\pre,A_\expo\in\bR$ are defined by
\begin{equation*}
	A_\pre:=\sup_{\pi\in\Pi}J_0(\pi;1),\ A_\expo:=\sup_{\pi\in\Pi}\tilde{J}_0(\pi;1).
\end{equation*}
Moreover, under standard martingale duality argument, optimization in \eqref{eq:precom_dpp} can be performed over the one-period stochastic gross return rates $Y:=X^{0,1,\pi}_{\tau}$ rather than the trading strategies $\pi$. Formally, define a set of random variables
\begin{equation}
	\cY:=\{Y|Y\in\mathcal{F}_\tau,\ Y\geq 0, \ \bE[Z_\tau Y]\leq 1\},
\label{eq:setY}
\end{equation}
where $Z=(Z_t)_{t\geq 0}$ is defined in \eqref{eq:bs_kernel}. Economically, $\cY$ contains all terminal wealth variables at time $\tau$ that can be attained by a non-negative self-financing portfolio in the Black-Scholes economy starting with one unit of initial capital. On substituting the ansatzes, dividing both sides of \eqref{eq:precom_dpp} by $x^\alpha$ and taking $Y=X_\tau^{0,x,\pi}\in\cY$ as the decision variable, \eqref{eq:precom_dpp} becomes
\begin{equation}
	\begin{dcases}
		A_\pre=\sup_{Y\in\cY}\bE\Big[e^{-\delta\tau}U(Y-\gamma)+\beta e^{-\delta\tau}A_\expo Y^\alpha\Big],\\
		A_\expo=\sup_{Y\in\cY}\bE\Big[e^{-\delta\tau}U(Y-\gamma)+e^{-\delta\tau}A_\expo Y^\alpha\Big].
	\end{dcases}
	\label{eq:precomm_sys}
\end{equation}

Our goal is to solve for $(A_\pre,A_\expo)$ from the system \eqref{eq:precomm_sys}, and to identify $Y_\pre\in\cY$ and $Y_\expo\in\cY$ satisfying
\begin{equation*}
	\begin{dcases}
		A_\pre=\bE\Big[e^{-\delta\tau}U(Y_\pre-\gamma)+\beta e^{-\delta\tau}A_\expo Y^\alpha_\pre\Big],\\
		A_\expo=\bE\Big[e^{-\delta\tau}U(Y_\expo-\gamma)+e^{-\delta\tau}A_\expo Y_\expo^\alpha\Big].
	\end{dcases}
\end{equation*}
Then $Y_\pre$ will represent the optimal gross return variable for the first period, and $Y_\expo$ will be the optimal gross return variable for all the subsequent periods. In particular, for $n\in\bN_0$ and $x\geq0$, the portfolio $\pi\in\Pi^{(n)}$ satisfying
\begin{equation*}
	X^{n,x,\pi}_{T_n}=x,\ X^{n,x,\pi}_{T_{n+1}}=x Y_{\pre,n+1},\ X^{n,x,\pi}_{T_{n+i}}=X^{n,x,\pi}_{T_{n+i-1}}Y_{\expo,n+i},\ i\geq2,
\end{equation*}
is an optimal pre-committed portfolio for $J_n(\cdot,x)$. Here, for $Y_{\expo}\in\cY$, $Y_{\pre}\in\cY$ and each $k\in\bN$, $Y_{\expo,k}$ and $Y_{\pre,k}$ denote the $\cF^{T_{k-1}}_{T_k}$-measurable copy of $Y_{\expo}$ and $Y_{\pre}$ respectively. 

\subsection{Sophisticated agent}

We now consider the case with a sophisticated agent via deriving an extended Bellman equation similar to the one studied in \cite{BjMu14}. Inspired by the benchmark problem without present bias, we focus on searching for an equilibrium strategy $\hat{\pi}\in\Pi$ which is {\it periodic}, in the sense that $(\hat{\pi}_t)_{t\in[T_n,T_{n+1})}$ with $n\in\bN_0$ are identically distributed.

\begin{prop}
Define 
\begin{equation*}
	W^{\pi}(x):=\bE\Big[\sum^\infty_{i=1}e^{-\delta\tau i}U(X^{0,x,\pi}_{T_i}-\gamma X^{0,x,\pi}_{T_{i-1}})\Big],
\end{equation*}
which represents the reward functional of an exponential discounter under a given strategy $\pi$. A periodic portfolio $\hat{\pi}\in\Pi$ is a subgame perfect equilibrium strategy if and only if the following holds:
	\begin{equation}
		\begin{dcases}
			V^{\hat{\pi}}(x)=\bE\Big[e^{-\delta\tau}U(X^{0,x,\hat{\pi}}_\tau-\gamma x)+\beta e^{-\delta\tau}W^{\hat{\pi}}(X^{0,x,\hat{\pi}}_\tau)\Big]\\
			\hspace{1.05cm}=\sup_{\pi^0\in\Pi^0}\bE\Big[e^{-\delta\tau}U(X^{0,x,\pi^0}_\tau-\gamma x)+\beta e^{-\delta\tau}W^{\hat{\pi}}(X^{0,x,\pi^0}_\tau)\Big],\\
			W^{\hat{\pi}}(x)=V^{\hat{\pi}}(x)+(1-\beta)e^{-\delta\tau}\bE[W^{\hat{\pi}}(X^{0,x,\hat{\pi}}_\tau)],
		\end{dcases}
	\label{eq_exHJB}
	\end{equation}
	where $V^{\hat{\pi}}(x):=J_0(\hat{\pi};x)$. 
\end{prop}
\label{prop:sophis_system}
\begin{proof}

Thanks to the periodicity of $\hat{\pi}$, $V^{\hat{\pi}}(x):=J_0(\hat{\pi};x)=J_n(\hat{\pi};x)$ for any $n$. Moreover, for each $n\in\bN_0$ and $x>0$, $(X^{n,x,\hat{\pi}}_{T_{n+i}})_{i\in\bN_0}$ is independent of $\cF_{T_n}$, and the joint distribution does not depend on $n$. Therefore, for each $n\in\bN_0$, $x>0$ and $\pi^n\in\Pi^n$, we have
\begin{align*}
	&J_n(\pi^n\oplus\hat{\pi}^{(n+1)};x)\\
	&=\bE\Big[e^{-\delta(T_{n+1}-T_n)}U(X^{n,x,\pi^n\oplus\hat{\pi}^{(n+1)}}_{T_{n+1}}-\gamma x)+\sum^\infty_{i=n+2}\beta e^{-\delta(T_i-T_n)}U(X^{n,x,\pi^n\oplus\hat{\pi}^{(n+1)}}_{T_i}-\gamma X^{n,x,\pi^n\oplus\hat{\pi}^{(n+1)}}_{T_{i-1}})\Big]\\
	&=\bE\Big[e^{-\delta(T_{n+1}-T_n)}U(X^{n,x,\pi^n}_{T_{n+1}}-\gamma x)+\sum^\infty_{i=n+2}\beta e^{-\delta(T_i-T_n)}U(X^{n+1,X^{n,x,\pi^n}_{T_{n+1}},\hat{\pi}}_{T_i}-\gamma X^{n+1,X^{n,x,\pi^n}_{T_{n+1}},\hat{\pi}}_{T_{i-1}})\Big]\\
	&=\bE\Big[e^{-\delta\tau}U(X^{n,x,\pi^n}_{T_{n+1}}-\gamma x)+\beta e^{-\delta\tau}\bE\Big[\sum^\infty_{i=1}e^{-\delta\tau i}U(X^{n+1,X^{n,x,\pi^n}_{T_{n+1}},\hat{\pi}}_{T_{n+i+1}}-\gamma X^{n+1,X^{n,x,\pi^n}_{T_{n+1}},\hat{\pi}}_{T_{n+i}})|\cF_{T_{n+1}}\Big]\Big]\\
	&=\bE\Big[e^{-\delta\tau}U(X^{n,x,\pi^n}_{T_{n+1}}-\gamma x)+\beta e^{-\delta\tau}W^{\hat{\pi}}(X^{n,x,\pi^n}_{T_{n+1}})\Big].
\end{align*}
Then, by Definition \ref{def:soph}, $\hat{\pi}$ is a subgame perfect equilibrium if and only if
\begin{align*}
	V^{\hat{\pi}}(x)&=\bE\Big[e^{-\delta\tau}U(X^{n,x,\hat{\pi}}_{T_{n+1}}-\gamma x)+\beta e^{-\delta\tau}W^{\hat{\pi}}(X^{n,x,\hat{\pi}}_{T_{n+1}})\Big]\\
	&=\sup_{\pi^n\in\Pi^n}\bE\Big[e^{-\delta\tau}U(X^{n,x,\pi^n}_{T_{n+1}}-\gamma x)+\beta e^{-\delta\tau}W^{\hat{\pi}}(X^{n,x,\pi^n}_{T_{n+1}})\Big]\\
	&=\sup_{\pi^0\in\Pi^0}\bE\Big[e^{-\delta\tau}U(X^{0,x,\pi^0}_\tau-\gamma x)+\beta e^{-\delta\tau}W^{\hat{\pi}}(X^{0,x,\pi^0}_\tau)\Big].
\end{align*}
On the other hand,
\begin{align*}
	W^{\hat{\pi}}(x)&=V^{\hat{\pi}}(x)+(1-\beta)\bE\Big[\sum^\infty_{i=2}e^{-\delta\tau i}U(X^{0,x,\hat{\pi}}_{T_i}-\gamma X^{0,x,\hat{\pi}}_{T_{i-1}})\Big]\\
	&=V^{\hat{\pi}}(x)+(1-\beta)e^{-\delta\tau}\bE\Big[\bE\Big[\sum^\infty_{i=1}U(X^{1,X^{0,x,\hat{\pi}}_\tau,\hat{\pi}}_{T_{i+1}}-\gamma X^{1,X^{0,x,\hat{\pi}}_\tau,\hat{\pi}}_{T_i})|\cF_\tau\Big]\Big]\\
	&=V^{\hat{\pi}}(x)+(1-\beta)e^{-\delta\tau}\bE[W^{\hat{\pi}}(X^{0,x,\hat{\pi}}_\tau)].
\end{align*}

Thus, a periodic portfolio $\hat{\pi}$ is a subgame perfect equilibrium strategy if and only if
\begin{equation*}
	\begin{dcases}
		V^{\hat{\pi}}(x)=\bE\Big[e^{-\delta\tau}U(X^{0,x,\hat{\pi}}_\tau-\gamma x)+\beta e^{-\delta\tau}W^{\hat{\pi}}(X^{0,x,\hat{\pi}}_\tau)\Big]\\
		\hspace{1.05cm}=\sup_{\pi^0\in\Pi^0}\bE\Big[e^{-\delta\tau}U(X^{0,x,\pi^0}_\tau-\gamma x)+\beta e^{-\delta\tau}W^{\hat{\pi}}(X^{0,x,\pi^0}_\tau)\Big],\\
		W^{\hat{\pi}}(x)=V^{\hat{\pi}}(x)+(1-\beta)e^{-\delta\tau}\bE[W^{\hat{\pi}}(X^{0,x,\hat{\pi}}_\tau)].
	\end{dcases}
\end{equation*}
This system is an extended HJB equation in the spirit of \cite{BjMu14}.	
\end{proof}

To proceed further, we again exploit the scaling property of the utility function $U(\cdot)$. $V^{\hat{\pi}}$ and $W^{\hat{\pi}}$ are in form of 
\begin{equation*}
	V^{\hat{\pi}}(x)=\hat{A} x^\alpha,\ W^{\hat{\pi}}(x)=\hat{B} x^\alpha,
\end{equation*}
for some constants $\hat{A},\hat{B}$ to be determined which depend on $\hat{\pi}$. Using these ansatzes followed by division of $x^\alpha$ on both sides of \eqref{eq_exHJB}, we obtain
\begin{equation}
	\begin{dcases}
		\hat{A}=\bE\Big[e^{-\delta\tau}U(X^{0,1,\hat{\pi}}_\tau-\gamma )+\beta e^{-\delta\tau}\hat{B}(X^{0,1,\hat{\pi}}_\tau)^\alpha\Big]\\
		\hspace{1.05cm}=\sup_{\pi^0\in\Pi^0}\bE\Big[e^{-\delta\tau}U(X^{0,1,\pi^0}_\tau-\gamma )+\beta e^{-\delta\tau}\hat{B}(X^{0,1,\pi^0}_\tau)^\alpha\Big],\\
		\hat{B}=\hat{A}+(1-\beta)e^{-\delta\tau}\hat{B}\bE[(X^{0,1,\hat{\pi}}_\tau)^\alpha].
	\end{dcases}
\end{equation}
As in the analysis of the pre-committing problem, one can replace the decision variable of $\pi^0\in\Pi^0$ by $Y:=X^{0,1,\pi^0}\in\cY$. Then if we write $\hat{Y}:=X^{0,1,\hat{\pi}}_\tau$, the extended Bellman system can be expressed as
\begin{equation*}
	\begin{dcases}
		\hat{A}=\bE\Big[e^{-\delta\tau}U(\hat{Y}-\gamma)+\beta e^{-\delta\tau}\hat{B}\hat{Y}^\alpha\Big]=\sup_{Y\in\cY}\bE\Big[e^{-\delta\tau}U(Y-\gamma)+\beta e^{-\delta\tau}\hat{B}Y^\alpha\Big],\\
		\hat{B}=\hat{A}+(1-\beta)e^{-\delta\tau}\hat{B}\bE[\hat{Y}^\alpha].
	\end{dcases}
\end{equation*}
Observe that, by the standing assumption \eqref{eq:assump}, $\sup_{Y\in\cY}\bE[Y^\alpha]\leq e^{h\tau}$, and hence $(1-\beta)e^{-\delta\tau}\bE[Y^\alpha]\in[0,1)$ for any $Y\in\cY$. We can then express $\hat{B}$ by $(\hat{A},\hat{Y})$ as
\begin{equation*}
	\hat{B}=\frac{\hat{A}}{1-(1-\beta)e^{-\delta\tau}\bE[\hat{Y}^\alpha]}.
\end{equation*}
After eliminating $\hat{B}$ from the system, we get
\begin{align}
	\hat{A}&=\bE\Big[e^{-\delta\tau}U(\hat{Y}-\gamma)+\frac{\beta e^{-\delta\tau}}{1-(1-\beta)e^{-\delta\tau}\bE[\hat{Y}^\alpha]}\hat{A}\hat{Y}^\alpha\Big]
	=\sup_{Y\in\cY}\bE\Big[e^{-\delta\tau}U(Y-\gamma)+\frac{\beta e^{-\delta\tau}}{1-(1-\beta)e^{-\delta\tau}\bE[\hat{Y}^\alpha]}\hat{A}Y^\alpha\Big].
	\label{eq:sophi_sys}
\end{align}

Our main goal is to find $\hat{A}\in\bR$ and $\hat{Y}\in\cY$ solving system \eqref{eq:sophi_sys}. If such $(\hat{A},\hat{Y})$ exists, then $\hat{Y}$ is an {\it equilibrium gross return} variable. By the standard replication arguments, there exists $\hat{\pi}\in\Pi$ such that
\begin{equation*}
	X^{0,x,\hat{\pi}}_{T_n}=X^{0,x,\hat{\pi}}_{T_{n-1}}\hat{Y}_n,\ n\in\bN,\ X^{0,x,\hat{\pi}}_0=x,
\end{equation*}
where $\hat{Y}_n$ is an $\cF^{T_{n-1}}_{T_n}$-measurable independent copy of $\hat{Y}$. This strategy $\hat{\pi}$ is then a periodic equilibrium strategy solving the problem faced by the sophisticated agent. Moreover, $\hat{V}(x):=\hat{A} x^{\alpha}$ is the corresponding equilibrium value function under such $\hat{\pi}$.

For $\beta\in(0,1)$, it is challenging to solve system \eqref{eq:sophi_sys} because the optimization problem
\begin{align*}
	\sup_{Y\in\cY}\bE\Big[e^{-\delta\tau}U(Y-\gamma)+\frac{\beta e^{-\delta\tau}}{1-(1-\beta)e^{-\delta\tau}\bE[\hat{Y}^\alpha]}\hat{A}Y^\alpha\Big]
\end{align*}
involves both an unknown constant $\hat{A}$ and an unknown random variable $\hat{Y}$ which have to be simultaneously determined as a part of the solution. Interestingly, \eqref{eq:sophi_sys}  is analogous to a static mean field game in the following sense: The current-self of the agent is responding to a countable but infinite number of players (the agent's infinite copies of their future-selves in all subsequent periods), whose collective strategy induces a value and a probability distribution that are eventually fed back to the objective function of the current-self. 

\begin{rem}
	In the special case of $\beta=1$, \eqref{eq:sophi_sys} simplifies to
	\begin{align*}
		\hat{A}&=\bE\Big[e^{-\delta\tau}U(\hat{Y}-\gamma)+e^{-\delta \tau}\hat{A}\hat{Y}^\alpha\Big]
		=\sup_{Y\in\cY}\bE\Big[e^{-\delta\tau}U(Y-\gamma)+e^{-\delta \tau}\hat{A}Y^\alpha\Big].
	\end{align*}
	The objective function in the last term no longer depends on $\hat{Y}$, and the problem degenerates to the one faced by an exponential discounter as considered in \cite{TsZh21}. In another special case of $\beta=0$ where the agent is completely myopic, \eqref{eq:sophi_sys} becomes
	\begin{align*}
		\hat{A}&=\bE\Big[e^{-\delta\tau}U(\hat{Y}-\gamma)\Big]
		=\sup_{Y\in\cY}\bE\Big[e^{-\delta\tau}U(Y-\gamma)\Big].
	\end{align*}
	This only involves solving a standard one-period portfolio optimization problem with terminal utility function $U$.
\end{rem}

\section{A family of one-period optimization problems}
\label{sect:auxprob}

From the analysis in Section \ref{sect:dpe}, characterization of the optimal portfolio entails solving system \eqref{eq:precomm_sys} (for the precommitting agent) or \eqref{eq:sophi_sys} (for the sophisticated agent). In this section, we study a family of one-period portfolio optimization problems which will serve as an important building block for construction of solutions to the problems introduced in Section \ref{sect:dpe}.

Define
\begin{equation}
	F(y;\theta):=U(y-\gamma)+\theta y^\alpha
\label{eq:F}
\end{equation}
over $(y,\theta)\in[0,\infty)\times\bR$.
Now, consider a family of optimization problems parametrized by $\theta\in\bR$ as
\begin{equation}
	\Phi(\theta):=\sup_{Y\in\cY}\bE[F(Y;\theta)],
\label{eq:axu}
\end{equation}
where $\cY$ is defined in \eqref{eq:setY}. 

Problem \eqref{eq:axu} is studied in details in \cite{TsZh21}. In the rest of this section, we will briefly recap some of the key results but will also state some new ones along the way. The approach to solve problem \eqref{eq:axu} is to consider a concavified problem $\sup_{Y\in\cY}\bE[\bar{F}(Y;\theta)]$, where $y\mapsto \bar{F}(y;\theta)$ is the smallest concave majorant of $y\mapsto F(y;\theta)$.\footnote{For $f:[0,\infty)\to\bR$, $\bar{f}(y):=\inf_{(a,b)\in\cA_f}\{ay+b\}$ with $\cA_f:=\{(a,b)\in\bR^2\,;\,f(y)\leq ay+b\ \forall\,y\in[0,\infty)\}$ is the smallest concave majorant of $f$. The function $\bar{f}$ is real-valued if and only if the positive part of $f(y)$ has at most linear growth.} The subtlety here is that the value of $\theta$ heavily influences the monotonicity and concavity/convexity behaviors of $F(y;\theta)$ and in turn $\bar{F}(y;\theta)$. We recall the following results from \cite{TsZh21} where there are four canonical cases.

\begin{lemm}[Lemma EC.2 of \cite{TsZh21}]

Fix $\theta$ and let $\bar{F}(y;\theta)$ be the smallest concave majorant of $F(y;\theta)$. Define
\begin{equation*}
	\underline{\theta}:=-(1+k^{\frac{1}{1-\alpha}})^{1-\alpha}\leq -1.
\end{equation*}

\begin{enumerate}
	\item If $\theta\in(0,\infty)$, then $\bar{F}(y;\theta)=F(y;\theta)$ on $y\in [0,c_1\gamma]\cup[c_2\gamma,\infty)$, and $\bar{F}(y;\theta)$ is a straight line with slope $m_1>0$ joining $(c_1\gamma,F(c_1\gamma;\theta))$ and $(c_2\gamma,F(c_2\gamma;\theta))$ on $y\in [c_1\gamma,c_2\gamma]$. Here
	\begin{equation}
		m_1=m_1(\theta):=\gamma^{\alpha-1}\frac{(c_2-1)^{\alpha}+\theta c_2^{\alpha}+k(1-c_1)^{\alpha}-\theta c_1^{\alpha}}{c_2-c_1},
		\label{eq:m1}
	\end{equation}
	and $c_1=c_1(\theta),c_2=c_2(\theta)$ are two constants which are the unique solutions to the system of equations
	\begin{equation}
		\frac{(c_2-1)^{\alpha}+\theta c_2^{\alpha}+k(1-c_1)^{\alpha}-\theta c_1^{\alpha}}{c_2-c_1}=\alpha[(c_2-1)^{\alpha-1}+\theta   c_2^{\alpha-1}]=\alpha[k (1-c_1)^{\alpha-1}+\theta   c_1^{\alpha-1}]
		\label{eq:c1c2}
	\end{equation}
 on $(c_1,c_2)\in(0,\frac{1}{1+(k/\theta)^{1/(2-\alpha)}})\times (1,\infty)$.

	\item If $\theta\in[-1,0]\cap(\underline{\theta},0]$,  then $\bar{F}(y;\theta)=F(y;\theta)$ on $y\in [c_3\gamma,\infty)$, and $\bar{F}(y;\theta)$ is a straight line with slope $m_2>0$ joining $(0,F(0;\theta))$ and $(c_3\gamma,F(c_3\gamma;\theta))$ on $y\in[0,c_3\gamma]$. Here
	\begin{align}
		m_2=m_2(\theta):=\gamma^{\alpha-1}\frac{(c_3-1)^{\alpha}+\theta c_3^{\alpha}+k}{c_3},
		\label{eq:m2}
	\end{align}
	and $c_3=c_3(\theta)$ is the unique solution to the equation
	\begin{align}
		\frac{(c_3-1)^{\alpha}+\theta c_3^{\alpha}+k}{c_3}=\alpha[(c_3-1)^{\alpha-1}+\theta   c_3^{\alpha-1}]
		\label{eq:c3}
	\end{align}
on $c_3\in(1,\infty)$.
	
	\item If $\theta\in(\underline{\theta},-1)$, then $\bar{F}(y;\theta)=F(y;\theta)$ on $y\in [c_3\gamma,c_4\gamma]$, $\bar{F}(y;\theta)$ is a straight line with slope $m_2>0$ joining $(0,F(0;\theta))$ and $(c_3\gamma,F(c_3\gamma;\theta))$ on $[0,c_3\gamma]$ and $\bar{F}(y;\theta)=F(c_4\gamma;\theta)$ on $y\in[c_4\gamma, \infty)$. Here $m_2$ is defined in \eqref{eq:m2},
	\begin{equation}
		c_4=c_4(\theta):=\frac{1}{1-|\theta|^{-1/(1-\alpha)}},
	\end{equation}
and $c_3=c_3(\theta)$ is the unique solution to equation \eqref{eq:c3} on $c_3\in(1,c_4)$.
	
	\item If $\theta\in(-\infty, \underline{\theta}]$, then $\bar{F}(y;\theta)=F(0;\theta)=-k\gamma^\alpha$.
\end{enumerate}
\label{lem:shape_F}
\end{lemm}
Readers are referred to Figure EC.1 in \cite{TsZh21} for the graphical illustrations of these four cases.

\begin{rem}
In the corner case of $k=0$, we have $\underline{\theta}=-1$ and therefore Case 3 in Lemma \ref{lem:shape_F} vanishes while the range of $\theta$ in Case 2 becomes $\theta\in(-1,0]$. 
\end{rem}

To describe the solution of problem \eqref{eq:axu}, we need to introduce several more notations. For  $\theta>0$ and $q\in[\tilde{n},\infty)$, define $I_1(q;\theta)\in(0,\tilde{c}\gamma]$ as the unique solution to the equation
\begin{equation}
	\alpha[k(\gamma-y)^{\alpha-1}+\theta y^{\alpha-1}]=q, \ y\in(0,\tilde{c}\gamma],
\end{equation} 
where $\tilde{c}:=\frac{1}{1+(k/\theta)^{1/(2-\alpha)}}$ and $\tilde{n}:=	\alpha[k(\gamma-\tilde{c}\gamma)^{\alpha-1}+\theta (\tilde{c}\gamma)^{\alpha-1}]$.

Similarly, for $q>0$ and $\theta\in\mathbb{R}$, let $I_2(q;\theta)\in(\gamma,\infty)$ be the unique solution to the equation
\begin{equation}
	\alpha[(y-\gamma)^{\alpha-1}+\theta y^{\alpha-1}]=q, \ y\in(\gamma,\infty).
\label{eq:I2_eq}
\end{equation} 

\begin{rem}
With simple calculus, it is easy to verify that
\begin{align*}
	(0,\tilde{c}\gamma)\ni y\mapsto \alpha[k(\gamma-y)^{\alpha-1}+\theta y^{\alpha-1}]\in(\tilde{n},\infty)
\end{align*}
is a strictly decreasing bijection when $\theta>0$, and
\begin{align}
	(\gamma,\infty)\ni y\mapsto	\alpha[(y-\gamma)^{\alpha-1}+\theta y^{\alpha-1}]\in(0,\infty)
\label{eq:I1_map}
\end{align}
is a strictly decreasing bijection when $\theta\in[-1,\infty)$. If $\theta<-1$, then the map in \eqref{eq:I1_map} is not monotonic in $y$ but it is not hard to show that $\alpha[(y-\gamma)^{\alpha-1}+\theta y^{\alpha-1}]=q$ still admits a unique solution on $y>\gamma$ for any $q\geq 0$ (including $q=0$). Hence, although unnecessary, we can extend the domain of $I_2(q;\theta)$ to $(q,\theta)\in\{[0,\infty)\times (-\infty,-1)\}\cup\{ (0,\infty)\times [-1,\infty)\}$ to cater the possibility that $I_2(0;\theta)$ is well-defined when $\theta<-1$. It is also not hard to verify that $I_i(q;\theta)$ is non-increasing in $q$ and non-decreasing in $\theta$ for $i\in\{1,2\}$. Finally, $I_1$ and $I_2$ are also jointly continuous in $(q,\theta)$. See Corollary \ref{cor:cont_fundamental_quantities} in Appendix \ref{app:H_cont}.
\end{rem}

\begin{lemm}[Proposition EC.1 of \cite{TsZh21}]
	
	Recall the definitions of $\{m_i\}_{i\in\{1,2\}}$, $\{c_i\}_{i\in\{1,2,3,4\}}$ and $\{I_i(q)\}_{i\in\{1,2\}}$ (where we have suppressed the dependence of $I_1$, $I_2$, $c_i$ and $m_i$ on $\theta$). $Y^*:=y(Z_\tau)$ is an optimizer to problem \eqref{eq:axu}, where the function $y(\cdot)=y(\cdot;\theta)$ is defined as follows depending on the value of $\theta$:
	
	\begin{enumerate}	
	
	\item If $\theta\in(0,\infty)$, then
	\begin{equation}
		y(z)=y_{\lambda^*}(z)=I_1(\lambda^* z)\I_{\{\lambda^*z>m_1\}}+I_2(\lambda^* z)\I_{\{\lambda^*z\leq m_1\}}
	\label{eq:maximizer_case1}
	\end{equation}
	and we have $\bP(Y^*\in (0,c_1\gamma)\cup(c_2\gamma,\infty))=1$.
	
	\item If $\theta \in[-1,0]\cap (\underline{\theta},0]$, or $\theta \in(\underline{\theta}, -1)$ and $|\theta|^{-\frac{1}{1-\alpha}}> 1-\gamma e^{-r\tau}$, then
	\begin{equation*}
		y(z)=y_{\lambda^*}(z):=I_2(\lambda^* z)\I_{\{\lambda^*z\leq m_2\}}
	\end{equation*}
and we have
\begin{equation*}
	\begin{dcases}
		\bP(Y^*\in \{0\}\cup(c_3\gamma,\infty))=1,& \theta\in[-1,0];\\
		\bP(Y^*\in \{0\}\cup(c_3\gamma,c_4\gamma))=1,& \theta \in(\underline{\theta}, -1) \text{ and } |\theta|^{-\frac{1}{1-\alpha}}> 1-\gamma e^{-r\tau}.
	\end{dcases}
\end{equation*}
	
	\item If $\theta \in (-\infty, \underline{\theta})$, then $y(z)=0$ and $Y^*=0$.
	
	\item If
	\begin{align}
		\theta \in(\underline{\theta}, -1) \ \text{ and } \ |\theta|^{-\frac{1}{1-\alpha}}\leq  1-\gamma e^{-r\tau},
		\label{eq:theta_impossible}
	\end{align}
	then $y(z)=\frac{\gamma}{1-|\theta|^{-1/(1-\alpha)}}$ and $Y^*=\frac{\gamma}{1-|\theta|^{-1/(1-\alpha)}}$.

	
\end{enumerate}
In Case 1 and 2, $\lambda^*>0$ is a constant defined as the unique solution to the equation in $\lambda>0$ given by $\bE[Z_\tau y_{\lambda}(Z_\tau)]=1$.

\label{lem:auxsol}
\end{lemm}

\begin{proof}
While the full proof can be found in \cite{TsZh21}, we will provide a brief sketch of the proof here as to introduce a few notations to be used in some of the subsequent results. We will suppress the argument $\theta$ in $F(y;\theta)$ and $\bar{F}(y;\theta)$ for brevity.

Suppose $\theta\in(0,\infty)$ is fixed and let $\bar{F}(y)$ be the smallest concave majorant of $F(y)$. Then
\begin{align*}
	\bar{F}(y)=
	\begin{cases}
		-k(\gamma-y)^{\alpha}+\theta y^{\alpha},& 0\leq y< c_1\gamma;\\
		-k\gamma^{\alpha}(1-c_1)^{\alpha}+\theta c_1^{\alpha}\gamma^{\alpha}+m_1(y-c_1\gamma ),& c_1\gamma\leq y\leq c_2\gamma;\\
		(y-\gamma)^{\alpha}+\theta  y^{\alpha},& y>c_2\gamma.
	\end{cases}
\end{align*}
The Legendre-Fenchel transformation of $\bar{F}(\cdot)$ is defined as $$J(q):=\sup_{y\geq 0}(\bar{F}(y)-qy),\qquad q>0,$$ and the corresponding maximizer is characterized by the set-valued function
\begin{align}
	f^*(q):=\argmax_{y\geq 0} (\bar{F}_{\xi}(y)-qy)=
	\begin{cases}
		\{I_2(q)\},& 0<q<m_1 ;\\
		[I_1(q),I_2(q)],& q=m_1;\\
		\{I_1(q)\},& q> m_1.
	\end{cases}
\label{eq:LF_maximizer}
\end{align}
For $\lambda>0$, define $y_{\lambda}(z)$ via
\begin{align*}
	y_{\lambda}(z)=I_1(\lambda z)\I_{\{\lambda z>m_1\}}+I_2(\lambda z)\I_{\{\lambda z\leq m_1\}}.
\end{align*}
Note that $y_{\lambda}(Z_\tau)=y_{\lambda}(Z_\tau(\omega))\in f^*(\lambda Z_{\tau}(\omega))$ for all $\omega\in\Omega$.

On the one hand, by monotone convergence theorem and the facts that $I_1\leq \tilde{c}\gamma<\gamma<I_2$ and $I_1,I_2$ are strictly deceasing in $q$ with $I_2(0+)=+\infty$ and $I_1(+\infty)=0$, we have $\zeta:\lambda\mapsto \mathbb{E}[Z_{\tau}y_{\lambda}(Z_\tau)]$ being strictly decreasing  with $\zeta(0+)=+\infty$ and $\zeta(+\infty)=0$. On the other hand, $Z_\tau$ being atomless together with continuity of $I_1,I_2$ in $q$ suggest $\zeta(\lambda)$ is continuous due to dominated convergence theorem.
Then there exists a unique $\lambda^*>0$ such that $\zeta(\lambda^*)=1$ and in turn $Y^*$ is admissible. For any  $Y\in\mathcal{Y}$, we have
\begin{align*}
	\mathbb{E}[F(Y)-\lambda^* Z_{\tau}Y]\leq \mathbb{E}[\bar{F}(Y)-\lambda^* (Z_{\tau}Y)] \leq \mathbb{E}[J(\lambda^* Z_{\tau})]=\mathbb{E}[\bar{F}(y_{\lambda^*}(Z_\tau))-\lambda^* Z_{\tau} y_{\lambda^*}(Z_\tau)]
\end{align*}
and thus
\begin{align*}
	\mathbb{E}[F(Y)]\leq \mathbb{E}[\bar{F}(y_{\lambda^*}(Z_\tau))-\lambda^* Z_{\tau} y_{\lambda^*}(Z_\tau)]+\lambda^* \mathbb{E}[Z_{\tau}Y]
	&\leq \mathbb{E}[\bar{F}(y_{\lambda^*}(Z_\tau))]-\lambda^*(\mathbb{E}[Z_{\tau} y_{\lambda^*}(Z_\tau)]-1)\\
	&=\mathbb{E}[\bar{F}(y_{\lambda^*}(Z_\tau))].
\end{align*}
Therefore $\sup_{Y\in\mathcal{Y}} \mathbb{E}[F(Y)]\leq \mathbb{E}[\bar{F}(y_{\lambda^*}(Z_\tau))]$. One can check that $I_1:(m_1,\infty)\to(0,c_1\gamma)$ and $I_2:(0,m_1)\to(c_2\gamma,\infty)$ are both bijections. Then since $Z_\tau$ is atomless with support on $[0,\infty)$, the support of $y_{\lambda^*}(Z_\tau)$ is $(0,c_1\gamma)\cup(c_2\gamma,\infty)\subseteq\{y\geq 0: F(y)=\bar{F}(y)\}$. Hence $\mathbb{E}[\bar{F}(y_{\lambda^*}(Z\tau))]=\mathbb{E}[F(y_{\lambda^*}(Z\tau))]$ and $Y^*$ must be an optimizer to problem \eqref{eq:axu}.

Case 2 can be handled similarly. The result of Case 3 follows trivially upon checking that the unique global maximum of $F(y)$ is attained at $y=0$, and $Y^*\equiv 0$ is clearly admissible. Likewise, in Case 4 the unique global maximum of $F(y)$ is attained at $y=\frac{\gamma}{1-|\theta|^{-1/(1-\alpha)}}$, and $Y^*=\frac{\gamma}{1-|\theta|^{-1/(1-\alpha)}}$ is admissible under the stated conditions on the parameters.
\end{proof}

For practical purpose, there is no need to consider Case 4 in Lemma \ref{lem:auxsol} because the choice of $\theta$ will eventually be endogenized when we study the original periodic portfolio selection problem, and the relevant value of $\theta$ will never have its range described by \eqref{eq:theta_impossible}. See Remark \ref{rem:impossible} as well. 

Careful readers may notice that we have avoided the corner case of $\theta=\underline{\theta}$ in Lemma \ref{lem:auxsol}. This corner case carries some important theoretical implications which we will separately discuss via the following proposition.

\begin{prop}
If $\theta\neq \underline{\theta}$, then the optimizer to problem \eqref{eq:axu} is unique up to a $\bP$-null set, i.e. $\bP(Y_1^*=Y_2^*)=1$ if $Y^*_i\in\argmax_{Y\in\cY} \bE[F(Y;\theta)]$ for $i\in\{1,2\}$. If $\theta=\underline{\theta}$ and $k>0$, then $Y^*\in \cY$ is optimal to problem \eqref{eq:axu} if and only if $\bP(Y^*\in\{0,\gamma(1+k^{-\frac{1}{1-\alpha}})\})=1$.
\label{prop:uniqueness}
\end{prop}

\begin{proof}
As before, we will suppress the argument $\theta$ in $F(y;\theta)$ and $\bar{F}(y;\theta)$. Suppose we are in Case 1 of Lemma \ref{lem:auxsol} such that $\theta\in(0,\infty)$. We already know from Lemma \ref{lem:auxsol} that $Y^*=y(Z_\tau)$ is a maximizer to problem \eqref{eq:axu} where $y(\cdot)$ is defined in \eqref{eq:maximizer_case1}. To show that this maximizer is unique, let $\tilde{Y}\in\mathcal{Y}$ be another optimal random variable which attains the same value as $Y^*$. Let $$E:=\{\omega\in\Omega| \tilde{Y}(\omega)\notin f^*(\lambda^*Z_{\tau}(\omega))\}$$ where $f^*$ is defined in \eqref{eq:LF_maximizer}. Suppose $\mathbb{P}(E)>0$. Then by definition of $\tilde{Y}$,
		\begin{align*}
			\mathbb{E}[F(\tilde{Y})]&=\mathbb{E}[F(y_{\lambda^*}(Z_\tau))]=\mathbb{E}[\bar{F}(y_{\lambda^*}(Z_\tau))]
			=\mathbb{E}[\bar{F}(y_{\lambda^*}(Z_\tau))-\lambda^* Z_{\tau} y_{\lambda^*}(Z_\tau)]+\lambda^*
		\end{align*}
		and hence
		\begin{align*}
			\mathbb{E}[F(\tilde{Y})]-\lambda^*=\mathbb{E}[\bar{F}(y_{\lambda^*}(Z_\tau))-\lambda^* Z_{\tau} y_{\lambda^*}(Z_\tau)]=\mathbb{E}[J(\lambda^* Z_{\tau})].
		\end{align*}
		Then since $\mathbb{E}[Z_{\tau }\tilde{Y}]\leq 1$ and $\lambda^*>0$, we have
		\begin{align*}
			\mathbb{E}[F(\tilde{Y})-\lambda^* Z_{\tau} \tilde{Y}]\geq \mathbb{E}[J(\lambda^* Z_{\tau})].
		\end{align*}
		But we also have
		\begin{align*}
			\mathbb{E}[F(\tilde{Y})-\lambda^* Z_{\tau} \tilde{Y}]&\leq \mathbb{E}[\bar{F}(\tilde{Y})-\lambda^* Z_{\tau} \tilde{Y}]\\
			&\leq \int_{E} (\bar{F}(\tilde{Y}(\omega))-\lambda^* Z_{\tau}(\omega) \tilde{Y}(\omega))d\mathbb{P}+\int_{E^{c}}(\bar{F}(\tilde{Y}(\omega))-\lambda^* Z_{\tau}(\omega) \tilde{Y}(\omega))d\mathbb{P}\\
			&<\int_{E} J(\lambda^* Z_{\tau}(\omega))d\mathbb{P}+\int_{E^c} J(\lambda^* Z_{\tau}(\omega))d\mathbb{P}=\mathbb{E}[J(\lambda^* Z_{\tau})]
		\end{align*}
		leading to a contradiction. Here, we used the fact that $J(q)\geq \bar{F}(y)-qy$ for all $y\geq 0$, and strict inequality holds whenever $y\notin f^*(q)$. Hence we must conclude $\mathbb{P}(E)=0$. Finally, since $Z_{\tau}$ is atomless such that $\mathbb{P}(Z_{\tau}=m_1/\lambda^*)=0$, we deduce
		\begin{align*}
			1=\mathbb{P}(E^c)&=\mathbb{P}(\{\omega\in\Omega| \tilde{Y}(\omega)\in f^*(\lambda^*Z_{\tau}(\omega))\})\\
			&=\mathbb{P}(\{\omega\in \Omega| \tilde{Y}(\omega)
			=I_1(\lambda^* Z_{\tau}(\omega)){\mathbbm 1}_{(\lambda^* Z_{\tau}(\omega)>m_1)} +I_2(\lambda^* Z_{\tau}(\omega)){\mathbbm 1}_{(\lambda^* Z_{\tau}(\omega)<m_1)}\})=\bP(\tilde{Y}=Y^*).
		\end{align*}
		Thus we have $\tilde{Y}=Y^*$ almost surely.
		
	Case 2 can be analyzed similarly. Case 3 and 4 are also easy to be handled where $F(y)$ attains a unique global maximizer at either $y=0$ or $y=\frac{\gamma}{1-|\theta|^{-1/(1-\alpha)}}$, and any feasible optimizer should put a probability mass of unity at the unique global maximizer. 
	
	When $\theta=\underline{\theta}$ and $k>0$, it is straightforward to check that the global maximum of $F(y)$ is attained at both $y=0$ and $y=\gamma(1+k^{-\frac{1}{1-\alpha}})$, i.e. $F(0)=F(\gamma(1+k^{-\frac{1}{1-\alpha}}))=-k\gamma^{\alpha}>F(y)$ on $y\notin\{0,\gamma(1+k^{-\frac{1}{1-\alpha}})\}$. Then clearly 
	$$\sup_{Y\in\mathcal{Y}}\mathbb{E}[F(Y)]\leq -k\gamma^{\alpha}=\bE[F(Y^*)]$$ for any $Y^*\in\cY$ such that $\bP(Y^*\in\{0,\gamma(1+k^{-\frac{1}{1-\alpha}})\})=1$ and hence such $Y^*$ must be an optimizer. Conversely, suppose $Y^*\in\argmax_{Y\in\cY}\bE[F(Y)]$ but $p:=\mathbb{P}\left(Y^*\in\left\{0,\gamma(1+k^{-\frac{1}{1-\alpha}})\right\}\right)<1$. Then  
		\begin{align*}
			\mathbb{E}[F(Y^*)]&=p (-k\gamma^{\alpha})+(1-p)\mathbb{E}\left[F(Y^*) \Bigl | Y^*\notin \left\{0,\gamma(1+k^{-\frac{1}{1-\alpha}})\right\}\right]
			<p(-k\gamma^{\alpha})+(1-p)(-k\gamma^{\alpha})=-k\gamma^{\alpha}
		\end{align*}
		leading to a contradiction and hence we must have $\mathbb{P}\left(Y^*\in\left\{0,\gamma(1+k^{-\frac{1}{1-\alpha}})\right\}\right)=1$.
\end{proof}

The (lack of) uniqueness of the optimal portfolio is not addressed in \cite{TsZh21}. From an optimization point of view, uniqueness is perhaps not a very economically important issue for as long as one can characterize at least one strategy that can achieve the optimal value. In the corner case of $\theta=\underline{\theta}$, \cite{TsZh21} reports $Y^*=0$ as an optimal solution. But more generally, Proposition \ref{prop:uniqueness} suggests that any feasible digital option with payout $\gamma(1+k^{-\frac{1}{1-\alpha}})$ can also be considered as an optimizer. Unlike \cite{TsZh21}, careful analysis of this corner case $\theta=\underline{\theta}$ is actually required in our current problem since it will influence the characterization of the equilibrium strategy pursued by a sophisticated agent.

\begin{rem}
It is not necessary to analyze the case of $k=0$ and $\theta=\underline{\theta}=-1$ in Proposition \ref{prop:uniqueness} because the optimally endogenized value of $\theta$ is always strictly positive when $k=0$. See Remark \ref{rem:impossible}. 
\end{rem}

The map $\theta\mapsto e^{-\delta\tau}\Phi(\theta)$ defined in \eqref{eq:axu} is analogous to a discrete-time ``Bellman operator''. \cite{TsZh21} show that this map admits a unique fixed point which is then used to construct a solution to the portfolio optimization problem (under $\beta=1$). In what follows, we prove a more general result.

\begin{prop}
	For any $\kappa\in[0,1]$, the map $\theta\mapsto e^{-\delta\tau}\Phi(\kappa\theta)$ is a contraction on $(\bR,||\cdot||)$, in particular there exists a unique $\theta^*(\kappa)\in\bR$ such that $\theta^*(\kappa)=e^{-\delta\tau}\Phi(\kappa\theta^*(\kappa))$. The map $\kappa\mapsto\theta^*(\kappa)$ is continuous, and it holds that, for any $\kappa_1,\kappa_2\in[0,1]$,
	\begin{equation}
		\theta^*(\kappa_2)\leq\frac{1-\kappa_1e^{-\delta\tau}\bE[(Y^*(\kappa_2))^\alpha]}{1-\kappa_2e^{-\delta\tau}\bE[(Y^*(\kappa_2))^\alpha]}\theta^*(\kappa_1),
	\label{eq:theta_kappa_ineq}
	\end{equation}
	where $Y^*(\kappa)\in\cY$ is a maximizer for the problem $e^{-\delta\tau}\Phi(\kappa\theta^*(\kappa))$.
	Furthermore:
	\begin{enumerate}
		\item
		If $\theta^*(0)=0$, then $\theta^*(\kappa)=\theta^*(0)=0$ for any $\kappa\in[0,1]$;
		\item
		If $\theta^*(0)>0$, then $\kappa\mapsto\theta^*(\kappa)$ is strictly increasing. Equality holds in \eqref{eq:theta_kappa_ineq} if and only if $\kappa_1=\kappa_2$;
		\item
		If $\theta^*(0)<0$ and $\theta^*(1)\geq\underline{\theta}$, then $\kappa\mapsto\theta^*(\kappa)$ is strictly decreasing. Equality holds in \eqref{eq:theta_kappa_ineq} if and only if $\kappa_1=\kappa_2$;
		\item
		If $\theta^*(0)<0$ and $\theta^*(1)<\underline{\theta}$, then $\kappa\mapsto\theta^*(\kappa)$ is strictly decreasing on $[0,\underline{\kappa}]$, and $\theta^*(\kappa)=\theta^*(1)$ for any $\kappa\in[\underline{\kappa},1]$, where
		\begin{equation*}
			\underline{\kappa}:=\inf\{\kappa\in[0,1]\,|\,\kappa\theta^*(\kappa)=\underline{\theta}\}\in(0,1).
		\end{equation*}
		Equality holds in \eqref{eq:theta_kappa_ineq} for some $Y^*(\kappa_2)$ if and only if $\kappa_1=\kappa_2<\underline{\kappa}$ or $\kappa_1,\kappa_2\geq \underline{\kappa}$.
	\end{enumerate}
	\label{lemm_revise}
\end{prop}

\begin{proof}
	The contraction property of $\theta\mapsto e^{-\delta\tau}\Phi(\theta)$ is shown in \cite[Proposition EC.2]{TsZh21}. This immediately implies that, for any $\kappa\in[0,1]$, the map $\theta\mapsto e^{-\delta\tau}\Phi(\kappa\theta)$ is contractive, and hence has a unique fixed point $\theta^*(\kappa)\in\bR$.
	
	Note that $0\leq e^{-\delta\tau}\bE[Y^\alpha]\leq e^{-(\delta-h)\tau}<1$ for any $Y\in\cY$. Let $\kappa_1,\kappa_2\in[0,1]$ be fixed. Then
	\begin{align}
		\theta^*(\kappa_2)-\theta^*(\kappa_1)&=e^{-\delta\tau}\{\Phi(\kappa_2\theta^*(\kappa_2))-\Phi(\kappa_1\theta^*(\kappa_1))\}\nonumber \\
		&\leq e^{-\delta\tau}\{\kappa_2\theta^*(\kappa_2)-\kappa_1\theta^*(\kappa_1)\}\bE[(Y^*(\kappa_2))^\alpha],
		\label{eq:theta_kappa_ineq2}
	\end{align}
	thus \eqref{eq:theta_kappa_ineq} holds. Since $\kappa_1$ and $\kappa_2$ are arbitrary, upon swapping $\kappa_1$ and $\kappa_2$ in \eqref{eq:theta_kappa_ineq} we can also deduce
\begin{align*}
	\theta^*(\kappa_1)\leq \frac{1-\kappa_2e^{-\delta\tau}\bE[(Y^*(\kappa_1))^\alpha]}{1-\kappa_1e^{-\delta\tau}\bE[(Y^*(\kappa_1))^\alpha]}
	\theta^*(\kappa_2)
\end{align*}
for any $\kappa_1,\kappa_2$. Therefore,
	\begin{equation*}
		\frac{1-\kappa_1e^{-\delta\tau}\bE[(Y^*(\kappa_1))^\alpha]}{1-\kappa_2e^{-\delta\tau}\bE[(Y^*(\kappa_1))^\alpha]}\theta^*(\kappa_1)\leq\theta^*(\kappa_2)\leq\frac{1-\kappa_1e^{-\delta\tau}\bE[(Y^*(\kappa_2))^\alpha]}{1-\kappa_2e^{-\delta\tau}\bE[(Y^*(\kappa_2))^\alpha]}\theta^*(\kappa_1).
	\end{equation*}
	This implies that the sign of $\theta^*(\kappa)$ does not depend on $\kappa\in[0,1]$. Also, from the above estimate, together with $e^{-\delta\tau}\sup_{Y\in\cY}\bE[Y^\alpha]=e^{-(\delta-h)\tau}<1$, we can easily show that
	\begin{equation*}
		|\theta^*(\kappa_1)-\theta^*(\kappa_2)|\leq\frac{e^{-(\delta-h)\tau}}{1-e^{-(\delta-h)\tau}}|\theta^*(\kappa_1)||\kappa_1-\kappa_2|
	\end{equation*}
	for any $\kappa_1,\kappa_2\in[0,1]$. Thus, $\kappa\mapsto\theta^*(\kappa)$ is continuous.
	
	Suppose that $\theta^*(0)>0$ (which is equivalent to $\Phi(0)>0$). Let $0\leq\kappa_1<\kappa_2\leq1$. Noting that $\theta^*(\kappa_1)>0$ and $\bE[(Y^*(\kappa_1))^\alpha]>0$, we see that
	\begin{equation*}
		\theta^*(\kappa_1)<\frac{1-\kappa_1e^{-\delta\tau}\bE[(Y^*(\kappa_1))^\alpha]}{1-\kappa_2e^{-\delta\tau}\bE[(Y^*(\kappa_1))^\alpha]}\theta^*(\kappa_1)\leq\theta^*(\kappa_2).
	\end{equation*}
	Hence, $\kappa\mapsto\theta^*(\kappa)$ is strictly increasing.
	
	Suppose that $\theta^*(0)<0$ (which is equivalent to $\Phi(0)<0$). Let $0\leq\kappa_1<\kappa_2\leq1$. Noting that $\theta^*(\kappa_1)<0$, we see that
	\begin{equation*}
		\frac{1-\kappa_1e^{-\delta\tau}\bE[(Y^*(\kappa_1))^\alpha]}{1-\kappa_2e^{-\delta\tau}\bE[(Y^*(\kappa_1))^\alpha]}\theta^*(\kappa_1)\leq\theta^*(\kappa_2)\leq\frac{1-\kappa_1e^{-\delta\tau}\bE[(Y^*(\kappa_2))^\alpha]}{1-\kappa_2e^{-\delta\tau}\bE[(Y^*(\kappa_2))^\alpha]}\theta^*(\kappa_1)\leq\theta^*(\kappa_1).
	\end{equation*}
	Hence, $\kappa\mapsto\theta^*(\kappa)$, and in turn $\kappa\mapsto\kappa\theta^*(\kappa)$, are non-increasing. Furthermore, if $\kappa_2\theta^*(\kappa_2)\geq\underline{\theta}$, then by Lemma \ref{lem:auxsol} and Proposition \ref{prop:uniqueness} we can choose $Y^*(\kappa_2)$ such that $Y^*(\kappa_2)>0$ with positive probability, and thus the third inequality above is strict, i.e. $\theta^*(\kappa_1)>\theta^*(\kappa_2)$. On the other hand, if $\kappa_1\theta^*(\kappa_1)\leq\underline{\theta}$, then we can take $Y^*(\kappa_1)$ such that $Y^*(\kappa_1)=0$ a.s., and thus $\theta^*(\kappa_1)\leq\theta^*(\kappa_2)$, which implies that $\theta^*(\kappa_1)=\theta^*(\kappa_2)$. In particular, if it holds that $\theta^*(0)<0$ and $\theta^*(1)\geq\underline{\theta}$, then $0>\kappa_2\theta^*(\kappa_2)\geq \theta^*(1)\geq  \underline{\theta}$, and hence $\kappa\mapsto\theta^*(\kappa)$ is strictly decreasing. Otherwise, if we have $\theta^*(0)<0$ and $\theta^*(1)<\underline{\theta}$, then for
	\begin{equation*}
		\underline{\kappa}:=\inf\{\kappa\in[0,1]\,|\,\kappa\theta^*(\kappa)=\underline{\theta}\}\in(0,1),
	\end{equation*}
	the function $\kappa\mapsto\theta^*(\kappa)$ is strictly decreasing on $[0,\underline{\kappa}]$, and $\theta^*(\kappa)=\theta^*(1)$ for any $\kappa\in[\underline{\kappa},1]$.
	
To complete the proof, we now verify the necessary and sufficient conditions for equality to hold in \eqref{eq:theta_kappa_ineq} in each of cases (2) to (4). Here, recall that $Y^*(\kappa_2)$ is a maximizer of $\Phi(\kappa_2\theta^*(\kappa_2))$. Thus equality holds in \eqref{eq:theta_kappa_ineq2} (and in turn \eqref{eq:theta_kappa_ineq}) if and only if $Y^*(\kappa_2)$ is also a maximizer of $\Phi(\kappa_1\theta^*(\kappa_1))$. 

Suppose we are in case (2) or case (3), and equality holds in  \eqref{eq:theta_kappa_ineq} such that $Y^*(\kappa_2)$ is a maximizer of $\Phi(\kappa_1\theta^*(\kappa_1))$. The (unique) characterization of the maximizer in Lemma \ref{lem:auxsol} and Proposition \ref{prop:uniqueness} implies that $\kappa_1\theta^*(\kappa_1)=\kappa_2\theta^*(\kappa_2)\geq \underline{\theta}$. But $\theta^*(\kappa)$ is the fixed point of $\kappa\mapsto e^{-\delta\tau}\Phi(\kappa\theta)$. Therefore,
\begin{align*}
	\theta^*(\kappa_1)=e^{-\delta\tau}\Phi(\kappa_1\theta^*(\kappa_1))=e^{-\delta\tau}\Phi(\kappa_2\theta^*(\kappa_2))=\theta^*(\kappa_2).
\end{align*}
Since $\kappa\mapsto\theta^*(\kappa)$ is strictly increasing (resp. decreasing) in case (2) (resp. case (3)), we must have $\kappa_1=\kappa_2$. The reverse implication that $\kappa_1=\kappa_2$ implies equality holds in \eqref{eq:theta_kappa_ineq} is trivial.

Suppose we are in case (4), and there exists some $Y^*(\kappa_2)$ which is a maximizer of $\Phi(\kappa_1\theta^*(\kappa_1))$. Consider two sub-cases: (a) If $\kappa_1<\underline{\kappa}\iff \kappa_1\theta^*(\kappa_1)>\underline{\theta}$, then by the property that $\kappa\mapsto\theta^*(\kappa)$ is strictly decreasing on $[0,\underline{\kappa}]$ and the same arguments used for case (2) and (3) above, we conclude $\kappa_1=\kappa_2$; (b) If $\kappa_1 \geq \underline{\kappa}\iff \kappa_1\theta^*(\kappa_1)\leq \underline{\theta}$, then by the characterization that $Y^*(\kappa_2)$, the maximize of $\Phi(\kappa_1\theta^*(\kappa_1))$, is a binary random variable or a constant of zero in such case, we have $\kappa_2\theta^*(\kappa_2)\leq \underline{\theta}$. This implies $\kappa_2\geq \underline{\kappa}$. Conversely, if $\kappa_1,\kappa_2\geq \underline{\kappa}$, then we can choose $Y^*(\kappa_2)=0$ as the maximizer of $\Phi(\kappa_2\theta^*(\kappa_2))$, which is also the maximizer of $\Phi(\kappa_1\theta^*(\kappa_1))$. Under this choice, \eqref{eq:theta_kappa_ineq} leads to $\theta^*(\kappa_1)= \theta^*(\kappa_2)$ thanks to the arbitrariness of $\kappa_1,\kappa_2$.
\end{proof}

We will soon see that $\theta^*(\kappa)$ can help us characterize the value functions of different types of agent (see in particular Remark \ref{rem:val_theta}). To numerically solve for $\theta^*(\kappa)$ (say under a fixed $\kappa$), one can generate a sequence $\{\theta_n(\kappa)\}_{n\in\bN_0}$ iteratively via $\theta_{n+1}(\kappa)=e^{-\delta \tau}\Phi(\kappa \theta_{n}(\kappa))$ for some arbitrary initial guess $\theta_{0}(\kappa)$. Each step of iteration requires us to solve an optimization problem in form of \eqref{eq:axu}. This can be repeated many times until some suitable error tolerance criterion is met.

\begin{rem}
By \cite[Proposition 2]{TsZh21}, if $\gamma\leq e^{r\tau}$ then necessarily $\theta^*(1)>0$. Then in turn $\theta^*(\kappa)>0$ for all $\kappa\in[0,1]$ and $\Phi(0)>0$ by Proposition \ref{lemm_revise}. Consequently, $\theta^*(\kappa)$ cannot have its range described by \eqref{eq:theta_impossible}. Similarly, if $k=0$ then $U(Y-\gamma)$ is non-negative for any $Y\in\cY$ and it is easy to conclude as well that $\theta^*(\kappa)>0$ for all $\kappa\in[0,1]$ and $\Phi(0)>0$. Moreover, in either case of $\gamma\leq e^{r\tau}$ or $k=0$, all numerical estimates generated by the iteration $\theta_{n+1}(\kappa)=e^{-\delta \tau}\Phi(\kappa \theta_{n}(\kappa))$ are guaranteed to be strictly positive for as long as one chooses $\theta_{0}(\kappa)>0$, and then the optimizer of problem $\Phi(\kappa \theta_{n}(\kappa))$ can be identified solely by part 1 of Lemma \ref{lem:auxsol}.
\label{rem:impossible}
\end{rem}

We close this section by presenting an important result concerning the fixed point of a set-valued map. As we will see in Section \ref{sect:mainresults}, this property is crucial behind the characterization of an equilibrium strategy for a sophisticated agent. The complete proof of the following theorem is long and technical, which is thus deferred to Appendix \ref{app:H_cont}.

\begin{theo}
	Define a set-valued map $G:[0,e^{h\tau}]\mapsto 2^{[0,e^{h\tau}]}$ via
	\begin{align}
		G(\xi):=\left\{\bE[Y^\alpha]\: \Bigl | \: Y\in\argmax_{Y\in\cY}\bE\left[F\left(Y;\: \frac{\beta}{1-(1-\beta)e^{-\delta\tau}\xi}\theta^*\left(\frac{\beta}{1-(1-\beta)e^{-\delta\tau}\xi}\right)\right)\right]\right\}.
		\label{eq:g_main}
 	\end{align}
	Then $G$ admits a fixed point $\hat{\xi}\in[0,e^{h\tau}]$, i.e. $\hat{\xi}\in G(\hat{\xi})$. Specifically:
	\begin{enumerate}
		\item If $\theta^*(0)>0$, $G(\xi)$ has at least one fixed point.
		\item If $\theta^*(0)=0$, $G(\xi)$ has exactly one fixed point given by $\hat{\xi}=\left\{\mathbb{E}(Y^{\alpha})|Y\in\argmax_{Y\in\mathcal{Y}}\mathbb{E}\left[U(Y-\gamma)\right]\right\}$.
		\item If $\theta^*(0)<0$, $G(\xi)$ has exactly one fixed point. 
	\end{enumerate}
	\label{thm:G_fp}
\end{theo}

\section{Solutions to the main problems}
\label{sect:mainresults}

We are now ready to state the main results of this paper characterizing the optimal portfolio for each type of agent. The associated economic intuitions will be more thoroughly discussed in Section \ref{sect:discuss}.


\begin{theo}[Optimal pre-committing strategy]
Recall $\Phi(\cdot)$ and $\theta^*(\cdot)$ defined in Proposition \ref{lemm_revise}. There exists a unique pair $(A_\pre,A_\expo)\in\bR^2$ solving system \eqref{eq:precomm_sys}, where $A_\expo=\theta^*(1)$ and $A_\pre=e^{-\delta\tau}\Phi(\theta^*(1))$. The optimal (with reference time $T_n$) pre-committed value function is given by
\begin{equation*}
	V_\pre(x):= \sup_{\pi\in\Pi^{(n)}}J_n(\pi;x)=A_\pre x^\alpha.
\end{equation*}
Moreover, there exists an optimal pre-committing strategy $\pi^{\pre,(n)}$ with $V_\pre(x)=J_n(\pi^{\pre,(n)},x)$ such that the corresponding portfolio process $\hat{X}:=X^{n,x,\pi^{\pre,(n)}}$ satisfies
\begin{equation}
	\hat{X}_{T_{i+1}}=
	\begin{dcases}
		 y\left(\frac{Z_{T_{i+1}}}{Z_{T_i}};\beta A_{\expo}\right) \hat{X}_{T_{i}},& i=n;\\
		y\left(\frac{Z_{T_{i+1}}}{Z_{T_i}}; A_{\expo}\right)\hat{X}_{T_{i}},& i\in\{n+1,n+2,...\},
	\end{dcases}
\label{eq:precomm_port}
\end{equation}
where the function $y(z;\theta)=y(z)$ is defined in Lemma \ref{lem:auxsol} and $Z=(Z_t)_{t\geq 0}$ is the pricing kernel of the Black-Scholes economy given by \eqref{eq:bs_kernel}.
\label{thm:precomm}
\end{theo}

\begin{proof}
	$\theta\mapsto e^{-\delta \tau}\Phi(\theta)$ is a contraction by Proposition \ref{lemm_revise} and hence there exists a unique $A_\expo$ such that $A_\expo= e^{-\delta \tau}\Phi(A_\expo)$. Then $(A_\pre:=e^{-\delta\tau}\Phi(\beta A_\expo), A_\expo)$ is the unique solution to system \eqref{eq:precomm_sys}.
	
	The rest of the proof is very similar to that of Theorem 1 in \cite{TsZh21}. Without loss of generality, we work with the optimal time-zero pre-committed strategy.
	For any admissible portfolio process $X=(X^{0,x,\pi}_t)_{t\geq 0}$ with arbitrary $\pi\in\Pi^{(0)}$ (we will suppress the superscripts in $X$ for brevity), define a discrete-time stochastic process $M=(M_n)_{n\in\bN_0}$ via
	\begin{align}
		M_n:=
		\begin{cases}
			A_\pre x^\alpha,& n=0;\\
			\sum_{i=1}^n D(T_i)U(X_{T_i}-\gamma X_{T_{i-1}})+\beta A_\expo e^{-\delta T_n} X^\alpha_{T_n}, & n\geq 1.
		\end{cases}
		\label{eq:M}
	\end{align}
	Then for $n\geq 1$,
	\begin{align*}
		M_{n+1}=M_n+\beta e^{-\delta T_n}\left[e^{-\delta \tau}\left(U(X_{T_{n+1}}-\gamma X_{T_n})+A_\expo X_{T_{n+1}}^\alpha\right)-A_\expo X^\alpha_{T_n}\right],
	\end{align*}
	and in turn
	\begin{align*}
		\bE[M_{n+1}|\cF_{T_n}]=M_n+\beta e^{-\delta T_n}X_{T_n}^\alpha\left[e^{-\delta\tau}\bE\left[U\left(Y_{n+1}-\gamma\right)+A_\expo Y_{n+1}^\alpha\Bigl | \cF_{T_n}\right]-A_\expo\right].
	\end{align*}
	Notice that $Y_{n+1}$ is $\cF^{T_n}_{T_{n+1}}$-measurable, and in turn $Y_{n+1}$ is independent of $\cF_{T_n}$ such that
	\begin{align*}
		\bE\left[U\left(Y_{n+1}-\gamma\right)+A_\expo Y_{n+1}^\alpha\Bigl | \cF_{T_n}\right]=\bE\left[U\left(Y_{n+1}-\gamma\right)+A_\expo Y_{n+1}^\alpha\right].
	\end{align*}
	Since $ZX$ is a non-negative supermartingale under any admissible portfolio strategy $\pi$, $\bE\left[\frac{Z_{T_{n+1}}}{Z_{T_n}}Y_{n+1}\Bigl| \cF_{T_n}\right]\leq 1$ and therefore
	\begin{align}
		e^{-\delta\tau}\bE\left[U\left(Y_{n+1}-\gamma\right)+A_\expo Y_{n+1}^\alpha\Bigl | \cF_{T_n}\right]\leq e^{-\delta\tau}\sup_{Y\in\cY}\bE\left[U\left(Y-\gamma\right)+A_\expo Y^\alpha\right]=A_\expo
		\label{eq:ineq1}
	\end{align}
	as $e^{-\delta\tau}\Phi(A_\expo)=A_\expo$ by design. We thus deduce $\bE[M_{n+1}|\cF_{T_n}]\leq M_n$ for $n\geq 1$.
	
	Meanwhile
	\begin{align*}
		M_1=M_0+e^{-\delta T_1}U(X_{T_1}-\gamma X_{T_0})+\beta A_\expo e^{-\delta T_1} X^\alpha_{T_1}-A_\pre x^\alpha
	\end{align*}
	and hence
	\begin{align}
		\bE[M_{1}]&=M_0+e^{-\delta T_1}x^\alpha\left[e^{-\delta \tau}\bE\left[U\left(Y_1-\gamma\right)+\beta A_\expo  Y_1^\alpha\right]-A_\pre \right]\nonumber\\
		&\leq M_0+e^{-\delta T_1}x^\alpha\left[e^{-\delta \tau} \sup_{Y\in\cY}\bE\left[U\left(Y-\gamma\right)+\beta A_\expo Y^\alpha\right]-A_\pre \right] \label{eq:ineq2}\\
		&=M_0+e^{-\delta T_1}x^\alpha (e^{-\delta\tau}\Phi(\beta A_{\expo})-A_\pre)=M_0\nonumber
	\end{align}
	by definition of $A_{\pre}$. Thus $M$ is a supermartingale with respect to $\cG$ where $\cG_n:=\cF_{T_n}$. Then
	\begin{align*}
		A_\pre x^\alpha=M_0\geq \bE[M_n]= \bE\left[\sum_{i=1}^n D(T_i)U(X_{T_i}-\gamma X_{T_{i-1}})+\beta A_\expo e^{-\delta T_n} X^\alpha_{T_n}\right]
	\end{align*}
	and in turn
	\begin{align*}
		\bE\left[\sum_{i=1}^n D(T_i)U(X_{T_i}-\gamma X_{T_{i-1}})\right]&\leq A_\pre x^\alpha-\beta A_\expo e^{-\delta T_n}\bE\left[X_{T_n}^\alpha\right]\\
		&\leq  A_\pre x^\alpha+\beta |A_\expo| e^{-(\delta-h) T_n}
	\end{align*}
	on recalling that $\bE\left[X_{T_n}^\alpha\right]\leq e^{h\tau}$ using the solution to a finite-horizon Merton problem as an estimate. Using assumption \eqref{eq:assump}, taking limit $n\to\infty$ and then supremum over $\pi\in\Pi^{(0)}$, we deduce
	\begin{align*}
		V_\pre(x):=\sup_{\pi\in\Pi^{(0)}}\bE\left[\sum_{i=1}^\infty D(T_i)U(X_{T_i}-\gamma X_{T_{i-1}})\right]\leq A_\pre x^\alpha.
	\end{align*}
	To show the reverse inequality that $V_\pre(x)\geq A_\pre x^{\alpha}$, one just needs to construct an admissible portfolio process which attains a value of $A_\pre x^{\alpha}$. Using the usual replication argument in a complete market, there exists some admissible $\pi^{\pre}\in\Pi^{(0)}$ such that $\hat{X}:=X^{0,x,\pi^{\pre}}$ satisfies \eqref{eq:precomm_port} at all $\{T_i\}_{i>n}$. Now we can define a process $\hat{M}$ in the same fashion as in \eqref{eq:M} except we replace $X$ by $\hat{X}$. Then using the fact that $y(Z_\tau;\theta)$ is an optimizer to problem \eqref{eq:axu}, one can show that $\hat{M}$ is indeed a $\cG$-martingale where the inequalities in \eqref{eq:ineq1} and \eqref{eq:ineq2} become equalities. One can then ultimately conclude
	\begin{align*}
		\sup_{\pi\in\Pi^{(0)}}\bE\left[\sum_{i=1}^\infty D(T_i)U(X_{T_i}-\gamma X_{T_{i-1}})\right]= A_\pre x^\alpha=\bE\left[\sum_{i=1}^\infty D(T_i)U(\hat{X}_{T_i}-\gamma \hat{X}_{T_{i-1}})\right].
	\end{align*}
\end{proof}

The ratio $Z_{T_{i+1}}/Z_{T_i}$ captures the change in the pricing kernel, which reflects how the state of the world changes within the period $[T_i,T_{i+1}]$. Moreover, $\{Z_{T_{i+1}}/Z_{T_i}\}_{i\in \bN_0}$ are iid with a common law identical to that of $Z_\tau$. A pre-committing agent thus trades in the way such that the periodic gross returns $\{\hat{X}_{T_{i+1}}/\hat{X}_{T_i}\}_{i}$ are independent across periods. Moreover, they target a gross return with risk profile described by $y(Z_\tau,\beta A_\expo)$ in the first period, and then in all subsequent periods they target a different payoff of $y(Z_\tau,A_\expo)$.

By construction, $A_{\expo}$ is the fixed point of $\theta\mapsto e^{-\delta \tau}\Phi(\theta)$ and $y(Z_\tau;A_\expo)\in \argmax \bE[F(Y;A_{\expo})]$. The pair $(A_\expo,y(Z_\tau,A_\expo))$ therefore represents exactly the value function of the periodic portfolio selection problem faced by an exponential discounter (i.e. an agent without present bias or equivalently $\beta=1$) and their corresponding optimal target gross return, which is the benchmark problem studied in \cite{TsZh21}. In our setup, this means the myopic pre-committing agent will invest in the same way as an exponential discounter and target a law of $y(Z_\tau,A_\expo)$ after the first period. This observation is not too surprising. The agent discounts future payoffs using the sequence of discount factors $\{e^{-\delta\tau},\beta e^{-2\delta\tau},\beta e^{-3\delta \tau},...\}$. As of today, they anticipate that they will behave just like an exponential discounter from the second period onward where then they will use the discount factors $\{e^{-2\delta\tau}, e^{-3\delta \tau}, e^{-4\delta \tau},...\}$, modulo the scaling factor $\beta$. Thus, as of today, they plan to trade like an exponential discounter starting from the second period. But due to their present bias over the short term outcome in the first period, they plan to deviate from the exponential discounter's strategy in the first period (only) which risk profile is characterized by $y(Z_\tau,\beta A_\expo)$.

But what if the agent cannot commit to a planned strategy? If the agent is instead naive (as defined in Definition \ref{def:naive}), they will keep reoptimizing and updating their strategy at the beginning of each period. The following corollary is a straightforward consequence of Theorem \ref{thm:precomm}.

\begin{cor}[Optimal naive strategy]
There exists an optimal naive strategy $\pi^{\text{naive}}$ such that the corresponding portfolio process $\hat{X}:=X^{0,x,\pi^{\text{naive}}}$ satisfies
\begin{equation*}
	\hat{X}_{T_{i+1}}=y\left(\frac{Z_{T_{i+1}}}{Z_{T_i}};\beta A_{\expo}\right) \hat{X}_{T_{i}}
\end{equation*}
for all $i\in\bN$, where the function $y(z;\theta)=y(z)$ is defined in Lemma \ref{lem:auxsol} and $Z=(Z_t)_{t\geq 0}$ is the pricing kernel of the Black-Scholes economy given by \eqref{eq:bs_kernel}.
\label{cor:naivesol}
\end{cor}

At time zero ($t=0$), a pre-committing agent plans to take risk of $y(Z_\tau;\beta A_\expo)$ over the first period and switch to a different form of risk of $y(Z_\tau; A_\expo)$ from the second period onward. However, if the agent turns out to be naive, then once they arrive at the beginning of the second period ($t=T_1$), they will reoptimize the strategy and conclude that the current best action is to take risk of $y(Z_\tau;\beta A_\expo)$ again in the second period and later switch to $y(Z_\tau; A_\expo)$ at the beginning of the third period ($t=T_3$). The naive agent runs into this infinite loop of reoptimization and eventually takes the same risk of $y(Z_\tau;\beta A_\expo)$ in all periods.

\begin{theo}[Subgame perfect equilibrium strategy of sophisticated agent]
There exists a pair $(\hat{A},\hat{Y})\in\bR\times \cY$ solving the system \eqref{eq:sophi_sys}. 
Furthermore, there exists $\hat{\pi}\in\Pi$ a periodic subgame perfect equilibrium strategy for a sophisticated agent such that
\begin{align*}
	\hat{V}(x):=J_n(\hat{\pi}^{(n)};x)=\sup_{\pi^n\in\Pi^n}J_n(\pi^n\oplus\hat{\pi}^{(n+1)};x)=\hat{A}x^\alpha
\end{align*}
for all $n\in \bN_0$ and $x>0$. The corresponding portfolio process $\hat{X}:=X^{0,x,\hat{\pi}}$ satisfies
\begin{equation}
	\hat{X}_{T_{i+1}}=
		y\left(\frac{Z_{T_{i+1}}}{Z_{T_i}};A_{\so}\right) \hat{X}_{T_{i}},\ \hat{X}_{0}=x,
	\label{eq:soph_port}
\end{equation}
where the function $y(z;\theta)=y(z)$ is defined in Lemma \ref{lem:auxsol} and $Z=(Z_t)_{t\geq 0}$ is the pricing kernel of the Black-Scholes economy given by \eqref{eq:bs_kernel}, and $A_{\so}$ is a constant defined via
\begin{align}
	A_{\so}:=\frac{\beta \hat{A}}{1-(1-\beta)e^{-\delta\tau}\bE[\hat{Y}^\alpha]}.
	\label{eq:A_soph}
\end{align}
\label{thm:soph}
\end{theo}

\begin{proof}
	We first show that there exists $(\hat{A},\hat{Y})$ which solves system \eqref{eq:sophi_sys}. Recall from Proposition \ref{lemm_revise} the definition of $\theta^*(\kappa)$ as the fixed point of $\theta \mapsto e^{-\delta\tau}\Phi(\kappa\theta)$. By Theorem \ref{thm:G_fp}, there exists $\hat{\xi}$ a fixed point of the set-valued map \eqref{eq:g_main}. Define
	\begin{align*}
	\hat{A}:=\theta^*\left(\frac{\beta}{1-(1-\beta)e^{-\delta\tau}\hat{\xi}}\right),\qquad \hat{Y}\in\argmax_{Y\in\cY}\bE\left[F\left(Y;\frac{\beta \hat{A}}{1-(1-\beta)e^{-\delta\tau}\hat{\xi}}\right)\right]=\argmax_{Y\in\cY}\bE\left[F\left(Y;A_{\so}\right)\right],
	\end{align*}
	where the last equality is due to the construction of $\hat{\xi}$ such that $\bE[\hat{Y}^\alpha]=\hat{\xi}$. 
	Then
	\begin{align*}
		\bE\Big[e^{-\delta\tau}U(\hat{Y}-\gamma)+\frac{\beta e^{-\delta\tau}}{1-(1-\beta)e^{-\delta\tau}\bE[\hat{Y}^\alpha]}\hat{A}\hat{Y}^\alpha\Big]&=e^{-\delta\tau}\bE\left[F\left(\hat{Y};\frac{\beta \hat{A}}{1-(1-\beta)e^{-\delta\tau}\hat{\xi}}\right)\right]\\
		&=e^{-\delta\tau}\sup_{Y\in\cY}\bE\left[F\left(Y;\frac{\beta \hat{A}}{1-(1-\beta)e^{-\delta\tau}\hat{\xi}}\right)\right]\\
		&=e^{-\delta \tau}\Phi\left(\frac{\beta \hat{A}}{1-(1-\beta)e^{-\delta\tau}\hat{\xi}}\right).		
	\end{align*}
But $\hat{A}=\theta^*\left(\frac{\beta}{1-(1-\beta)e^{-\delta\tau}\hat{\xi}}\right)$ and hence
	\begin{align*}
		\hat{A}=e^{-\delta\tau}\Phi\left(\frac{\beta \hat{A}}{1-(1-\beta)e^{-\delta\tau}\hat{\xi}}\right)=\sup_{Y\in\cY}\bE\Big[e^{-\delta\tau}U(Y-\gamma)+\frac{\beta e^{-\delta\tau}}{1-(1-\beta)e^{-\delta\tau}\bE[\hat{Y}^\alpha]}\hat{A}Y^\alpha\Big].
	\end{align*}
	Therefore $(\hat{A},\hat{Y})$ is a solution to \eqref{eq:sophi_sys}.
	
	Using the standard portfolio replication argument, there exists a periodic $\hat{\pi}\in\Pi$ such that the resulting portfolio process $\hat{X}$ satisfies \eqref{eq:soph_port}, and each copy of $y(Z_{T_{i+1}}/Z_{T_i};A_{\so})$ has the same law as $\hat{Y}$. Then by design that $(\hat{A},\hat{Y})$ satisfies \eqref{eq:sophi_sys}, ($V^{\hat{\pi}}(\cdot), W^{\hat{\pi}}(\cdot))$ satisfies \eqref{eq_exHJB} such that $\hat{\pi}$ is indeed a subgame perfect equilibrium due to Proposition \ref{prop:sophis_system}.
\end{proof}

Theorem \ref{thm:soph} suggests that the agent should trade in the way such that the periodic gross return of the equilibrium portfolio follows a common target law of $y(Z_{\tau};A_{\so})$. Unlike the strategy adopted by a pre-committing agent, the target law of the sophisticated agent remains the same across all periods. A sophisticated agent therefore acts very similarly as a naive agent, in the sense that both types of agent will target some iid law of the periodic gross return across all periods. Suppose $\beta\in(0,1)$ (i.e. we exclude the cases of $\beta=0$ and $\beta=1$). Then if $A_{\expo}=\theta^*(1)\neq 0\iff \Phi(0)\neq 0$, we can define 
\begin{align*}
	\hat{\beta}:=A_{\so}/A_{\expo}\in(0,1)
\end{align*}
such that the target law of the sophisticated agent is $y(Z_{\tau};\hat{\beta}A_{\expo})$. Note that in the above expression, the property that $\hat{\beta}\in(0,1)$ is non-trivial and we will verify this claim in Proposition \ref{prop:type_cs}. The sophisticated agent can then be regarded as a naive agent but with a modified present bias parameter $\hat{\beta}$ (which endogenously depends on the agent's ``natural'' present bias parameter $\beta$ since $A_{\so}$ depends on $\beta$). In Section \ref{sect:discuss}, we will contrast the behaviors of the sophisticated and the naive agent by comparing $\beta$ and $\hat{\beta}$.

To close this section, we address the uniqueness of the equilibrium strategy for the sophisticated agent. It is useful to first introduce an economically important quantity
\begin{align}
	A_{\my}:=\theta^*(0)=e^{-\delta\tau}\sup_{Y\in\cY}\bE[U(Y-\gamma)],
	\label{eq:A_myopic}
\end{align}
which is the value function of a one-period optimization problem or equivalently a version of the problem faced by a completely myopic agent with $\beta=0$. Note that the sign of $A_{\my}$ does not depend on $(\beta,\delta)$ the time preference parameters of the agent, and we will soon see in Section \ref{sect:discuss} that this quantity has significant economic impact on the agent's risk taking behaviors.

\begin{prop}
Let $\hat{\pi}^i\in\Pi$ for $i\in\{1,2\}$ be two periodic subgame perfect equilibrium strategies such that $X^{0,x,\hat{\pi}^i}_{T_{n+1}}=X^{0,x,\hat{\pi}^i}_{T_n}Y^i_{n+1}$ for all $n\in \bN_0$, where $\{Y^i_n\}_{n\in\bN}$ are the i.i.d. periodic gross return variables induced by $\hat{\pi}^i$. If $A_{\my}\leq 0$, then $\bP(Y^1_n=Y^2_n)=1$ for all $n\in\bN$. 
\label{prop:unique}
\end{prop}

The consequence of Proposition \ref{prop:unique} is that the subgame perfect equilibrium portfolio is unique (up to $\bP$-null set) within the class of all periodic portfolios, at least under the assumption that $A_{\my}\leq 0$. Despite not being able to prove this formally, numerical evidence seems to suggest that Proposition \ref{prop:unique} also holds in the case of $A_\my>0$. Formal verification of this conjecture is left as a future work to be done.

\begin{rem}
Recall the constants arising in the characterization of the optimal/equilibrium strategies for different types of agent are such that
\begin{align*}
	A_{\my}=\theta^*(0),\quad A_{\expo}=\theta^*(1),\quad A_{\pre}=e^{-\delta\tau}\Phi(\beta\theta^*(1)),\quad  	A_{\so}=\frac{\beta }{1-(1-\beta)e^{-\delta\tau}\hat{\xi}}\theta^*\left(\frac{\beta }{1-(1-\beta)e^{-\delta\tau}\hat{\xi}}\right),
\end{align*}
In other words, the value functions of pre-committing, completely myopic and sophisticated agent with present bias, as well as agent with exponential discounting preference, are all connected to the function $\theta^*(\cdot)$ introduced in Proposition \ref{lemm_revise}. This observation will help us establish some comparative statics in the next section. For a sophisticated agent, the value of $\hat{\xi}$ (the fixed point of $G$ where $G$ is defined in Theorem \ref{thm:G_fp}) may not be unique in general. In such case, we interpret $A_{\so}$ using the above expression under a fixed choice of $\hat{\xi}$.
\label{rem:val_theta}
\end{rem}

\section{Discussion of the main results and comparative statics}
\label{sect:discuss}

To summarize the theoretical findings in Section \ref{sect:mainresults}, different types of agent will trade in a way such that the periodic gross returns of the portfolio follow some independent distribution given by $y(Z_\tau;\theta)$ which solves an auxiliary problem in form of \eqref{eq:axu}. Here, $Z_\tau$ represents the realized change in the state of the world within one period, and the choice of $\theta$ depends on the nature of the agent. An exponential discounter will choose $\theta=A_{\expo}$ in all periods (\cite{TsZh21}); a pre-committing agent will choose $\theta=\beta A_{\expo}$ in the first period and then $\theta=A_{\expo}$ in all the subsequent periods (Theorem \ref{thm:precomm}); a naive agent will take $\theta=\beta A_{\expo}$ in all periods (Corollary \ref{cor:naivesol}); while a sophisticated agent will pick $\theta=A_{\so}$ or equivalently $\theta=\hat{\beta}A_{\expo}$ in all periods (Theorem \ref{thm:soph}).

As discussed in \cite{TsZh21} (and see Lemma \ref{lem:auxsol} as well), different values of $\theta$ will result in different probabilistic behaviors of the optimizer $Y^*=y(Z_\tau;\theta)$. For positive value of $\theta$, the optimal gross return is always strictly positive without atom attached to zero while the upside unbounded. However, for mildly negative value of $\theta\in[-1,0]$, $Y^*$ will have a probability mass at zero, i.e. there is a chance that the portfolio will be ruined at the end of a period. This downside risk is typically associated with excessive leverage when the portfolio is experiencing losses. When $\theta$ becomes moderately negative such that it is in the range of $\theta\in(\underline{\theta},-1)$, the support of $Y^*$ now also has an upper bound which means the growth potential of the portfolio becomes capped, in addition to the bankruptcy possibility signified by an atom at zero. In this case, the agent not only engages in excessive risk taking in bad states of the world but they also disinvest in the good states of the world. The case of highly negative value of $\theta$ with $\theta<\underline{\theta}$ is a mathematically degenerate one. The agent will intentionally go bankrupt via a ``suicidal strategy'' (e.g. a doubling-down strategy) as to avoid the terminal penalty caused by a very negative $\theta$. Informally, we can say that a smaller value of $\theta$ leads to a payoff profile $y(Z_\tau;\theta)$ with ``more negatively skewed risk'', in the sense that the agent tends to take more (less) risk when they are losing (winning).

Why does the value of (the endogenized) $\theta$ affect the risk profile of the portfolio so drastically? Detailed economic explanations can be found in \cite{TsZh21}. But to highlight the main ideas briefly, it is due to the trade-off between maximizing the reward from the current period and the continuation value such that the agent's incentive is summarized by the effective utility function in \eqref{eq:F} which we restate here as $$U(y-\gamma)+\theta y^{\alpha}.$$ If $\theta$ is large and positive, then the continuation value (per unit capital available at the start of the next period) is large and positive such that the agent will seek a less risky strategy which prioritizes value preservation and guarantees solvency. If $\theta$ is close to zero, then the agent does not care whether the portfolio goes bust or not in the current period because the continuation value from the future rewards are insignificant. With their S-shaped utility function $U$, they are then incentivized to gamble aggressively when falling behind (which exposes the portfolio to bankruptcy risk) and off-load risk when being ahead without taking the long-term prospect into consideration, resulting in a more negatively skewed strategy. In the more extreme case where $\theta$ is moderately negative, the continuation value is negative which discourages the portfolio from growing too much. It is due to the phenomenon of underperformance aversion where the agent intentionally limits the portfolio growth to avoid setting up a higher absolute benchmark to be adopted in subsequent evaluations. This will result in a even more negatively skewed periodic payoff as the upside growth is now bounded.

Since the endogenized value of $\theta$ is heavily influencing the risk profile of the optimal portfolio, it is useful to establish the ranking of these values across different types of the agents. 
\begin{prop}[Comparison of risk profiles across agent's types]
Recall the definitions of $A_{\my}$, $A_{\expo}$ and $A_{\so}$ in Remark \ref{rem:val_theta}. In the case of $A_{\my}\neq 0$, define $\hat{\beta}:=\frac{A_{\so}}{A_{\expo}}$. We have for any given $\beta\in(0,1)$ that:
\begin{enumerate}
	\item If $A_\my=0$, then $\beta A_\expo=A_\expo = A_{\so}=0$.
	\item If $A_\my>0$, then $0<A_{\so}< \beta A_\expo< A_\expo$. In particular, $\hat{\beta}\in(0,\beta)$.
	\item If $A_\my<0$, then $A_\expo< A_{\so}\leq \beta A_\expo<0$. In particular, $\hat{\beta}\in[\beta,1)$. If we further have $A_{\expo}\geq \underline{\theta}$, then $A_{\so}<\beta A_{\expo}$ and $\hat{\beta}>\beta$.
\end{enumerate}
\label{prop:type_cs}
\end{prop}

\begin{proof}

	The results mostly follow from Proposition \ref{lemm_revise} and Remark \ref{rem:val_theta}. Part (1) is trivial. For part (2), $A_{\my}>0$ implies $A_{\expo}=\theta^*(1)>0$ and hence we obviously have $\beta A_{\expo}< A_{\expo}$ under $\beta<1$. Meanwhile,
	\begin{align}
		A_{\so}=\frac{\beta }{1-(1-\beta)e^{-\delta\tau}\hat{\xi}}\theta^*\left(\frac{\beta }{1-(1-\beta)e^{-\delta\tau}\hat{\xi}}\right)\leq \beta \theta^*(1)=\beta A_{\expo}
	\label{eq:Aso_less_Aexp}
	\end{align}
	by \eqref{eq:theta_kappa_ineq} under the choice of $\kappa_1=1$ and $\kappa_2=\frac{\beta }{1-(1-\beta)e^{-\delta\tau}\hat{\xi}}$, and recall that $\hat{\xi}=\bE[\hat{Y}^\alpha]=\bE[(Y^*(\kappa_2))^\alpha]$ by the construction of $\hat{Y}$ in the proof of Theorem \ref{thm:soph}. Moreover, by part (1) of Proposition \ref{lemm_revise}, equality holds in \eqref{eq:Aso_less_Aexp} if and only if $\kappa_1=\kappa_2$ or equivalently
	\begin{align}
		1=\frac{\beta }{1-(1-\beta)e^{-\delta\tau}\hat{\xi}}
		\label{eq:equality}
	\end{align}
	which would imply $\beta=1$, i.e. a contradiction. Hence we must have $A_{\so}<\beta A_{\expo}$.
	
	For part (3), we have $\theta^*(\cdot)<0$ under $A_{\my}<0$. Hence by similar arguments used in part (2) of the proof and the strictly decreasing property of $\kappa\mapsto \kappa\theta^*(\kappa)$, we have
	\begin{align*}
		A_{\expo}=\theta^*(1)\leq \frac{\beta }{1-(1-\beta)e^{-\delta\tau}\hat{\xi}}\theta^*\left(\frac{\beta }{1-(1-\beta)e^{-\delta\tau}\hat{\xi}}\right)=A_{\so}\leq \beta \theta^*(1)=\beta A_{\expo}.
	\end{align*}
	This first inequality becomes equality if and only if \eqref{eq:equality} holds which again would lead to a contradiction of $\beta=1$, and hence $A_{\expo}<A_{\so}$. If we further have $A_{\expo}=\theta^*(1)\geq \underline{\theta}$, then part (3) of Proposition \ref{lemm_revise} together with the same arguments used in part (2) of this proof would allow us to conclude $A_{\so}<\beta A_{\expo}$. 
\end{proof}

\begin{rem}
Proposition \ref{prop:type_cs} does not require uniqueness of the equilibrium periodic strategy in Theorem \ref{thm:soph}.
\end{rem}

In the following discussion, we assume $\beta\in(0,1)$. In the case that the investment prospect is valuable such that $A_\my>0$, we have $\beta A_\expo< A_\expo$. The pre-committing agent will take more negatively skewed risk with $\theta=\beta A_{\expo}$  in the first period. However, they have a plan to reduce the risk level to $\theta=A_{\expo}$ from the second period onward, which is the same level of risk a rational exponential discounter would have taken right from the beginning. But if the agent turns out to be naive who cannot commit to a planned strategy, then they will be taking the risk of $\theta=\beta A_{\expo}$ in all periods. Economically, the agent is indefinitely delaying the action of de-risking (i.e. the switch from the more negatively skewed risk of $\theta=\beta A_\expo$ into a safer one with $\theta=A_\expo$). This is an example of procrastination induced by time-inconsistency, which could potentially be costly in terms of welfare if undertaking negatively skewed risk is considered socially undesirable. For further examples on the linkage among procrastination, hyperbolic discounting and social welfare losses, see \cite{akerlof1991} for example. 

Economically, why will a pre-committing agent (in the first period) and a naive agent take more risk than an exponential discounter when $A_{\my}>0$? Recall that the main economic mechanism of the periodic portfolio selection model is the trade-off between the reward from the current period and those from the future periods. Both the pre-committing agent and the naive agent think they will behave like an exponential discounter starting from the second period. If they are time-rational with $\beta=1$, then the trade-off is simply governed by $U(Y-\gamma)$ versus $A_\expo Y^\alpha$. But if they suffer from present bias with $\beta< 1$, then the weights across these components are distorted and the trade-off becomes $U(Y-\gamma)$ versus $\beta A_\expo Y^\alpha$ instead. Present bias makes the agent impatient over short term outcomes and hence the continuation component now carries a smaller decision weight due to the multiplication of a factor of $\beta< 1$. If $A_{\my}>0\iff A_{\expo}>0$, the contribution from the continuation value component $\beta A_{\expo}$ decreases relative to the rational benchmark $A_\expo$. Hence the agent will care less about the long term performance of the portfolio and in turn are inclined to take more negatively skewed risk. 

The opposite phenomenon will occur if $A_{\my}<0\iff A_{\expo}<0$. The contribution from the continuation value $\beta A_{\expo}$ is larger (i.e. less negative) relative to the rational benchmark $A_\expo$. Here, the present bias induces the agent to focus more on the outcome in the current period and worry less about potential penalties due to underperformance in the future. Consequently, the agent is more willing to take risk in a way that would result in better upside potential of the portfolio. Indeed, if the model parameters are such that $A_\expo< -1< \beta A_\expo<0$, then the corresponding optimal gross portfolio returns of an exponential discounter and a myopic (naive/pre-committing in the first period) agent can have drastically different supports where the former is capped from the above while the latter enjoys an unlimited upside. See Lemma \ref{lem:auxsol}.

In either case of $A_\my\gtrless 0$, the sophisticated agent always takes more negatively skewed risk than the naive agent as revealed by $A_{\so}< \beta A_{\expo}$. From the perspective at time zero, all three types of agent agree that they should behave as an exponential discounter from the second period onward. Both the pre-committing agent and the naive agent think they will be able to adhere to this plan, and therefore the present-bias-adjusted continuation value for these agents are $\beta A_{\expo}$. However, the sophisticated agent anticipates in advance that they will suffer from time-inconsistency and they will sub-optimally (from the perspective of today) deviate from the exponential discounter's strategy in the future. They know that the sub-optimal strategy adopted by their future-selves will result in a lower net-present-value of the continuation component, say $\tilde{A}$, relative to the exponential discounter's benchmark of $A_{\expo}$. Consequently, the optimization problem faced by the current-self of the sophisticated agent involves a present-bias-adjusted continuation value component of $\beta \tilde{A}$ which is smaller than $\beta A_{\expo}$. The correct choice of $\tilde{A}$ is determined by the equilibrium condition among all the incarnations of the agent in the sequential game. Simply speaking, the fact that a sophisticated agent is aware of their time-inconsistency makes them more ``pessimistic'' over the value of the future rewards relative to a pre-committing agent (in the first period) and a naive agent, and as a result the sophisticated agent is willing to adopt a strategy with higher negatively skewed risk.

As an alternative perspective, a sophisticated agent can be regarded as a naive agent with an adjusted present bias parameter $\hat{\beta}=A_{\so}/A_{\expo}$. Proposition \ref{prop:type_cs} then suggests that $\hat{\beta}< \beta$ (resp. $\hat{\beta}> \beta$) when $A_{\my}>0$ (resp. $A_{\my}<0$) under which the sophisticated agent is a version of a naive agent with stronger (resp. weaker) present bias. In other words, a sophisticated agent discounts long-dated positive (negative) outcomes more (less) heavily when the investment prospect is (un)favorable. This is consistent with the idea that the sophisticated agent has a more pessimistic valuation of the future outcomes relative to a naive agent.

The situation of $A_\my=0$ is a theoretically interesting corner case under which the strategies adopted by the three agents become indistinguishable and they degenerate to the one adopted by a completely myopic, one-period agent. This condition does not depend on the time preference parameters of the agent. In this case, the value of a one-period investment game is neutral to the agent, and its value remains neutral even if the agent can play this game repeatedly. Applying a sequence of quasi-hyperbolic discount factors to these neutral outcomes does not make them more or less attractive to the agent.

\begin{figure}[!htbp]
	\captionsetup[subfigure]{width=0.5\textwidth}
	\centering
	\subcaptionbox{$A_\my>0$.\label{fig:diff_agents_positive}}{\includegraphics[scale =0.475] {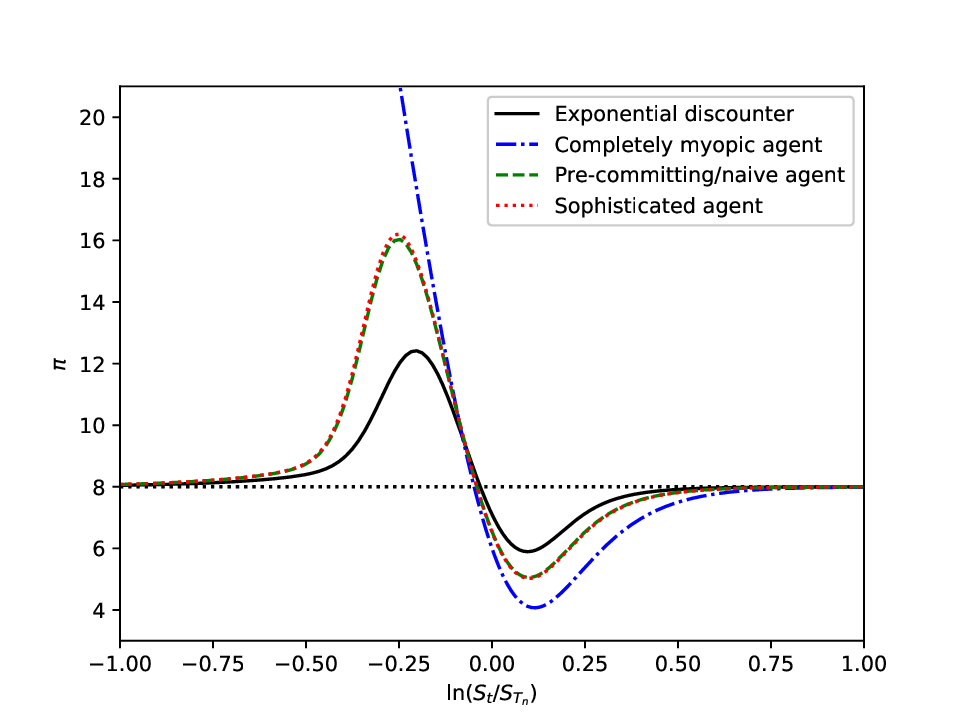}}
	\subcaptionbox{$A_\my<0$.\label{fig:diff_agents_negative}}{\includegraphics[scale =0.475] {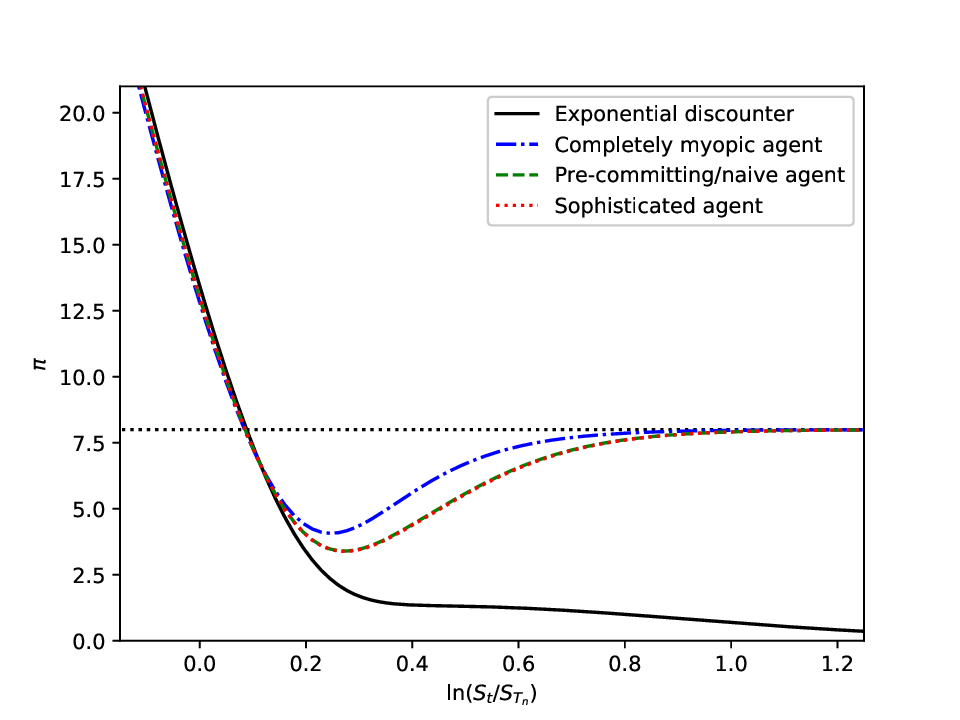}}
	
	\caption{The optimal investment level (fraction of wealth invested in the risky asset) as a function of log-return of the underlying stock at a fixed time for different types of agents. The horizontal dotted line indicates the Merton ratio $(\mu-r)/(\sigma^2(1-\alpha))$.  An exponential discounter and a completely myopic agent correspond to the special cases of $\beta=1$ and $\beta=0$ respectively. We set $\beta=0.4$ for the remaining types of agent. We assume the pre-committing agent is currently in the first period of the investment horizon such that their strategy is identical to that of a naive agent, and recall that a pre-committing agent behaves as an exponential discounter starting from the second period. Other base parameters used are: $\alpha=0.5$, $k=1.25$, $\gamma=1$ (for the case of $A_\my>0$), $\gamma=2.5$ (for the case of $A_\my<0$), $\delta=0.3$, $\beta=0.4$, $\mu=0.1$, $\sigma=0.15$, $r=0.01$, $\tau=1$, $t=T_{n}+0.5\tau$.}
	\label{fig:diff_agents}
\end{figure}

The discussion above can be visually summarized by Figure \ref{fig:diff_agents}, which illustrates the risk taking level of different types of the agent as the state of the world varies. We plot the optimal proportion of wealth invested in the risky asset as a function of the running periodic log-return of the risky asset. See Proposition EC.3 of \cite{TsZh21} for the expression of this quantity. In Figure \ref{fig:diff_agents_positive} where $A_\my>0$, the exponential discounter takes the least amount of negatively skewed risk in that the investment level is the lowest (highest) in the bad (good) states of the world among all types of agent. The investment level of the pre-committing/naive agent is numerically close to that of the sophisticated agent, but a closer inspection of the figure can still reveal that the latter invests less (more) during a bearish (bullish) market leading to less negatively skewed risk taken overall. 

When $A_\my<0$, Figure \ref{fig:diff_agents_negative} shows that all types of agents will engage excessive leverage without any upper bound on the investment level in the bad states of the world. It is due to Proposition \ref{prop:type_cs} that all $A_\expo$, $\beta A_\expo$ and $A_{\so}$ are negative when $A_\my<0$. Hence in all cases the periodic gross return variable has an atom at zero, which is associated with unboundedly large risk taking during downturns. While the investment levels of different types of agent are very close in the bad states of the world, the exponential discounter indeed invests the most when the risky stock is declining in value, followed by the pre-committing/naive agent, then the sophisticated agent and finally the completely myopic agent. The ranking of the investment levels is clearer on the positive return regime, which is opposite to that in the negative return regime. In this case, the ``rational'' exponential discounter actually takes way more negatively skewed risk to their present-biased counterparts.

Note that in either case of $A_\my>0$ or $A_\my<0$, the sophisticated agent invests more (less) than a naive agent in the bad (good) states, albeit the small numerical difference. If a portfolio strategy with negative skew risk is deemed to be economically undesirable (e.g. a social planner might want to advocate a long-term, steady financial growth while minimizing insolvency risk within the asset management sector), then sophisticated thinking is indeed more detrimental than naivety from a welfare viewpoint.

Under certain conditions, we can also establish how the risk profile of the strategy varies with $\beta$ within the same type of the agent.
\begin{prop}[Comparative statics with respect to myopia level]
View $A_{\so}=A_{\so}(\beta)$ as a quantity depending on $\beta\in(0,1]$:
	\begin{enumerate}
		\item If $A_\my>0$ and the fixed point of $G$ defined in \eqref{eq:g_main} is unique, then $\beta A_\expo$ and $A_{\so}(\beta)$ are both non-decreasing in $\beta$.
		\item If $A_\my<0$, then $\beta A_\expo$ and $A_{\so}(\beta)$ are both non-increasing in $\beta$.
	\end{enumerate}
\label{prop:myopia_cs}
\end{prop}

\begin{figure}[!htbp]
	\captionsetup[subfigure]{width=0.5\textwidth}
	\centering
	\subcaptionbox{$A_\my>0$.\label{fig:val_multiplier_positive}}{\includegraphics[scale =0.475] {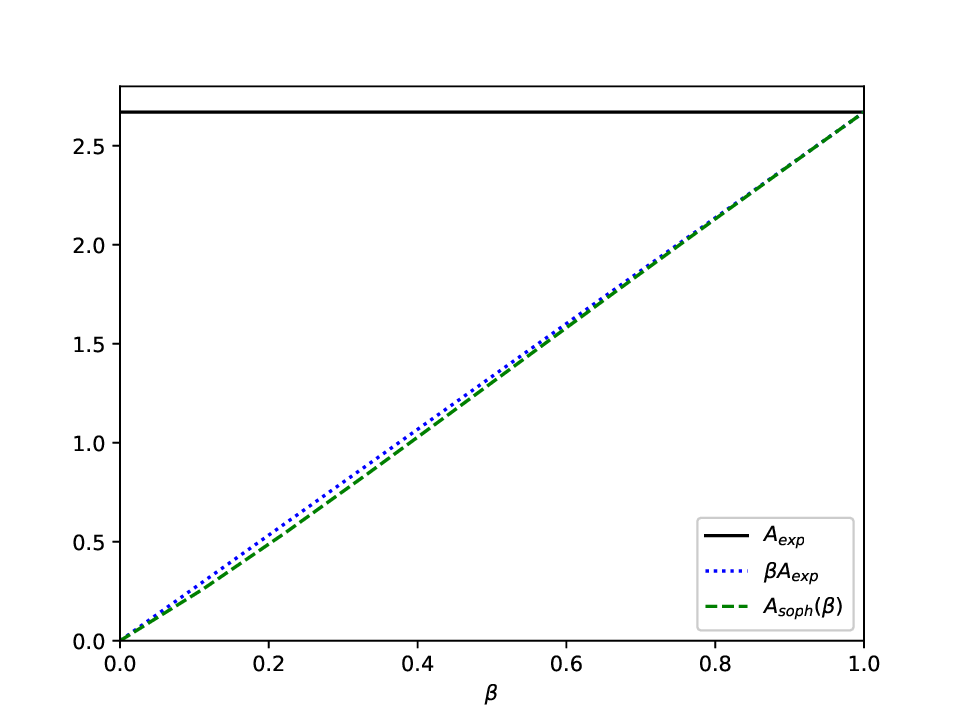}}
	\subcaptionbox{$A_\my<0$.\label{fig:val_multiplier_negative}}{\includegraphics[scale =0.475] {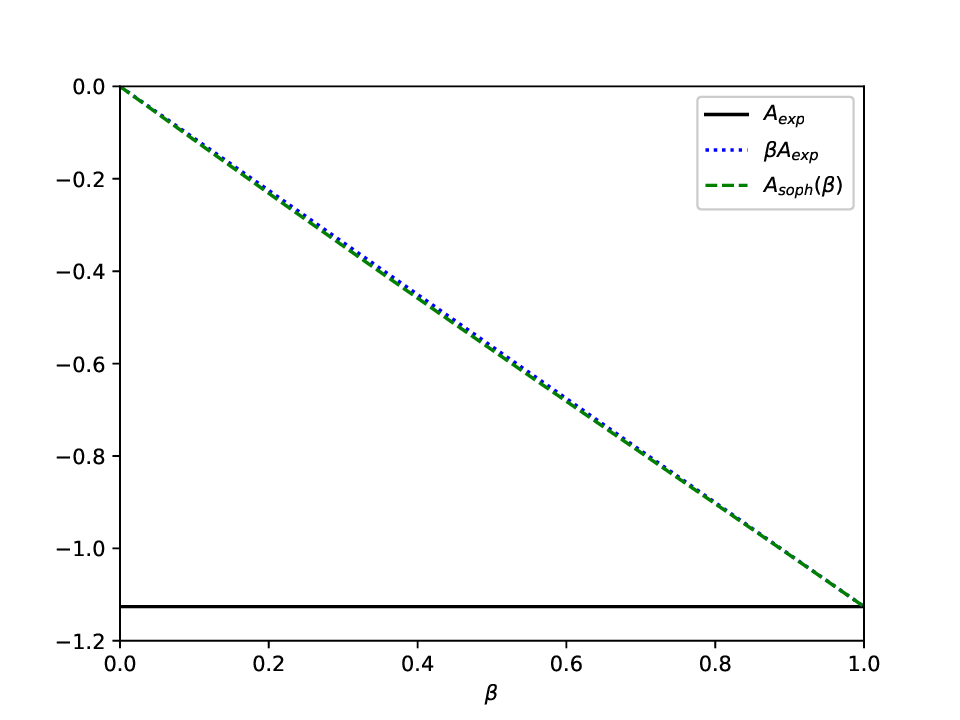}}
	
	\caption{Comparison of $A_{\expo}$, $\beta A_{\expo}$ and $A_{\so}$ as functions of $\beta$. Base parameters used are: $\alpha=0.5$, $k=1.25$, $\gamma=1$ (for the case of $A_\my>0$), $\gamma=2.5$ (for the case of $A_\my<0$), $\delta=0.3$, $\mu=0.1$, $\sigma=0.15$, $r=0.01$, $\tau=1$, $t=T_{n}+0.5\tau$.} 
	\label{fig:val_multiplier_beta}
\end{figure}

Figure \ref{fig:val_multiplier_beta} numerically compares $A_\expo$, $\beta A_\expo$ and $A_{\so}$ under different values of $\beta$. The qualitative behaviors of the plots agree with the theoretical statements shown in Proposition \ref{prop:type_cs} and \ref{prop:myopia_cs}. $\beta A_{\expo}$ and $A_{\so}$ are numerically very similar. This echoes the observation in Figure \ref{fig:diff_agents} that the optimal strategy of the naive and sophisticated agent are close to each other. This phenomenon seems to be holding under a wide range of model parameters, suggesting that an optimal naive strategy could be a reasonable approximation of an interpersonal equilibrium strategy of the sophisticated agent. 

If the investment prospect is favorable ($A_\my>0$), then stronger the present bias (smaller $\beta$), more negatively skewed risk will be taken by the (first-period) pre-committing, naive and sophisticated agent. The economic intuition is largely the same as before, where a stronger present bias generally induces the agent to put a relatively larger decision weight to the current reward which favors a more negatively skewed strategy. Otherwise when the investment prospect is poor ($A_\my<0$), a smaller $\beta$ makes the agent less concerned about the penalty embedded in the negative continuation value and in turn they are more willing to take a strategy that could yield a higher upside. Refer to Figure \ref{fig:compstat_beta} for some numerical examples of how the optimal investment level varies with $\beta$ for each type of the agent.

\begin{figure}[!htbp]
	\captionsetup[subfigure]{width=0.5\textwidth}
	\centering
	\subcaptionbox{$A_\my>0$.\label{fig:beta_precommnaive_positive}}{\includegraphics[scale =0.475] {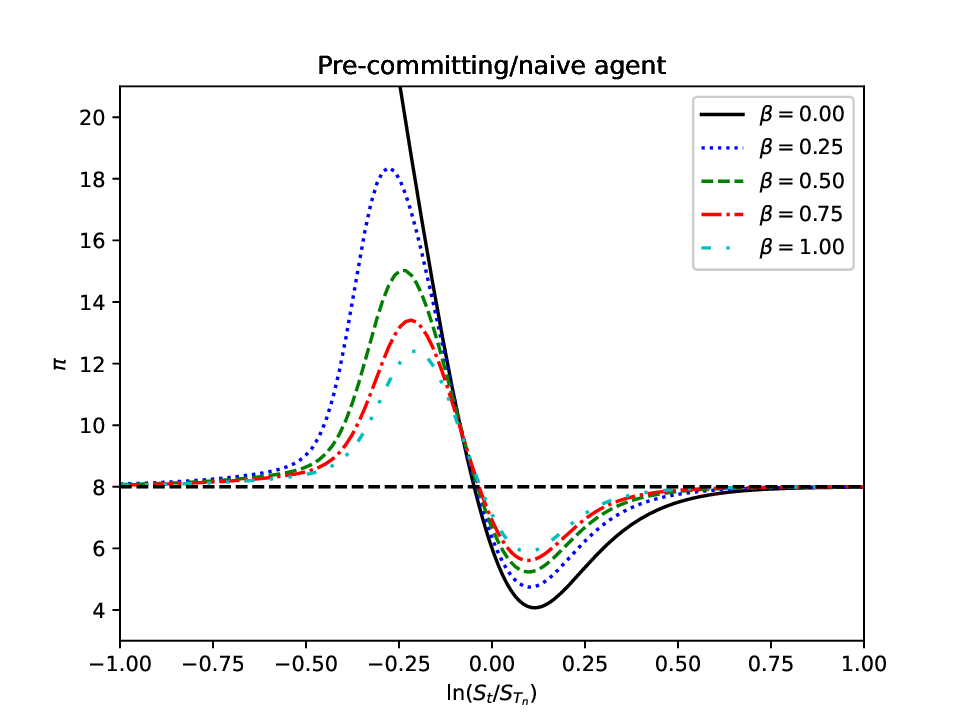}}
	\subcaptionbox{$A_\my>0$.\label{fig:beta_soph_positive}}{\includegraphics[scale =0.475] {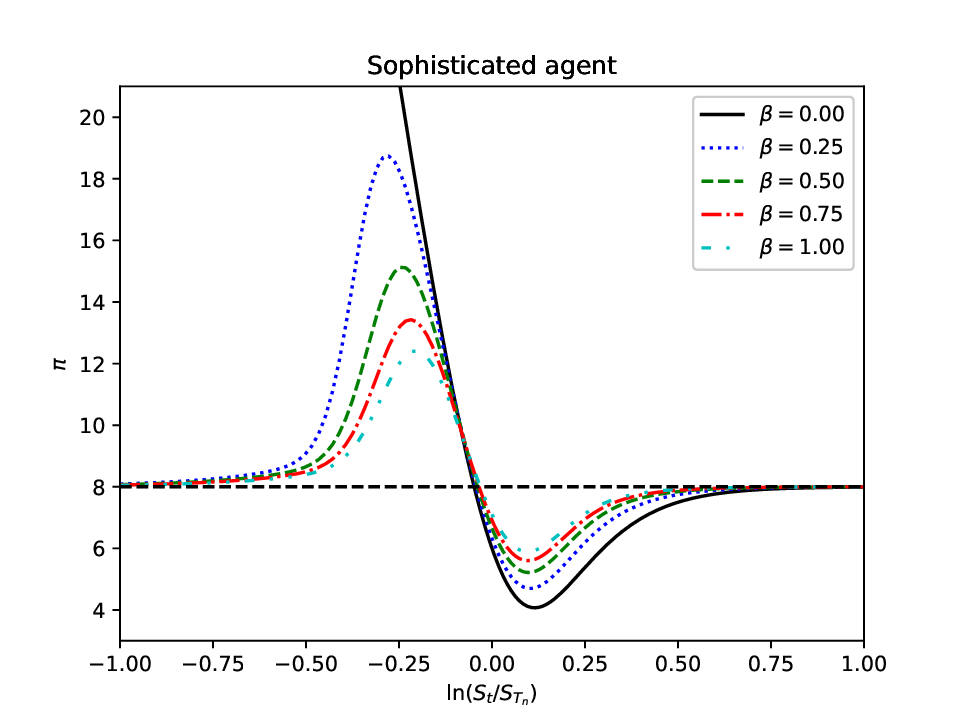}}
	\subcaptionbox{$A_\my<0$.\label{fig:beta_precommnaive_negative}}{\includegraphics[scale =0.475] {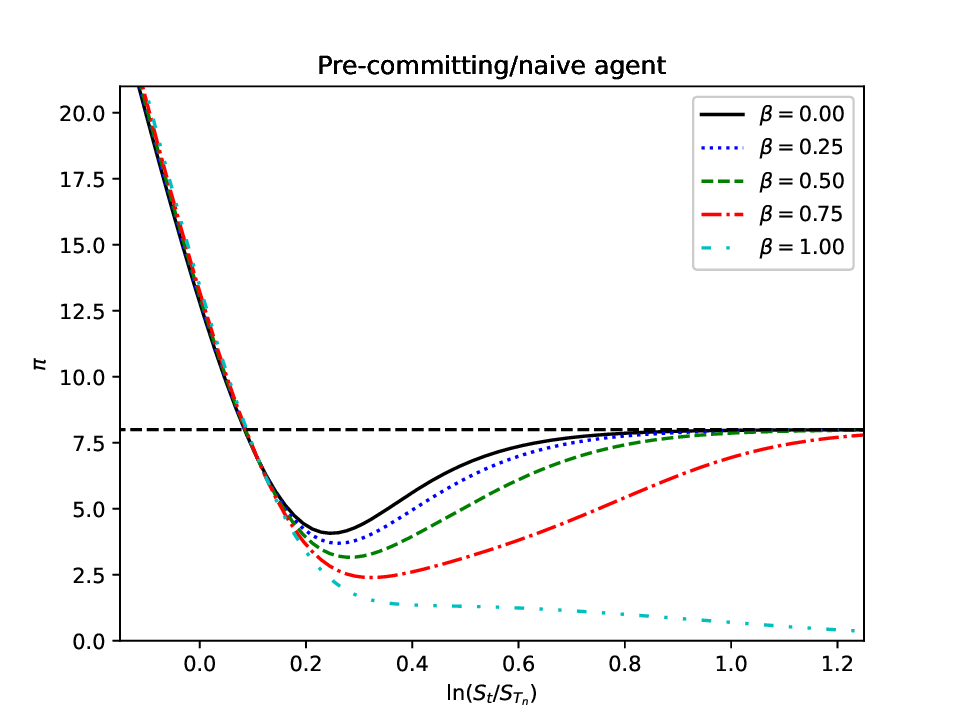}}
	\subcaptionbox{$A_\my<0$.\label{fig:beta_soph_negative}}{\includegraphics[scale =0.475] {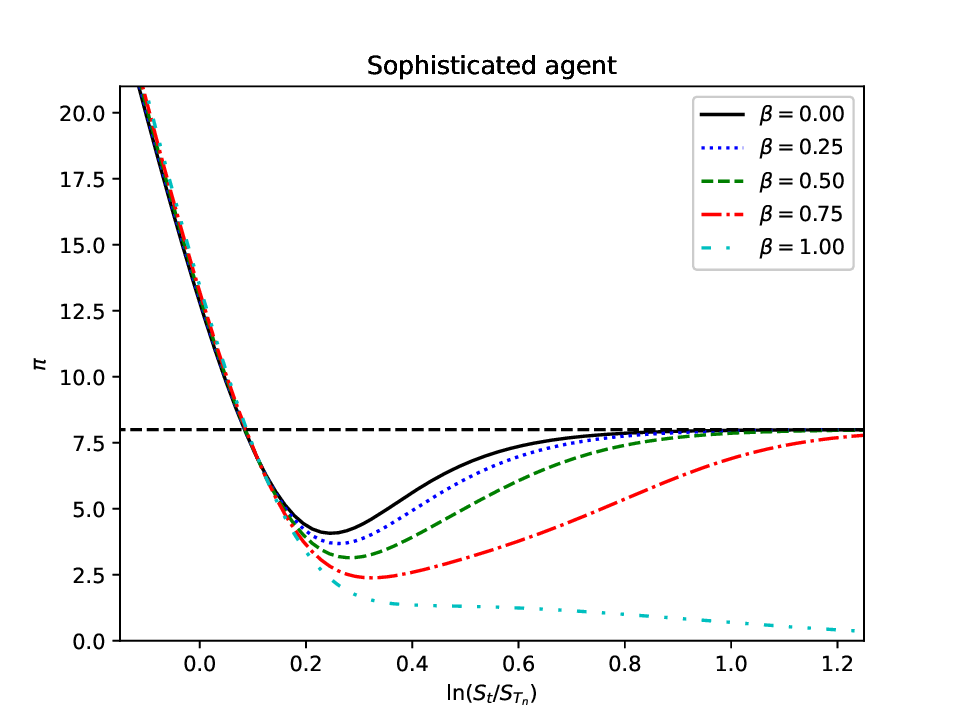}}
	
	\caption{The optimal investment level (fraction of wealth invested in the risky asset) as a function of log-return of the underlying stock at a fixed time under different values of $\beta$. The horizontal dotted line indicates the Merton ratio $(\mu-r)/(\sigma^2(1-\alpha))$. We assume the pre-committing agent is currently in the first period of the investment horizon such that their strategy is identical to that of a naive agent. Other base parameters used are: $\alpha=0.5$, $k=1.25$, $\gamma=1$ (for the case of $A_\my>0$), $\gamma=2.5$ (for the case of $A_\my<0$), $\delta=0.3$, $\beta=0.4$, $\mu=0.1$, $\sigma=0.15$, $r=0.01$, $\tau=1$, $t=T_{n}+0.5\tau$.}
	\label{fig:compstat_beta}
\end{figure}

\section{Concluding remarks}
\label{sect:conclude}

In this paper, we show that present bias can heavily influence agent's risk taking behaviors in the context of periodic portfolio selection. Depending on the attractiveness of the underlying investment opportunity, present bias can either moderate or exaggerate undesirable trading strategies that result in negatively skewed portfolio return. Relative to a naive agent, a sophisticated agent invests more during market downturns and deleverage conservatively during market rally.

Our current study focuses on describing the behaviors of a present-biased agent under each criterion of optimality. Normative status of our results is not directly addressed, in the sense that we do not offer recommendations how the agent should actually behave nor quantify the associated social costs of deviation from such recommendations. While our discussion of the skewness of the portfolio return might offer some guidance, a full welfare analysis will be an interesting direction for future works which can yield more precise policy insights in the area of delegated portfolio management. There is already a long strand of literature addressing the social implications of present bias and time-inconsistency on consumption/saving behaviors (see \cite{laibson98}, \cite{diamond-koszegi03}, \cite{oDonoghue-rabin06}, \cite{oDonoghue-rabin15}, \cite{calcott-petkov2021}, among others). Specific follow-up research questions in our framework may include, for example, what welfare benchmark(s) should be considered to evaluate the portfolio strategy adopted by a certain type of agent, how the demand for commitment devices can be endogenized, and how to quantify the social value of ``paternalistic policies'' such as risk regulation on portfolio managers.

\appendix
\numberwithin{equation}{section}

\begin{appendices}
	
\section{Proof of Theorem \ref{thm:G_fp}}
\label{app:H_cont}

This entire section is dedicated to the proof of Theorem \ref{thm:G_fp} concerning the existence of fixed point of the set-valued map defined in \eqref{eq:g_main}. We begin by studying the theoretical properties of a closely related map
\begin{equation}
	H(\theta):=\left\{\bE[Y^\alpha ]\ | \ Y\in\argmax _{Y\in\cY}\bE[F(Y;\theta)]\right\},
	\label{eq:H}
\end{equation}
where $F$ is defined in \eqref{eq:F}.

\begin{prop}
	Recall the definition of $\underline{\theta}$ in Lemma \ref{lem:shape_F}. For $\theta\in\mathbb{R}$, we have:
	\begin{enumerate}
		\item $H(\theta)\subseteq [0,e^{h\tau}]$, where $h$ is defined in \eqref{eq:assump}.
		\item $H(\theta)$ is a singleton for $\theta\neq \underline{\theta}$.
		\item $H(\theta)=\{0\}$ when $\theta<\underline{\theta}$.
		\item Suppose $\gamma>e^{r\tau}$ and $k>0$. Then $H(\underline{\theta})=[0,\underline{H}]$ where
		\begin{equation}
			\underline{H}:=L(\eta^*)\gamma^{\alpha}(1+k^{-\frac{1}{1-\alpha}})^{\alpha},
			\label{eq:cri_h}
		\end{equation}
		with $\eta^*>0$ being the unique solution to an equation in $\eta$ defined as
		\begin{equation*}
			\gamma(1+k^{-\frac{1}{1-\alpha}})\int_0^{\eta } z\ell(z)\rd z = 1,
		\end{equation*}
		and $\ell(\cdot)$ and $L(\cdot)$ are the pdf and cdf of the log-normal random variable $Z_{\tau}$.
	\end{enumerate}
	\label{prop:H}
\end{prop}

\begin{proof}
	
	Part 1 of the statement is trivial because $0\leq \mathbb{E}[Y^\alpha]\leq e^{h\tau}$ for any $Y\in\cY$.
	
	Part 2 and 3 follow immediately from Proposition \ref{prop:uniqueness} and the fact that $Y^*=0$ is the unique optimizer of problem \eqref{eq:axu} when $\theta<\underline{\theta}$.
	
	When $\theta=\underline{\theta}$ we have from Proposition \ref{prop:uniqueness} that any optimizer $Y^*$ of \eqref{eq:axu} must satisfy $\bP(Y^*\in\{0,\gamma(1+k^{-\frac{1}{1-\alpha}})\})=1$. Hence $$H(\underline{\theta})=\left\{\bE[Y^\alpha ]\ | \ Y\in\cF_\tau, \ \bP(Y\in\{0,\gamma(1+k^{-\frac{1}{1-\alpha}})\})=1 ,\ \bE[Z_\tau Y]\leq 1\right\}.$$ To prove part 4, it is hence sufficient to establish that
	\begin{equation}
		H(\underline{\theta})=\left\{L(\eta)\gamma^{\alpha}(1+k^{-\frac{1}{1-\alpha}})^{\alpha}\ \Bigl| \ \eta\geq 0: \gamma(1+k^{-\frac{1}{1-\alpha}})\int_0^{\eta } z\ell(z)\rd z \leq 1\right\}=:\cH.
		\label{eq:set_equivalence}
	\end{equation}
	Then the result will follow on observing that $\cH$ can be expressed as $[0,\underline{H}]$ using the continuity and strict monotonicity of $\eta\mapsto  \gamma(1+k^{-\frac{1}{1-\alpha}})\int_0^{\eta } z\ell(z)\rd z=:\Xi(\eta)$ and the fact that $$\Xi(\infty)=\gamma(1+k^{-\frac{1}{1-\alpha}})\int_0^\infty z\ell(z)\rd z=\gamma(1+k^{-\frac{1}{1-\alpha}})e^{-r\tau}>1$$ under the stated condition of $\gamma>e^{r\tau}$, which in turn implies the existence of $\eta^*>0$ such that $\Xi(\eta^*)=1$. 
	
	We first argue that $H(\underline{\theta})\supseteq \cH$. Suppose $\nu\in\cH$, i.e. $\nu=L(\eta_\nu)\gamma^{\alpha}(1+k^{-\frac{1}{1-\alpha}})^{\alpha}$ for some $\eta_\nu\geq 0$ such that $\gamma(1+k^{-\frac{1}{1-\alpha}})\int_0^{\eta_\nu } z\ell(z)\rd z \leq 1$. Construct a random variable $Y_\nu$ via $Y_\nu:=\gamma(1+k^{-\frac{1}{1-\alpha}})\I_{\{Z_\tau\leq\eta_\nu\}}$. Then obviously $Y_\nu$ is $\cF_\tau$-measurable, $\bP(Y_\nu\in\{0,\gamma(1+k^{-\frac{1}{1-\alpha}})\})=1$, $\bE[Z_\tau Y_\nu]=\gamma(1+k^{-\frac{1}{1-\alpha}})\int_0^{\eta_\nu } z\ell(z)\rd z \leq 1$ and
	\begin{equation*}
		\bE[Y_\nu^\alpha]=\gamma^\alpha(1+k^{-\frac{1}{1-\alpha}})^\alpha\bP(Z_\tau\leq \eta_\nu)=L(\eta_\nu)\gamma^{\alpha}(1+k^{-\frac{1}{1-\alpha}})^{\alpha}=\nu.
	\end{equation*}
	Thus $\nu\in H(\underline{\theta})$.
	
	Next, we show the reverse direction $H(\underline{\theta})\subseteq \cH$. Suppose $\rho\in H(\underline{\theta})$, i.e. $\rho=\bE[Y_\rho^\alpha]$ for some $Y_\rho$ which is $\cF_\tau$-measurable, $\bP(\{0,\gamma(1+k^{-\frac{1}{1-\alpha}})\})=1$ and $\bE[Z_\tau Y_\rho]\leq 1$. Write $E_\rho:=\{\omega\in \Omega | Y_\rho(\omega)=\gamma(1+k^{-\frac{1}{1-\alpha}})\}\in\cF_\tau$. Then $\rho=\bE[Y_\rho^\alpha]=\gamma^\alpha(1+k^{-\frac{1}{1-\alpha}})^\alpha\bP(E_\rho)$. Note that the requirement of $\bE[Z_\tau Y_\rho]\leq 1$ implies $\bP(E_\rho)<1$ because otherwise $1\geq \bE[Z_\tau]\gamma(1+k^{-\frac{1}{1-\alpha}})=\gamma(1+k^{-\frac{1}{1-\alpha}})e^{-r\tau}>1$ which is a contradiction. Set $\eta_\rho:=L^{-1}(\bP(E_\rho))\in[0,\infty)$. Then $L(\eta_\rho)=\bP(E_\rho)$ such that
	\begin{align*}
		\gamma^{\alpha}(1+k^{-\frac{1}{1-\alpha}})^{\alpha} L(\eta_\rho)=\gamma^{\alpha}(1+k^{-\frac{1}{1-\alpha}})^{\alpha}\bP(E_\rho)=\rho.
	\end{align*}
	Finally, the quantile function of the random variable $\I_{E_\rho}$ is given by $p\mapsto \I_{\{p\geq 1-\bP(E_\rho)\}}$. The Hardy-Littlewood inequality for general probability space (see, for example, part (ii) of Theorem 1 in \cite{jin-zhou2010}) states that for any $X\in\cF_\tau$, we have
	\begin{align}
		\bE[Z_\tau Q_X(1-L(Z_\tau))]\leq \bE[Z_\tau X]
		\label{eq:HL}
	\end{align}
	where $Q_X(\cdot)$ is the quantile function of $X$. Specialization of $X=\I_{E_\rho}$ in \eqref{eq:HL} yields $$\bE[Z_\tau \I_{\{\bP(E_\rho)> L(Z_\tau)\}}]\leq \bE[Z_\tau \I_{E_\rho}].$$ Hence,
	\begin{align*}
		\gamma(1+k^{-\frac{1}{1-\alpha}})\int_0^{\eta_\rho } z\ell(z)\rd z&=\gamma(1+k^{-\frac{1}{1-\alpha}})\bE[Z_\tau \I_{\{Z_\tau< \eta_\rho\}}]\\
		&=\gamma(1+k^{-\frac{1}{1-\alpha}})\bE[Z_\tau \I_{\{\bP(E_\rho)> L(Z_\tau)\}}]\\
		&\leq \gamma(1+k^{-\frac{1}{1-\alpha}})\bE[Z_\tau \I_{E_\rho}]\\
		&=\bE[Z_\tau Y_\rho]\leq 1.
	\end{align*} 
	We therefore conclude $\rho\in\cH$ and hence $H(\underline{\theta})\subseteq \cH$. The equivalence of $H(\underline{\theta})=\cH$ is now established.
\end{proof}

For $\theta\neq \underline{\theta}$, $H(\theta)$ is a singleton and hence it can be interpreted as an ordinary function. With a slight abuse of notation, in the rest of this section we will view $H(\theta)$ in such case as the singleton element of the set $H(\theta)$.

\begin{lemm}
	For $\theta\neq \underline{\theta}$, $H$ is non-decreasing on $\theta>\underline{\theta}$, and $H(\theta)=0$ on $\theta<\underline{\theta}$.
	\label{lem:H_increase}
\end{lemm}

\begin{proof}
	The second part of the lemma is obvious since $H(\theta)=\{0\}$ when $\theta<\underline{\theta}$. To prove the first part, consider $\theta_2>\theta_1>\underline{\theta}$. For $i\in\{1,2\}$, there exists $Y^*_i\in \argmax_{Y\in\mathcal{Y}} \bE[F(Y;\theta_i)]$ such that $\sup_{Y\in\mathcal{Y}}\bE[F(Y;\theta_i)]=\bE[F(Y^*_i;\theta_i)]$. Note that for as long as $\theta_i>\underline{\theta}$, $Y^*_i$ is unique such that $H(\theta_i)=\mathbb{E}[(Y^*_i)^\alpha]>0$ is uniquely defined.
	
	Suppose on contrary that we have $H(\theta_2)<H(\theta_1)$. Then
	\begin{align*}
		\sup_{Y\in\mathcal{Y}}\bE[F(Y;\theta_2)]=\bE[F(Y_2^*;\theta_2)]&=\mathbb{E}\left[U(Y^*_2-\gamma)\right]+\theta_2\mathbb{E}\left[ (Y^*_2)^{\alpha}\right]\\
		&=\mathbb{E}\left[U(Y^*_2-\gamma)\right]+\theta_1 H(\theta_2)+(\theta_2-\theta_1)H(\theta_2) \\
		&=\bE[F(Y_2^*;\theta_1)] + (\theta_2-\theta_1) H(\theta_2)\\
		&< \sup_{Y\in\mathcal{Y}}\bE[F(Y;\theta_1)]+ (\theta_2-\theta_1)H(\theta_1)\\
		&=\bE[F(Y_1^*;\theta_1)]+ (\theta_2-\theta_1)H(\theta_1)\\
		&=\mathbb{E}\left[U(Y^*_1-\gamma)\right]+\theta_2\bE[(Y^*_1)^\alpha]=\bE[F(Y^*_1;\theta_2)]\leq \sup_{Y\in\mathcal{Y}}\bE[F(Y;\theta_2)]
	\end{align*}
	arriving at a contradiction. Hence we must have $H(\theta_2)\geq H(\theta_1)$.
\end{proof}

We now proceed to prove the continuity of $H(\theta)$. The high-level idea is as follows: The quantities $(c_1,c_2,c_3,m_1,m_2)$, the functions $(I_1, I_2)$ and the Lagrangian multiplier $\lambda^*$ defined in Section \ref{sect:auxprob} are all related to the unique solutions to some equations parametrized by $\theta$. We can then make use of a result which states that the unique solution to an equation parametrized by some parameters is indeed continuous in those parameters if the solution lives on a compact space. Then the optimizer to problem \eqref{eq:axu} can be expressed as a random variable which depends on $\theta$ continuously, and ultimately $H(\theta)$ is just some integral with an integrand which is continuous in $\theta$. 

The following lemma is the building block of our argument.

\begin{lemm}
	Let $f:\mathcal{X}\times\mathfrak{A}\to\mathbb{R}^d$ be a jointly continuous function where $d\in \bN$, and $\mathcal{X}$ and $\mathfrak{A}$ are metric spaces with $\mathcal{X}$ being compact. If for each fixed $\theta\in\mathfrak{A}$ there exists a unique $\hat{x}(\theta)\in\mathcal{X}$ such that $f(\hat{x}(\theta),\theta)=\vec{0}$, then $\hat{x}:\mathfrak{A}\to\mathcal{X}$ is a continuous function.
	\label{lem:ctsdependence}
\end{lemm}
\begin{proof}
	For an arbitrary $\bar{\theta}\in\mathfrak{A}$, let $\{\theta_n\}_{n\in\bN}$ be a sequence where  $\theta_n\in\mathfrak{A}$ for each $n$ such that $\theta_n\to\bar{\theta}$. Since $f$ is continuous and since $\mathcal{X}$ is compact, the sequence  $\{\hat{x}(\theta_n)\}_{n\in\bN}$ has a convergent subsequence $\{\hat{x}(\theta_{n_k})\}_{k\in\bN}$ with some limit $\bar{x}\in\mathcal{X}$. Then the joint continuity of $f$ implies
	\begin{align*}
		f(\bar{x},\bar{\theta})=f\left(\lim_{k\to\infty}\hat{x}(\theta_{n_k}),\lim_{k\to\infty} \theta_{n_k}\right)=\lim_{k\to\infty} f\left(\hat{x}(\theta_{n_k}),\theta_{n_k}\right)=\vec{0}
	\end{align*}
	and hence $\hat{x}(\bar{\theta})=\bar{x}$ by uniqueness of the solution to $f(\bar{x},\bar{\theta})=\vec{0}$, i.e. any convergent subsequence of $\{\hat{x}(\theta_n)\}_{n\in\bN}$ must have the same limit $\hat{x}(\bar{\theta})$. Thus we must have $\lim_{n\to\infty}\hat{x}(\theta_n)=\hat{x}(\bar{\theta})$ which establishes the continuity of $\hat{x}:\mathfrak{A}\to\mathcal{X}$.
\end{proof}

Since the application of Lemma \ref{lem:ctsdependence} requires the solution space $\mathcal{X}$ to be compact, it will be useful to establish some bounds on several fundamental quantities that appear within the optimizer to problem \eqref{eq:axu}.

\begin{lemm}
	Let $0<\epsilon<M$ and $0<\underline{q}<\bar{q}$ be some positive constants where $\epsilon,\underline{q}$ are arbitrarily small and $M,\bar{q}$ are arbitrarily large. Recall the definitions of $(c_1,c_2,c_3)$, $I_1(q;\theta)$ and $I_2(q;\theta)$ in Section \ref{sect:auxprob}. 
	
	\begin{enumerate}
		
		\item If $\theta\in[\epsilon,M]$, then there exists constants $0<\underline{c}_1<\bar{c}_1<1<\underline{c}_2<  \bar{c_2}<\infty$, independent of $\theta$, such that
		\begin{align*}
			0<\underline{c}_1\leq c_1\leq\bar{c}_1<\underline{c}_2\leq c_2\leq \bar{c_2}.
		\end{align*}
		
		\item If $k>0$ and $\theta\in[\underline{\theta},0]$, then there exists constants $1<\underline{c}_3\leq \bar{c}_3<\infty$, independent of $\theta$, such that $1<\underline{c}_3\leq c_3\leq \bar{c}_3$.
		
		\item If $\theta\in[\underline{\theta}+\epsilon,\underline{\theta}+M]$ and $q\in[\underline{q},\bar{q}]$, we have $\underline{I}_2(\bar{q})\leq I_2(q;\theta)\leq \bar{I}_2(\underline{q})$ where $\underline{I}_2(q)$ and $\bar{I}_2(q)$ are defined as the unique solutions to $\alpha[(y-\gamma)^{\alpha-1}+(\underline{\theta}+\epsilon) y^{\alpha-1}]=q$ and $\alpha[(y-\gamma)^{\alpha-1}+(\underline{\theta}+M) y^{\alpha-1}]=q$ respectively on $y>\gamma$.

		\item If $(q,\theta)\in\{(q,\theta)|\theta\in[\epsilon,M],q\in[\tilde{n}(\theta),\bar{q}]\}$ where  $\tilde{n}(\theta):=	\alpha[k(\gamma-\tilde{c}(\theta)\gamma)^{\alpha-1}+\theta (\tilde{c}(\theta)\gamma)^{\alpha-1}]$ with $\tilde{c}(\theta):=\frac{1}{1+(k/\theta)^{1/(2-\alpha)}}$, then $\underline{I}_1(\bar{q})\leq I_1(q;\theta)\leq \tilde{c}(M)\gamma $. Here $\underline{I}_1(q)$ is defined as the unique solution to $\alpha[k(\gamma-y)^{\alpha-1}+\epsilon y^{\alpha-1}]=q$ on $y\in(0,\tilde{c}(\epsilon)\gamma)$.
		
	\end{enumerate}

\label{lem:bound}
\end{lemm}

\begin{proof}

\begin{enumerate}

\item We give a proof under $k>0$. In the corner case of $k=0$, we will be able to express $c_1$ explicitly in term of $c_2$ and the system of equations defining $(c_1,c_2)$ degenerates to a single equation in $c_2$. 

With $\theta\in[\epsilon,M]$, an upper bound of $c_1$ can be obtained easily as
\begin{equation*}
c_1\leq \frac{1}{1+(k/\theta)^{1/(2-\alpha)}}\leq  \frac{1}{1+(k/M)^{1/(2-\alpha)}}=:\bar{c}_1.
\end{equation*}
Now, recall that $c_2$ satisfies
\begin{eqnarray*}
	\phi_2(c_2;\theta):=(c_2-1)^{\alpha-1}+\theta c_2^{\alpha-1}=k(1-c_1)^{\alpha-1}+\theta c_1^{\alpha-1}.
\end{eqnarray*}
Define $\bar{c}_2\in(1,\infty)$ as the unique solution to $\phi(\bar{c}_2;M)=k$, we have
\begin{eqnarray*}
\phi_2(\bar{c}_2;\theta)\leq\phi(\bar{c}_2;M) =k\leq k(1-c_1)^{\alpha-1}+\theta c_1^{\alpha-1}=\phi_2(c_2;\theta)
\end{eqnarray*}
and hence $c_2\leq \bar{c}_2$ using the properties that $\phi_2(c;\theta)$ is non-decreasing in $\theta$ and non-increasing in $c>1$.

Meanwhile, $c_2$ also satisfies
\begin{equation*}
	\alpha\phi_2(c_2;\theta)=\alpha[(c_2-1)^{\alpha-1}+\theta c_2^{\alpha-1}]=\frac{(c_2-1)^{\alpha}+\theta c_2^\alpha+k(1-c_1)^{\alpha}-\theta c_1^\alpha}{c_2-c_1}.
\end{equation*}
Then if we define $\underline{c}_2$ as the unique solution to the equation $$(c-1)^{\alpha-1}+\epsilon c^{\alpha-1}=\frac{(\bar{c_2}-1)^{\alpha}+M\bar{c}_2^\alpha+k}{\alpha(1-\bar{c}_1)}$$ on $c>1$, then
\begin{align*}
	\phi_2(c_2;\theta)\leq \frac{(\bar{c}_2-1)^{\alpha}+M \bar{c}_2^\alpha+k}{\alpha(1-\bar{c}_1)}=\phi(\underline{c}_2;\epsilon)\leq \phi_2(\underline{c}_2;\theta)
\end{align*}
and hence $c_2\geq \underline{c}_2$, where we have used the definition of $\underline{c}_2$ and the monotonicity of $\phi_2$. 

Lastly, $c_1$ satisfies 
\begin{equation*}
	\alpha \phi_1(c_1;\theta):=\alpha[k(1-c_1)^{\alpha-1}+\theta c_1^{\alpha-1}]=\frac{(c_2-1)^{\alpha}+\theta c_2^\alpha+k(1-c_1)^{\alpha}-\theta c_1^\alpha}{c_2-c_1}.
\end{equation*}
Since
\begin{align*}
	\frac{(c_2-1)^{\alpha}+\theta c_2^\alpha+k(1-c_1)^{\alpha}-\theta c_1^\alpha}{c_2-c_1}\leq 	\frac{(\bar{c}_2-1)^{\alpha}+\theta \bar{c}_2^\alpha+k}{1-\bar{c}_1},
\end{align*}
$\phi_1(c;\theta)$ is strictly decreasing in $c\in(0,1/[1+(k/\theta)^{1/(2-\alpha)}])$ with $\phi_1(0+;\theta)=+\infty$ and $\phi_1(c;\theta)\geq \alpha\theta c^{\alpha-1}$, one can define $\tilde{c}_1$ as the unique solution to the equation $$\alpha\theta \tilde{c}_1^{\alpha-1}=\frac{(\bar{c}_2-1)^{\alpha}+\theta \bar{c}_2^\alpha+k}{1-\bar{c}_1}$$ or equivalently
\begin{align*}
	\tilde{c}_1=\left[\frac{\alpha\theta(1-\bar{c}_1)}{(\bar{c}_2-1)^\alpha+\theta\bar{c}_2^\alpha+k}\right]^{\frac{1}{1-\alpha}},
\end{align*}
under which $c_1\geq \tilde{c}_1$. Take $$\underline{c}_1:=\left[\frac{\alpha\epsilon(1-\bar{c}_1)}{(\bar{c}_2-1)^\alpha+M\bar{c}_2^\alpha+k}\right]^{\frac{1}{1-\alpha}},$$

The conclusion of $c_1\geq \underline{c}_1$ follows immediately on noticing that $\tilde{c}_1\geq \underline{c}_1$.

\item Let
\begin{align*}
	f_3(c;\theta):=(c-1)^\alpha+\theta (1-\alpha) c^\alpha - \alpha c(c-1)^{\alpha-1}+k.
\end{align*}
Note that $f_3(1+;\theta)=-\infty$ and by definition $c_3$ is the unique solution to $f_3(c;\theta)=0$ over $c\in(1,\infty)$ if $\theta\in[-1,0]$ or $c\in(1,c_4)$ if $\theta\in[\underline{\theta},-1)$. Hence, on this range of $c$, we must have $f_3(c;\theta)\lessgtr0$ if and only if $c\lessgtr c_3$. Observe also that, since $k>0$, for $\bar{c}_3:=1+k^{-\frac{1}{1-\alpha}}<\infty$ we have
\begin{align*}
	f_3(\bar{c}_3;\theta)=f(1+k^{-\frac{1}{1-\alpha}})&=(1-\alpha)k^{-\frac{\alpha}{1-\alpha}}\left[1+k^{\frac{1}{1-\alpha}}+\theta(1+k^{\frac{1}{1-\alpha}})^\alpha\right]\\
	&\geq (1-\alpha)k^{-\frac{\alpha}{1-\alpha}}\left[1+k^{\frac{1}{1-\alpha}}+\underline{\theta}(1+k^{\frac{1}{1-\alpha}})^\alpha\right]=0.
\end{align*}
We then deduce $c_3\leq \bar{c}_3$. On the other hand, if we define $\underline{c}_3$ as the unique solution to the equation $(c-1)^\alpha+k=\alpha c(c-1)^{\alpha-1}$ on $c>1$, then
\begin{align*}
f_3(c_3;0)\geq f_3(c_3;\theta)=0=f_3(\underline{c}_3;0)
\end{align*}
by definition of $\underline{c}_3$ and that $f_3(c;\theta)$ is non-decreasing in $\theta$. The claim $c_3\geq \underline{c}_{3}$ follows from the fact that $f_3(c;0)$ is non-decreasing in $c$.

\item This follows immediately from the fact that $I_2(q;\theta)$ is non-increasing in $q$ and non-decreasing in $\theta$.

\item The result is due to the fact that $I_1(q;\theta)$ is non-increasing in $q$ and non-decreasing in $\theta$. For the upper bound, in particular,  $\tilde{c}(\theta)$ is non-decreasing and hence $I_2(q;\theta)\leq \tilde{c}(\theta)\gamma \leq \tilde{c}(M)\gamma$.

\end{enumerate}	
\end{proof}

\begin{rem}
In the special case that $k=0$, $\bar{c}_3$ becomes infinite and we indeed have $c_3(\theta)\to +\infty$ as $\theta\downarrow\underline{\theta}=-1$. But as per Remark \ref{rem:impossible}, there is no need to consider non-positive values of $\theta$ when $k=0$.
\end{rem}

\begin{cor}

View $\{c_i\}_{i=1,2,3}=\{c_i(\theta)\}_{i=1,2,3}$, $\{m_i\}_{i=1,2}=\{m_i(\theta)\}_{i=1,2}$ and $\{I_{i}(q;\theta)\}_{i=1,2}$ as functions of $\theta$ or $(q,\theta)$. Then:

\begin{enumerate}
	\item $c_1(\theta)$, $c_2(\theta)$ and $m_1(\theta)$ are continuous on $\theta\in(0,\infty)$.
	\item If $k>0$, then $c_3(\theta)$ and $m_2(\theta)$ are continuous on $\theta\in[\underline{\theta},0]$.
	\item $I_2(q;\theta)$ is jointly continuous on $(q,\theta)\in(0,\infty)\times (\underline{\theta},\infty)$.
	\item $I_1(q;\theta)$ is jointly continuous on $(q,\theta)\in\{(q,\theta)|\theta\in(0,\infty),q\in[\tilde{n}(\theta),\infty)\}$.
\end{enumerate}
\label{cor:cont_fundamental_quantities}
\end{cor}

\begin{proof}

Fix $0<\epsilon<M$ and consider $\theta\in[\epsilon,M]=:\mathfrak{A}$. From part 1 of Lemma \ref{lem:bound}, we know that $(c_1(\theta),c_2(\theta))$ lives on the compact set $\cX:=[\underline{c}_1,\bar{c}_1]\times [\underline{c}_2,\bar{c}_2]$. Hence one can define the map $O:\cX\times \mathfrak{A}\mapsto\mathbb{R}^2$ via
	\begin{align*}
		O((v_1,v_2);\theta):=
		\begin{bmatrix}
				\frac{(v_2-1)^{\alpha}+\theta v_2^{\alpha}+k(1-v_1)^{\alpha}-\theta v_1^{\alpha}}{v_2-v_1}-\alpha[(v_2-1)^{\alpha-1}+\theta   v_2^{\alpha-1}]\\
					\frac{(v_2-1)^{\alpha}+\theta v_2^{\alpha}+k(1-v_1)^{\alpha}-\theta v_1^{\alpha}}{v_2-v_1}-\alpha[k (1-v_1)^{\alpha-1}+\theta   v_1^{\alpha-1}]
		\end{bmatrix}
	\end{align*}
and characterize $(c_1(\theta),c_2(\theta))\in\cX$ as the unique solution to the equation $O((c_1(\theta),c_2(\theta));\theta)=\vec{0}$. Trivially $O$ is continuous on $\cX\times\mathfrak{A}$, and hence the continuity of $(c_1(\theta),c_2(\theta))$ on $\theta\in\mathfrak{A}$ follows from Lemma \ref{lem:ctsdependence}. Since $\epsilon,M$ are arbitrary, we can extend the conclusion to $\theta\in(0,\infty)$. Continuity of $m_1(\theta)$ immediately follows since it is a continuous composition of $c_1(\theta)$ and $c_2(\theta)$.

Part 2 and 3 can be proven similarly. We will show the proof for part 4 which requires a slightly different argument. Recall that $I_1(q;\theta)$ is the unique solution to the equation $g_1(y;q,\theta)=0$ on $y\in(0,\tilde{c}(\theta)\gamma)$, where 
\begin{align*}
	g_1(y;q,\theta):=\alpha[k(\gamma-y)^{\alpha-1}+\theta y^{\alpha-1}]-q.
\end{align*}
Fix $0<\epsilon<M$ and $\bar{q}>0$. Define 
\begin{align*}
	\mathfrak{A}_{\epsilon,M,\bar{q}}:=\{(q,\theta):\theta\in[\epsilon,M],q\in[n_0(\theta),\bar{q}]\}.
\end{align*}
For any $(\theta,q)\in\mathfrak{A}_{\epsilon,M,\bar{q}}$, we have from part 4 of Lemma \ref{lem:bound} that $\underline{I}_1(\bar{q})\leq I_1(q;\theta)\leq \tilde{c}(M)\gamma$, i.e. $I_1(q;\theta)\in[\underline{I}_1(\bar{q}),\tilde{c}(M)\gamma]=:\cX_{\epsilon,M,\bar{q}}$ which is a compact set. We can therefore view $g_1$ as a map $g_1:\cX_{\epsilon,M,\bar{q}}\times\mathfrak{A}_{\epsilon,M,\bar{q}}\mapsto \bR$ and characterize $I_1(q;\theta)\in\cX_{\epsilon,M,\bar{q}}$ as the unique solution to the equation $g_1(I_1(q;\theta);q,\theta)=0$. Then $I_1(q;\theta)$ is continuous on $\mathfrak{A}_{\epsilon,M,\bar{q}}$ by Lemma \ref{lem:ctsdependence}. On letting $\epsilon\downarrow 0$, $M\uparrow \infty$ and $\bar{q}\uparrow \infty$, the continuity of $I_1(q;\theta)$ holds on $(q,\theta)\in\{(q,\theta)|\theta\in(0,\infty),q\in[\tilde{n}(\theta),\infty)\}$.
\end{proof}

\begin{prop}
Recall the definition of $\lambda^*$ in Case 1 and 2 in Lemma \ref{lem:auxsol}. View $\lambda^*=\lambda^*(\theta)$ as a function of $\theta$. Then:
\begin{enumerate}
	\item $\lambda^*(\theta)$ is continuous on $\theta\in(0,\infty)$.
	\item If $\gamma>e^{r\tau}$ and $k>0$, $\lambda^*(\theta)$ is continuous on $\theta\in(\underline{\theta},0]$.
\end{enumerate}
\label{prop:lambda_cont}
\end{prop}

\begin{proof}
Suppose we are in Case 1 with $\theta>0$. For $Y^*(\theta)\in\argmax_{Y\in\cY}\bE[F(Y;\theta)]$, $\lambda^*>0$ is the unique solution to the equation $\varphi(\lambda;\theta)=0$ where
\begin{align*}
	\varphi(\lambda;\theta):=\bE[Z_\tau Y^*(\theta)]-1=\int_0^\infty z[I_1(\lambda z;\theta)\I_{\{\lambda z>m_1(\theta)\}}+I_2(\lambda z;\theta)\I_{\{\lambda z\leq m_1(\theta)\}}]\ell(z)\rd z-1,
\end{align*}
and we recall that $\ell(\cdot)$ denotes the pdf of $Z_\tau$. Fix some arbitrary constants $0<\underline{\lambda}<\bar{\lambda}$ and $0<\epsilon<M$. For $\theta\in[\epsilon,M]$, let $\bar{m}_1:=\max_{\theta\in[\epsilon,M]}m_1(\theta)$ and define $\bar{\lambda}$ as the unique solution to the equation $\bar{\varphi}(\lambda)=0$ with
\begin{align*}
	\bar{\varphi}(\lambda):=\int_0^\infty z[I_1(\lambda z;M)\I_{\{\lambda z>\bar{m}_1\}}+I_2(\lambda z;M)\I_{\{\lambda z\leq \bar{m}_1\}}]\ell(z)\rd z-1.
\end{align*}
Since $\varphi(\lambda;\theta)$ and $\bar{\varphi}(\lambda)$ are both non-increasing in $\lambda$ and $	\varphi(\lambda;\theta)\leq \bar{\varphi}(\lambda)$, we must have $\lambda^*(\theta)\leq \bar{\lambda}$. Following a similar argument, we can deduce $\lambda^*(\theta)\geq \underline{\lambda}$ where $\underline{\lambda}$ is defined as the solution to $\underline{\varphi}(\lambda)$ with
\begin{align*}
	\underline{\varphi}(\lambda):=\int_0^\infty z I_2(\lambda z;M)\I_{\{\lambda z\leq \underline{m}_1\}}\ell(z)\rd z-1
\end{align*}
where $\underline{m}_1:=\min_{\theta\in[\epsilon,M]}m_1(\theta)$. As $\varphi:[\underline{\lambda},\bar{\lambda}]\times [\epsilon,M]\mapsto \bR$ is continuous due to the continuity of $I_1$, $I_2$ and $m_1$ as shown in Corollary \ref{cor:cont_fundamental_quantities}, Lemma \ref{lem:ctsdependence} again implies $\lambda^*(\theta)$ the solution to $\varphi(\lambda,\theta)=0$ is continuous on $[\epsilon,M]$, and the continuity can be extended to $(0,\infty)$ due to the arbitrariness of $\epsilon$ and $M$. Case 2 can be handled similarly where the optimizer $Y^*$ takes a simpler form. 
\end{proof}

\begin{lemm}
On $\theta\in(\underline{\theta},-1)$, we have
\begin{align*}
	c_3(\theta)\gamma\I_{\{q<m_2(\theta)\}}\leq I_2(q;\theta) \I_{\{q<m_2(\theta)\}}\leq \frac{\gamma}{1-|\theta|^{-\frac{1}{1-\alpha}}} \I_{\{q<m_2(\theta)\}}.
\end{align*}
\label{lem:I2_bound}
\end{lemm}
\begin{proof}
Using the monotonicity of $I_2(q;\theta)$, we have $I_2(q;\theta)\I_{\{q<m_2(\theta)\}}\leq I_2(0;q) \I_{\{q<m_2(\theta)\}}=\frac{\gamma}{1-|\theta|^{-\frac{1}{1-\alpha}}} \I_{\{q<m_2(\theta)\}}$ where $I_2(0;\theta)$ is computed explicitly from the equation $\alpha[(y-\gamma)^{\alpha-1}+\theta y^{\alpha-1}]=0$. Similarly, for the lower bound we have
\begin{align*}
	I_2(q;\theta)  \I_{\{q<m_2(\theta)\}}\geq  I_2(m_2(\theta);\theta) \I_{\{q<m_2(\theta)\}}=c_3(\theta)\gamma \I_{\{q<m_2(\theta)\}},
\end{align*}
where the last equality holds because by construction of $m_2$ in \eqref{eq:m2}, $y=c_3\gamma$ is clearly the (unique) solution to the equation in \eqref{eq:I2_eq}.
\end{proof}

\begin{theo}
	$H(\theta)$ is continuous on $\theta\in(\underline{\theta},\infty)$. Moreover, if $\gamma>e^{r\tau}$ and $k>0$, then $\lim_{\theta\downarrow \underline{\theta}}H(\theta)=\underline{H}$ where $\underline{H}$ is defined in \eqref{eq:cri_h}.
	\label{thm:continuity}
\end{theo}

\begin{proof}

We first show that $H(\theta)$ is continuous on $[\epsilon,M]$ for some $0<\epsilon<M$. In the case of $\theta>0$, $H(\theta)$ is a singleton given by
\begin{align}
	H(\theta)=\int_0^\infty [I_1(\lambda^*(\theta) z;\theta)^\alpha\I_{\{\lambda^*(\theta) z>m_1(\theta)\}}+I_2(\lambda^*(\theta) z;\theta)^\alpha\I_{\{\lambda^*(\theta) z\leq m_1(\theta)\}}]\ell(z)\rd z.
\label{eq:integrand_H}
\end{align}
The integrand in \eqref{eq:integrand_H} is continuous in $\theta$ since $m_1(\cdot)$, $\lambda^*(\cdot)$, $I_1(\cdot;\cdot)$ and $I_2(\cdot;\cdot)$ are all continuous by Corollary \ref{cor:cont_fundamental_quantities}. Moreover, \eqref{eq:integrand_H} has an upper bounded of
\begin{align*}
&[I_1(\lambda^*(\theta) z;\theta)^\alpha\I_{\{\lambda^*(\theta) z>m_1(\theta)\}}+I_2(\lambda^*(\theta) z;\theta)^\alpha\I_{\{\lambda^*(\theta) z\leq m_1(\theta)\}}]\ell(z)\\
&\leq
	[I_1(m_1(\theta);\theta)^\alpha+I_2(\underline{\lambda} z;M)^\alpha\I_{\{ z\leq \bar{m}/\underline{\lambda}\}}]\ell(z)\\
	&=(c_1(\theta)\gamma)^\alpha \ell(z) + I_2(\underline{\lambda} z;M)^\alpha\I_{\{ z\leq \bar{m}/\underline{\lambda}\}}\ell(z)\\
	&\leq \gamma^\alpha \ell(z) + I_2(\underline{\lambda} z;M)^\alpha\I_{\{ z\leq \bar{m}/\underline{\lambda}\}}\ell(z).
\end{align*}
Here $\underline{\lambda}$ and $\bar{m}$ are defined as in the proof of Proposition \ref{prop:lambda_cont}. The last expression is integrable using the fact that $I_2(q)\propto q^{-\frac{1}{1-\alpha}}$ for small $q$ (page ec19, e-companion of \cite{TsZh21}). Continuity of $H(\theta)$ on $\theta\in[\epsilon,M]$ now follows from dominated convergence theorem. As $\epsilon$ and $M$ are arbitrary, we conclude $H(\theta)$ is continuous on $(0,\infty)$. By a similar argument, we can deduce $H(\theta)$ is continuous on $(\underline{\theta},0]$.

Finally, we show that $\lim_{\theta\downarrow\underline{\theta}}H(\theta)=\underline{H}$. Using the continuity of $c_3(\theta)$ on $[\underline{\theta},0]$ under the stated assumption of $k>0$, $\lim_{\theta\downarrow\underline{\theta}}c_3(\theta)=c_3(\underline{\theta})$ where $c_3(\underline{\theta})$ is simply the solution to the equation
\begin{align*}
	\frac{(c_3(\underline{\theta})-1)^{\alpha}-(1+k^{\frac{1}{1-\alpha}})^{1-\alpha} c_3(\underline{\theta})^{\alpha}+k}{c_3(\underline{\theta})}=\alpha[(c_3(\underline{\theta})-1)^{\alpha-1}+\theta  c_3(\underline{\theta})^{\alpha-1}].
\end{align*}
This equation admits an explicit solution of $c_3(\underline{\theta})=1+k^{-\frac{1}{1-\alpha}}$. For $\theta\in(\underline{\theta},0)$, we have the budget constraint
\begin{align*}
	1=\bE[Z_\tau Y^*]=\int_0^\infty z I_2(\lambda^*(\theta)z;\theta)\I_{\{z<\frac{m_2(\theta)}{\lambda^*(\theta)}\}}\ell(z)\rd z.
\end{align*}
Then by Fatou's lemma, Lemma \ref{lem:I2_bound} and the fact that $c_3(\underline{\theta})=1+k^{-\frac{1}{1-\alpha}}$,
\begin{align*}
	1&\geq \int_0^\infty \liminf_{\theta\downarrow\underline{\theta}} [ I_2(\lambda^*(\theta)z;\theta)\I_{\{z<\frac{m_2(\theta)}{\lambda^*(\theta)}\}}\ell(z)]\rd z
	\geq  \int_0^\infty z\gamma (1+k^{-\frac{1}{1-\alpha}})\I_{\{z<\liminf_{\theta\downarrow\underline{\theta}}\frac{m_2(\theta)}{\lambda^*(\theta)}\}}\ell(z) \rd z.
\end{align*}
Then as $\eta\mapsto \int_0^\eta \gamma (1+k^{-\frac{1}{1-\alpha}})z\ell(z)\rd z=:\gamma (1+k^{-\frac{1}{1-\alpha}})R(\eta)$ is strictly increasing, we have
\begin{align*}
	\liminf_{\theta\downarrow\underline{\theta}} \frac{m_2(\theta)}{\lambda^*(\theta)}\leq \eta^*
\end{align*}
where $\eta^*$ is defined Proposition \ref{prop:H}. Note that $\eta^*$ satisfies $1=\gamma (1+k^{-\frac{1}{1-\alpha}})R(\eta^*)\iff \eta^*=R^{-1}\left(\frac{1}{\gamma (1+k^{-\frac{1}{1-\alpha}})}\right)$.

Now, one can choose a converging subsequence $\{\theta_k\}_{k\in\bN}$ with $\theta_k\downarrow\underline{\theta}$ such that
\begin{align*}
	\lim_{k\to\infty}\frac{m_2(\theta_k)}{\lambda^*(\theta_k)}=\liminf_{\theta\downarrow\underline{\theta}} \frac{m_2(\theta)}{\lambda^*(\theta)}\leq \eta^*.
\end{align*} 
On the other hand, with Lemma \ref{lem:I2_bound} again we have
\begin{align*}
	1&=\int_0^\infty z I_2(\lambda^*(\theta)z;\theta)\I_{\{z<\frac{m_2(\theta)}{\lambda^*(\theta)}\}}\ell(z)\rd z\leq \int_0^{\frac{m_2(\theta)}{\lambda^*(\theta)}}\frac{\gamma}{1-|\theta|^{-\frac{1}{1-\alpha}}}  z\ell(z)\rd z=\frac{\gamma}{1-|\theta|^{-\frac{1}{1-\alpha}}}R\left(\frac{m_2(\theta)}{\lambda^*(\theta)}\right)
\end{align*}
and hence
\begin{align*}
		\liminf_{\theta\downarrow\underline{\theta}}\frac{m_2(\theta)}{\lambda^*(\theta)}\geq  \liminf_{\theta\downarrow\underline{\theta}} R^{-1}\left( \frac{1-|\theta|^{-\frac{1}{1-\alpha}}}{\gamma}\right)\geq R^{-1}\left( \frac{1-|\underline{\theta}|^{-\frac{1}{1-\alpha}}}{\gamma}\right)=R^{-1}\left(\frac{1}{\gamma (1+k^{-\frac{1}{1-\alpha}})}\right)=\eta^*.
\end{align*}
We therefore conclude
\begin{align*}
		\lim_{k\to\infty}\frac{m_2(\theta_k)}{\lambda^*(\theta_k)}=\liminf_{\theta\downarrow\underline{\theta}} \frac{m_2(\theta)}{\lambda^*(\theta)}= \eta^*.
\end{align*}
Finally, since $H(\theta)$ is monotonic, all converging monotonic subsequences have the same limit such that
\begin{align*}
	\lim_{\theta\downarrow\underline{\theta}}H(\theta)&=\lim_{k\to\infty}H(\theta_k)\\
	&=\lim_{k\to\infty}\int_0^\infty  [I_2(\lambda^*(\theta_k)z;\theta_k)]^\alpha\I_{\{z<\frac{m_2(\theta_k)}{\lambda^*(\theta_k)}\}}\ell(z)\rd z\\	
	&=\int_0^{\eta^*} \gamma^\alpha(1+k^{-\frac{1}{1-\alpha}})^\alpha\ell(z)\rd z=L(\eta^*)\gamma^\alpha(1+k^{-\frac{1}{1-\alpha}})^\alpha=\underline{H}
\end{align*} 
by Lemma \ref{lem:I2_bound} and bounded convergence theorem.
\end{proof}

We now recall \eqref{eq:g_main} that $G:[0,e^{h\tau}]\mapsto 2^{[0,e^{h\tau}]}$ is defined as
\begin{align*}
	G(\xi):=\left\{\bE[Y^\alpha] \: \Bigl | \: Y\in\argmax_{Y\in\cY}\bE\left[F\left(Y;\: \frac{\beta}{1-(1-\beta)e^{-\delta\tau}\xi}\theta^*\left(\frac{\beta}{1-(1-\beta)e^{-\delta\tau}\xi}\right)\right)\right]\right\},
\end{align*}
where the function $\theta^*(\cdot)$ is defined in Proposition \ref{lemm_revise}. Clearly, $H$ and $G$ are linked via $$G(\xi)=H\left(\frac{\beta}{1-(1-\beta)e^{-\delta\tau}\xi}\theta^*\left(\frac{\beta}{1-(1-\beta)e^{-\delta\tau}\xi}\right)\right).$$  

\begin{prop}	
$G$ has the following properties:
	\begin{enumerate}
		\item If $\theta^*(0)>0$, then $G(\xi)$ is a singleton and is a continuous, monotonically increasing and strictly positive function on $\xi\in[0,e^{h\tau}]$.
		\item If $\theta^*(0)=0$, then $G(\xi)=\left\{\mathbb{E}(Y^{\alpha})|Y\in\argmax_{Y\in\mathcal{Y}}\mathbb{E}\left[U(Y-\gamma)\right]\right\}$ which is a strictly positive constant singleton for all $\xi\in[0,e^{h\tau}]$.
		\item If $\theta^*(0)<0$, then $G(\xi)$ is non-increasing.\footnote{Throughout this section, a set-valued map $f(x)$ is said to be non-decreasing (resp. non-increasing) if for any $x_2>x_1$ we have $f_1\leq f_2$ (resp. $f_2\leq f_1$) for all $f_i\in f(x_i)$ with $i\in\{1,2\}$.} Moreover, for
		\begin{align}
			\underline{\xi}:=\frac{1}{(1-\beta)e^{-\delta\tau}}\left[1-\frac{\beta\gamma^{\alpha}e^{-\delta\tau}}{(1+k^{-\frac{1}{1-\alpha}})^{1-\alpha}}\right],
			\label{eq:xi_low}
		\end{align}
		we have:
		\begin{enumerate}
			\item If $\underline{\xi}<0$, then $G(\xi)=\{0\}$ for all $\xi\in[0,e^{h\tau}]$.
			\item If $\underline{\xi}>e^{h\tau}$, then $G(\xi)$ is a strictly positive singleton for all $\xi\in[0,e^{h\tau}]$.
			\item If $\underline{\xi}\in[0,e^{h\tau}]$, then $G(\xi)$ is a singleton for $\xi\neq\underline{\xi}$. $G(\xi)>0$ and is continuous on $\xi\in[0,\underline{\xi})$, $G(\xi)=\{0\}$ on $\xi\in(\underline{\xi},e^{h\tau}]$, and $G(\underline{\xi})=[0,\underline{H}]$ where $\underline{H}$ is defined in \eqref{eq:cri_h}. Moreover, if $\underline{\xi}>0$ then $\lim_{\xi\uparrow\underline{\xi}}G(\xi)=\underline{H}$.
		\end{enumerate}
	\end{enumerate}
	\label{prop:G_properties}
\end{prop}

\begin{proof}
	
	Let
	\begin{align}
		\Theta(\xi):=\frac{\beta }{1-(1-\beta)e^{-\delta\tau}\xi} \theta^*\left(\frac{\beta}{1-(1-\beta)e^{-\delta\tau}\xi}\right).
	\label{eq:Theta_xi}
	\end{align}
	We first claim that, in the case of $\theta^*(0)<0$ (which implies $k>0$ and $\gamma>e^{r\tau}$), we have $\Theta(\xi)\gtreqqless\underline{\theta}$ if and only if $\xi \lesseqqgtr\underline{\xi}$. Suppose there exists $\xi$ such that $\Theta(\xi)=\underline{\theta}<0$. By Lemma \ref{lem:auxsol}, we know $Y^*=0$ is an optimizer to \eqref{eq:axu} when $\theta=\underline{\theta}$. Then
	\begin{align*}
		-k\gamma^\alpha e^{-\delta\tau}=e^{-\delta \tau}F(Y=0;\underline{\theta})=e^{-\delta\tau}\sup_{Y\in\cY}\bE[F(Y;\Theta(\xi))]&=e^{-\delta\tau}\Phi\left(\frac{\beta }{1-(1-\beta)e^{-\delta\tau}\xi} \theta^*\left(\frac{\beta}{1-(1-\beta)e^{-\delta\tau}\xi}\right)\right)\\
		&=\theta^*\left(\frac{\beta}{1-(1-\beta)e^{-\delta\tau}\xi}\right)
	\end{align*}
	by construction that $\theta^*(\kappa)$ is the fixed point of $\theta\mapsto e^{-\delta\tau}\Phi(\kappa\theta)$. But then we must have
	\begin{align*}
		-(1+k^{\frac{1}{1-\alpha}})^{1-\alpha}=	:\underline{\theta}=\Theta(\xi)=\frac{\beta }{1-(1-\beta)e^{-\delta\tau}\xi} \theta^*\left(\frac{\beta}{1-(1-\beta)e^{-\delta\tau}\xi}\right)=-\frac{\beta 	k\gamma^\alpha e^{-\delta\tau}}{1-(1-\beta)e^{-\delta\tau}\xi}
	\end{align*}
	which implies
	\begin{align*}
		\xi=\frac{1}{(1-\beta)e^{-\delta\tau}}\left[1-\frac{\beta\gamma^{\alpha}e^{-\delta\tau}}{(1+k^{-\frac{1}{1-\alpha}})^{1-\alpha}}\right]=:\underline{\xi}.
	\end{align*}
	The claim now follows from Proposition \ref{lemm_revise} that $\theta^*(\cdot)$ is strictly negative and non-increasing when $\theta^*(0)<0$ such that $\Theta(\xi)$ is strictly decreasing.
	
	Since $G(\xi)=H(\Theta(\xi))$, the stated results now easily follow from the properties of $H(\cdot)$ in Proposition \ref{prop:H} and those of $\theta^*(\cdot)$ in Proposition \ref{lemm_revise}. Note that $G(\xi)=\{0\}$ if and only if $\xi$ is such that $\sup_{Y\in\cY}\bE[F(Y;\Theta(\xi))]$ has a unique optimizer being zero almost surely, which in turn occurs if and only if $\Theta(\xi)<\underline{\theta}<0$. Hence $G(\xi)$ must be strictly positive if $\theta^*(0)> 0$ under which $\Theta(\xi)\geq 0$. When $\theta^*(0)<0$, $G(\xi)$ is strictly positive singleton if $\Theta(\xi)>\underline{\theta}\iff \xi<\underline{\xi}$.
	
	As a remark on the special case of $\theta^*(0)=0$, this implies $\Theta(\xi)=0$ due to Proposition \ref{lemm_revise} and hence 
	\begin{align*}
		G(\xi)=H(0)=\left\{\bE[Y^\alpha]\Bigl|Y\in\argmax_{Y\in\cY}\bE[F(Y;0)]\right\}=\left\{\bE[Y^\alpha]\Bigl|Y\in\argmax_{Y\in\cY}\bE[U(Y-\gamma)]\right\},
	\end{align*}
	which is a strictly positive singleton as the optimizer is non-degenerate and unique under $\theta=0$.
\end{proof}

The lemma below is a useful result for establishing the comparative statics with respect to $\beta$.
\begin{lemm}
	Write $G=G(\xi;\beta)$ to stress the dependence of $G$ on $\beta$. For any fixed $\xi\in[0,e^{h\tau}]$, $G(\xi;\beta)$ is non-decreasing (resp. non-increasing) in $\beta$ if $\theta^*(0)>0$ (resp. $\theta^*(0)<0$). 
	\label{lem:G_beta_mono}
\end{lemm}
\begin{proof}
	By Proposition \ref{prop:H}, Lemma \ref{lem:H_increase} and Theorem \ref{thm:continuity}, the set-valued map $H(\theta)$ is non-decreasing. 
	Note that $G(\xi;\beta)=H(\Theta(\xi,\beta))$ where
	\begin{align}
		\Theta(\xi,\beta):=\frac{\beta }{1-(1-\beta)e^{-\delta\tau}\xi} \theta^*\left(\frac{\beta}{1-(1-\beta)e^{-\delta\tau}\xi}\right).
	\label{eq:Theta_xi_beta}
	\end{align}
	The result follows from the fact that $H(\cdot)$ has no dependence on $\beta$, Proposition \ref{lemm_revise} that $\theta^*(\cdot)$ is positive and strictly increasing (resp. negative and weakly decreasing) when $\theta^*(0)>0$ (resp. $\theta^*(0)<0$), and the observation that $\beta\mapsto \frac{\beta }{1-(1-\beta)e^{-\delta\tau}\xi}$ is non-negative and strictly increasing.
\end{proof}

We are now finally ready to prove Theorem \ref{thm:G_fp}. 

\begin{proof}[Proof of Theorem \ref{thm:G_fp}]
	Using the properties derived in Proposition \ref{prop:G_properties}, $G:[0,e^{h\tau}]\to2^{[0,e^{h\tau}]}$ is a closed graph and $G(\xi)$ is non-empty and convex for all $\xi\in[0,e^{h\tau}]$. By Kakutani's fixed-point theorem, at least one fixed point exists on $[0,e^{h\tau}]$.
	
	When $\theta^*(0)=0$, $G(\xi)$ is a constant singleton given by $\left\{\mathbb{E}(Y^{\alpha})|Y\in\argmax_{Y\in\mathcal{Y}}\mathbb{E}\left[U(Y-\gamma)\right]\right\}$ and hence it must be the unique fixed point of $G$.
	
	To show that the fixed point is unique when $\theta^*(0)<0$. Suppose on contrary that $\hat{\xi}_1<\hat{\xi}_2$ are both fixed points. Then by definition we have $\hat{\xi}_i\in G(\hat{\xi}_i)$ for $i\in\{1,2\}$. There are five possible cases: i) $\hat{\xi}_1<\hat{\xi}_2<\underline{\xi}$; ii) $\underline{\xi}<\hat{\xi}_1<\hat{\xi}_2$; iii) $\hat{\xi}_1<\underline{\xi}=\hat{\xi}_2$; iv) $\hat{\xi}_1=\underline{\xi}<\hat{\xi}_2$; and v) $\hat{\xi}_1<\underline{\xi}<\hat{\xi}_2$. Case i) cannot happen because this would imply $G(\hat{\xi}_1)=\hat{\xi}_1<\hat{\xi}_2=G(\hat{\xi}_2)$ which contradicts the fact that $G$ is non-increasing. Case ii) cannot happen because $G(\xi)=\{0\}$ on $\xi>\underline{\xi}$ and then we have $0=G(\hat{\xi}_1)=\hat{\xi}_1<\hat{\xi}_2=G(\hat{\xi}_2)=0$ as a contradiction. Case iii) cannot happen as this would otherwise imply $G(\hat{\xi}_1)=\hat{\xi}_1< \hat{\xi}_2=\underline{\xi}\in G(\underline{\xi})=[0,\underline{H}]$ and in turn $G(\hat{\xi}_1)< \underline{H}=\lim_{\xi\uparrow \underline{\xi}}G(\xi)$ which contradicts the fact that $G(\xi)$ is continuously non-increasing on $[0,\underline{\xi})$. Case iv) also cannot happen because we would then have $0\leq \hat{\xi}_1<\hat{\xi}_2=G(\hat{\xi}_2)=0$ as a contradiction. Case v) is impossible since this will lead to $0<G(\hat{\xi}_1)=\hat{\xi}_1<\hat{\xi}_2=G(\hat{\xi}_2)=0$ which is a contradiction. Thus the fixed point must be unique when $\theta^*(0)<0$.
\end{proof}

\section{Other proofs}
\label{app:cs}

\begin{proof}[Proof of Proposition \ref{prop:unique}]
	Without loss of generality, we just need to prove $\bP(Y^1_1=Y^2_1)=1$. Let $V^{\hat{\pi}^i}$ be the equilibrium value function and $Y^i=Y^i_1$ be the corresponding first-period equilibrium gross return variable induced by $\hat{\pi}^i$. Then $V^{\hat{\pi}^i}(x)=\hat{A}^i x^\alpha$ for some constant $\hat{A}^i$, and then $(\hat{A}^i,Y^i)$ must be a solution to \eqref{eq:sophi_sys} by Proposition \ref{prop:sophis_system}. Consequently, $\xi_i:=\bE[(Y^i)^\alpha]$ must be a fixed point of $\xi\mapsto G(\xi)$. Due to the uniqueness of the fixed point of $G$ under $A_{\my}\leq 0\iff \theta^*(0)\leq 0$ as per Theorem \ref{thm:G_fp}, we have $\xi_1=\xi_2$.
	
	Recall the definition of $\Theta(\xi)$ in \eqref{eq:Theta_xi}. Since $$Y^i\in\argmax_{Y\in \cY}\bE\left[F(Y;\Theta(\xi_i))\right]=\argmax_{Y\in \cY}\bE\left[F(Y;\Theta(\xi_1))\right],$$ by Proposition \ref{prop:uniqueness} we must have $\bP(Y^1=Y^2)=1$ unless $\Theta(\xi_1)=\underline{\theta}$.
	
	Now suppose $A_{\my}=0\iff\theta^*(0)=0$. By Proposition \ref{lemm_revise}, $\theta^*(\kappa)=0$ for any $\kappa\in[0,1]$ and therefore $\Theta(\xi_1)=0>\underline{\theta}$. Hence we must have $\bP(Y^1=Y^2)=1$. If instead $A_{\my}<0\iff\theta^*(0)<0$ and $\Theta(\xi_1)=\underline{\theta}\iff \xi_1=\underline{\xi}$, then Proposition \ref{prop:uniqueness}  suggests $\bP(Y^i\in\{0,\gamma(1+k^{-\frac{1}{1-\alpha}})\})=1$. But then for $i\in\{1,2\}$, the requirement that $\xi_1=\xi_i=\bE[(Y^i)^\alpha]$ uniquely determines the law of $Y^i$ such that
	\begin{align*}
		\bP(Y^i=\gamma(1+k^{-\frac{1}{1-\alpha}}))=\frac{\xi_1}{\gamma^\alpha(1+k^{-\frac{1}{1-\alpha}})^\alpha},\qquad \bP(Y^i=0)=1-\frac{\xi_1}{\gamma^\alpha(1+k^{-\frac{1}{1-\alpha}})^\alpha}.
	\end{align*}
	As a remark, the expressions above are well-defined probability values because $$0\leq \xi_1=\underline{\xi}\leq \underline{H}< \gamma^\alpha(1+k^{-\frac{1}{1-\alpha}})^\alpha,$$ where the second inequality is due to the properties of $G$ in Proposition \ref{prop:G_properties} and the fact that $\xi_1$ is the fixed point of $G$, while the last inequality is due to the definition of $\underline{H}$ in \eqref{eq:cri_h}.
	
	Finally, by the Hardy-Littlewood inequality \eqref{eq:HL}, we must have $Y^i=Q_{Y^i}(1-L(Z_\tau))$ where $Q_{Y^i}$ is the quantile function of $Y^i$ (see again part (ii) of Theorem 1 in \cite{jin-zhou2010}). The conclusion now follows from the fact that $Y^1$ and $Y^2$ have the same probability distribution.
\end{proof}

\begin{proof}[Proof of Proposition \ref{prop:myopia_cs}]

	The monotonicity of $\beta A_{\expo}$ is obvious as $A_{\expo}$ does not depend on $\beta$.
	
	Under the stated assumption, let $\hat{\xi}=\hat{\xi}(\beta)$ be the unique fixed point of $\xi\mapsto G(\xi;\beta)$. Then we have
	\begin{align*}
		A_{\so}=A_{\so}(\beta)=\frac{\beta }{1-(1-\beta)e^{-\delta\tau}\hat{\xi}(\beta)}\theta^*\left(\frac{\beta }{1-(1-\beta)e^{-\delta\tau}\hat{\xi}(\beta)}\right).
	\end{align*}
	Under $A_\my>0\iff\theta^*(0)>0$, it is known from Proposition \ref{lemm_revise} that $\theta^*(\kappa)$ and in turn $\kappa\theta^*(\kappa)$ are positive and non-decreasing. Hence the result follows if we can show that $\hat{\xi}(\beta)$ is non-decreasing in $\beta$. From Proposition \ref{prop:G_properties}, $G(\xi;\beta)$ is strictly positive singleton when $\theta^*(0)>0$. Then if there is a unique fixed point, it must be given by an up-crossing by the diagonal line. Together with Lemma \ref{lem:G_beta_mono} which suggests $G(\xi;\beta)$ is non-decreasing in $\beta$, $\hat{\xi}(\beta)$ must be non-decreasing in $\beta$ as well.
	
	If $A_\my<0\iff\theta^*(0)<0$, then by Lemma \ref{lem:G_beta_mono}, $G(\xi;\beta)$ is non-increasing in $\beta$ which implies $\hat{\xi}(\beta)$ is non-increasing in $\beta$ as well. For $\beta_2>\beta_1$, write $\hat{\Theta}_i:=\Theta(\hat{\xi}(\beta_i),\beta_i)$ for $i\in\{1,2\}$ where $\Theta(\xi,\beta)$ is defined in \eqref{eq:Theta_xi_beta}. Then 
	$$H(\hat{\Theta}_2)=G(\hat{\xi}_2;\beta_2)=\hat{\xi}_2\leq \hat{\xi}_1=G(\hat{\xi}_1;\beta_1)=H(\hat{\Theta}_1).$$ Note that $H(\cdot)$ has no dependence on $\beta$ and is non-decreasing, we must have $\hat{\Theta}_2\leq \hat{\Theta}_1$. The result immediately follows on observing that $A_{\so}(\beta)= \Theta(\hat{\xi}(\beta),\beta)$.
\end{proof}


\end{appendices}


\bibliographystyle{abbrv}
\bibliography{ref}

\begin{thebibliography}{10}

\bibitem{akerlof1991}
G.~A. Akerlof.
\newblock Procrastination and obedience.
\newblock {\em The American Economic Review}, 81(2):1--19, 1991.

\bibitem{ariely-wertenbroch02}
D.~Ariely and K.~Wertenbroch.
\newblock Procrastination, deadlines, and performance: Self-control by
  precommitment.
\newblock {\em Psychological Science}, 13(3):219--224, 2002.

\bibitem{berkelaar-kouwenberg-post04}
A.~B. Berkelaar, R.~Kouwenberg, and T.~Post.
\newblock Optimal portfolio choice under loss aversion.
\newblock {\em Review of Economics and Statistics}, 86(4):973--987, 2004.

\bibitem{bjork-khapko-murgoci17}
T.~Bj{\"o}rk, M.~Khapko, and A.~Murgoci.
\newblock On time-inconsistent stochastic control in continuous time.
\newblock {\em Finance and Stochastics}, 21:331--360, 2017.

\bibitem{BjMu14}
T.~Bj{\"o}rk and A.~Murgoci.
\newblock A theory of markovian time-inconsistent stochastic control in
  discrete time.
\newblock {\em Finance and Stochastics}, 18:545--592, 2014.

\bibitem{calcott-petkov2021}
P.~Calcott and V.~Petkov.
\newblock Excessive consumption and present bias.
\newblock {\em Economic Theory}, 74:113--134, 2022.

\bibitem{chen-xiang-he19}
S.~Chen, S.~Xiang, and H.~He.
\newblock Do time preferences matter in intertemporal consumption and portfolio
  decisions?
\newblock {\em The BE Journal of Theoretical Economics}, 19(2):20170122, 2019.

\bibitem{diamond-koszegi03}
P.~Diamond and B.~K{\"o}szegi.
\newblock Quasi-hyperbolic discounting and retirement.
\newblock {\em Journal of Public Economics}, 87(9-10):1839--1872, 2003.

\bibitem{ekeland-mbodji-pirvu12}
I.~Ekeland, O.~Mbodji, and T.~A. Pirvu.
\newblock Time-consistent portfolio management.
\newblock {\em SIAM Journal on Financial Mathematics}, 3(1):1--32, 2012.

\bibitem{ekeland-pirvu08}
I.~Ekeland and T.~A. Pirvu.
\newblock Investment and consumption without commitment.
\newblock {\em Mathematics and Financial Economics}, 2(1):57--86, 2008.

\bibitem{grenadier-wang07}
S.~R. Grenadier and N.~Wang.
\newblock Investment under uncertainty and time-inconsistent preferences.
\newblock {\em Journal of Financial Economics}, 84(1):2--39, 2007.

\bibitem{hamaguchi21}
Y.~Hamaguchi.
\newblock Time-inconsistent consumption-investment problems in incomplete
  markets under general discount functions.
\newblock {\em SIAM Journal on Control and Optimization}, 59(3):2121--2146,
  2021.

\bibitem{harris-laibson01}
C.~Harris and D.~Laibson.
\newblock Dynamic choices of hyperbolic consumers.
\newblock {\em Econometrica}, 69(4):935--957, 2001.

\bibitem{hu-jin-zhou12}
Y.~Hu, H.~Jin, and X.~Y. Zhou.
\newblock Time-inconsistent stochastic linear--quadratic control.
\newblock {\em SIAM Journal on Control and Optimization}, 50(3):1548--1572,
  2012.

\bibitem{hu-jin-zhou17}
Y.~Hu, H.~Jin, and X.~Y. Zhou.
\newblock Time-inconsistent stochastic linear-quadratic control:
  Characterization and uniqueness of equilibrium.
\newblock {\em SIAM Journal on Control and Optimization}, 55(2):1261--1279,
  2017.

\bibitem{jaskiewicz-nowak21}
A.~Ja{\'s}kiewicz and A.~S. Nowak.
\newblock Markov decision processes with quasi-hyperbolic discounting.
\newblock {\em Finance and Stochastics}, 25(2):189--229, 2021.

\bibitem{jin-zhou2010}
H.~Jin and X.~Y. Zhou.
\newblock Erratum to ``behavioral portfolio selection in continuous time''.
\newblock {\em Mathematical Finance}, 20(3):521--525, 2010.

\bibitem{KaSh16}
I.~Karatzas and S.~E. Shreve.
\newblock {\em Methods of mathematical finance}, volume Corrected 4th printing.
\newblock Springer, 2016.

\bibitem{kirby97}
K.~N. Kirby.
\newblock Bidding on the future: Evidence against normative discounting of
  delayed rewards.
\newblock {\em Journal of Experimental Psychology: General}, 126(1):54, 1997.

\bibitem{laibson97}
D.~Laibson.
\newblock Golden eggs and hyperbolic discounting.
\newblock {\em The Quarterly Journal of Economics}, 112(2):443--478, 1997.

\bibitem{laibson98}
D.~Laibson.
\newblock Life-cycle consumption and hyperbolic discount functions.
\newblock {\em European Economic Review}, 42(3-5):861--871, 1998.

\bibitem{laibson-repetto-tobacman98}
D.~I. Laibson, A.~Repetto, and J.~Tobacman.
\newblock Self-control and saving for retirement.
\newblock {\em Brookings Papers on Economic Activity}, 1998(1):91--196, 1998.

\bibitem{liu16}
B.~Liu, L.~Lu, C.~Mu, and J.~Yang.
\newblock Time-inconsistent preferences, investment and asset pricing.
\newblock {\em Economics Letters}, 148:48--52, 2016.

\bibitem{loewenstein-prelec93}
G.~F. Loewenstein and D.~Prelec.
\newblock Preferences for sequences of outcomes.
\newblock {\em Psychological Review}, 100(1):91, 1993.

\bibitem{love-phelan15}
D.~Love and G.~Phelan.
\newblock Hyperbolic discounting and life-cycle portfolio choice.
\newblock {\em Journal of Pension Economics \& Finance}, 14(4):492--524, 2015.

\bibitem{mcclure-etal04}
S.~M. McClure, D.~I. Laibson, G.~Loewenstein, and J.~D. Cohen.
\newblock Separate neural systems value immediate and delayed monetary rewards.
\newblock {\em Science}, 306(5695):503--507, 2004.

\bibitem{meier-sprenger10}
S.~Meier and C.~Sprenger.
\newblock Present-biased preferences and credit card borrowing.
\newblock {\em American Economic Journal: Applied Economics}, 2(1):193--210,
  2010.

\bibitem{Merton:69}
R.~C. Merton.
\newblock Lifetime portfolio selection under uncertainty: the continuous-time
  case.
\newblock {\em The Review of Economics and Statistics}, 51(3):247--257, 1969.

\bibitem{Merton:71}
R.~C. Merton.
\newblock Optimum consumption and portfolio rules in a continuous-time model.
\newblock {\em Journal of Economic Theory}, 3(4):373--413, 1971.

\bibitem{oDonoghue-rabin06}
T.~O'Donoghue and M.~Rabin.
\newblock Optimal sin taxes.
\newblock {\em Journal of Public Economics}, 90(10-11):1825--1849, 2006.

\bibitem{oDonoghue-rabin15}
T.~O'Donoghue and M.~Rabin.
\newblock Present bias: Lessons learned and to be learned.
\newblock {\em American Economic Review}, 105(5):273--279, 2015.

\bibitem{oster-scott05}
S.~M. Oster and F.~M. Scott~Morton.
\newblock Behavioral biases meet the market: The case of magazine subscription
  prices.
\newblock {\em The BE Journal of Economic Analysis \& Policy}, 5(1), 2005.

\bibitem{palacios11}
I.~Palacios-Huerta and A.~P{\'e}rez-Kakabadse.
\newblock Consumption and portfolio rules with stochastic quasi-hyperbolic
  discounting.
\newblock {\em Available at SSRN 1465110}, 2011.

\bibitem{phelps-pollak68}
E.~S. Phelps and R.~A. Pollak.
\newblock On second-best national saving and game-equilibrium growth.
\newblock {\em The Review of Economic Studies}, 35(2):185--199, 1968.

\bibitem{shigeta22}
Y.~Shigeta.
\newblock Quasi-hyperbolic discounting under recursive utility and
  consumption--investment decisions.
\newblock {\em Journal of Economic Theory}, 204:105518, 2022.

\bibitem{shin-roh19}
Y.~H. Shin and K.~H. Roh.
\newblock An optimal consumption and investment problem with stochastic
  hyperbolic discounting.
\newblock {\em Advances in Difference Equations}, 2019:1--7, 2019.

\bibitem{strotz1956}
R.~H. Strotz.
\newblock Myopia and inconsistency in dynamic utility maximization.
\newblock {\em The Review of Economic Studies}, 23(3):165--180, 1956.

\bibitem{thaler81}
R.~Thaler.
\newblock Some empirical evidence on dynamic inconsistency.
\newblock {\em Economics Letters}, 8(3):201--207, 1981.

\bibitem{tian16}
Y.~Tian.
\newblock Optimal capital structure and investment decisions under
  time-inconsistent preferences.
\newblock {\em Journal of Economic Dynamics and Control}, 65:83--104, 2016.

\bibitem{TsZh21}
A.~S.~L. Tse and H.~Zheng.
\newblock Portfolio selection, periodic evaluations and risk taking.
\newblock {\em Operations Research}, 71(6):2078--2091, 2023.

\bibitem{tversky-kahneman1992}
A.~Tversky and D.~Kahneman.
\newblock Advances in prospect theory: Cumulative representation of
  uncertainty.
\newblock {\em Journal of Risk and Uncertainty}, 5:297--323, 1992.

\bibitem{yong12}
J.~Yong.
\newblock Time-inconsistent optimal control problems and the equilibrium hjb
  equation.
\newblock {\em Mathematical Control \& Related Fields}, 2(3):271--329, 2012.

\bibitem{zou-chen-wedge14}
Z.~Zou, S.~Chen, and L.~Wedge.
\newblock Finite horizon consumption and portfolio decisions with stochastic
  hyperbolic discounting.
\newblock {\em Journal of Mathematical Economics}, 52:70--80, 2014.

\end{thebibliography}

\end{document}